\documentclass[onecolumn]{autart}    

\PassOptionsToPackage{svgnames}{xcolor}

\usepackage{titlesec}

\usepackage[numbers,sort&compress]{natbib}
\usepackage{graphicx}
\graphicspath{ {Images/} }

\usepackage{amsmath}

\usepackage{float}           
\usepackage{caption}
\usepackage{subcaption}
\usepackage{setspace}
\usepackage{tocloft}

\usepackage{bigints}
\usepackage{amsfonts}
\usepackage{amssymb}
\usepackage{amsbsy}
\usepackage{mathtools, nccmath, textcomp}
\usepackage{multirow}

\usepackage{enumitem}
\usepackage{booktabs}

\usepackage{tabularx}
\usepackage{xcolor}
\usepackage{tikz}
\usetikzlibrary{positioning}
\usetikzlibrary{shapes.geometric,arrows.meta,chains,positioning,shapes.symbols}
\usetikzlibrary{shadings,shadows}

\usepackage{fixmath}
\usepackage{bm}
\usepackage{longtable}
\usepackage{soul}
\usepackage{lmodern}
\usepackage[T1]{fontenc}
\usepackage{multicol}

\usepackage{tcolorbox}
\tcbuselibrary{skins,breakable}

\usepackage{bbm}

\usepackage{hyperref}
\hypersetup{colorlinks=true,
	linkcolor=blue,
	citecolor=blue,
	urlcolor=blue,
}

\font\f=cmbsy7 at2.5pt
\def\ointp{\mathop
   {\oint \kern-5.55pt\raise1.8pt\hbox{%
    \rlap{\f\char"5E}\kern.1pt\rlap{\f\char"5E}} \hspace{1mm}  }\limits
    }
    
\font\f=cmbsy7 at2.5pt
\def\ointn{\mathop
   {\oint \kern-5.45pt\raise1.8pt\hbox{%
    \rlap{\f\char"5F}\kern.1pt\rlap{\f\char"5F}} \hspace{1mm}  }\limits
    }

    {\endtcolorbox}

    {\endtcolorbox}

    {\endtcolorbox}

\makeatletter
\newlength\xvec@height%
\newlength\xvec@depth%
\newlength\xvec@width%
\newcommand{\xvec}[2][]{%
  \ifmmode%
    \settoheight{\xvec@height}{$#2$}%
    \settodepth{\xvec@depth}{$#2$}%
    \settowidth{\xvec@width}{$#2$}%
  \else%
    \settoheight{\xvec@height}{#2}%
    \settodepth{\xvec@depth}{#2}%
    \settowidth{\xvec@width}{#2}%
  \fi%
  \def\xvec@arg{#1}%
  \def\xvec@dd{:}%
  \def\xvec@d{.}%
  \raisebox{.2ex}{\raisebox{\xvec@height}{\rlap{%
    \kern.05em
    \begin{tikzpicture}[scale=1]
    \pgfsetroundcap
    \draw (.05em,0)--(\xvec@width-.05em,0);
    \draw (\xvec@width-.05em,0)--(\xvec@width-.15em, .075em);
    \draw (\xvec@width-.05em,0)--(\xvec@width-.15em,-.075em);
    \ifx\xvec@arg\xvec@d%
      \fill(\xvec@width*.45,.5ex) circle (.5pt);%
    \else\ifx\xvec@arg\xvec@dd%
      \fill(\xvec@width*.30,.5ex) circle (.5pt);%
      \fill(\xvec@width*.65,.5ex) circle (.5pt);%
    \fi\fi%
    \end{tikzpicture}%
  }}}%
  #2%
}
\makeatother

\def\onedot{$\mathsurround0pt\ldotp$}
\def\cdddot#1{
  \mathbin{\vcenter{\baselineskip.67ex
    \hbox{\onedot}\hbox{\onedot}\hbox{\onedot}%
  }}%
}

\makeatletter
\def\underbracex#1#2{\mathop{\vtop{\m@th\ialign{##\crcr
   $\hfil\displaystyle{#2}\hfil$\crcr
   \noalign{\kern3\p@\nointerlineskip}%
   #1\crcr\noalign{\kern3\p@}}}}\limits}

\def\upbracefilla{$\m@th \setbox\z@\hbox{$\braceld$}%
  \bracelu\leaders\vrule \@height\ht\z@ \@depth\z@\hfill 
\kern\p@\vrule \@width\p@\kern\p@\vrule \@width\p@\kern\p@\vrule \@width\p@
$}

\def\upbracefillb{$\m@th \setbox\z@\hbox{$\braceld$}%
\vrule \@width\p@\kern\p@\vrule \@width\p@\kern\p@\vrule \@width\p@\kern\p@
 \leaders\vrule \@height\ht\z@ \@depth\z@\hfill\bracerd
  \braceld\leaders\vrule \@height\ht\z@ \@depth\z@\hfill
\kern\p@\vrule \@width\p@\kern\p@\vrule \@width\p@\kern\p@\vrule \@width\p@
$}

\def\upbracefillc{$\m@th \setbox\z@\hbox{$\braceld$}%
\vrule \@width\p@\kern\p@\vrule \@width\p@\kern\p@\vrule \@width\p@\kern\p@
\leaders\vrule \@height\ht\z@ \@depth\z@\hfill
\kern\p@\vrule \@width\p@\kern\p@\vrule \@width\p@\kern\p@\vrule \@width\p@
$}

\def\upbracefilld{$\m@th \setbox\z@\hbox{$\braceld$}%
\vrule \@width\p@\kern\p@\vrule \@width\p@\kern\p@\vrule \@width\p@\kern\p@
 \leaders\vrule \@height\ht\z@ \@depth\z@\hfill\braceru$}

\def\upbracefillbd{$\m@th \setbox\z@\hbox{$\braceld$}%
\vrule \@width\p@\kern\p@\vrule \@width\p@\kern\p@\vrule \@width\p@\kern\p@
\bracerd\braceld
 \leaders\vrule \@height\ht\z@ \@depth\z@\hfill\braceru$}

\makeatother

\renewcommand{\vec}[1]{\xvec[]{#1}}

\newcommand{\HP}{{\scalebox{.65}{\stx{\!H\!P}}}}
\newcommand{\HW}{{\scalebox{.65}{\stx{\!H\!W}}}}

\newcommand*{\B}[1]{\ifmmode\bm{#1}\else\textbf{#1}\fi}

\newcommand{\ub}[2]{\underbrace{#1}_{#2}}

\newcommand{\p}{\partial}
\newcommand{\var}{\delta}
\newcommand{\dint}{\displaystyle\int}

\newcommand{\stx}[1]{\mbox{\scriptsize$#1$}}

\newcommand{\R}[1]{\mathbb{R}^{#1}}

\newcommand{\F}{\mathcal{F}}

\newcommand{\X}{\mathtt{X}}

\newcommand{\SX}{\texttt{S}}

\newcommand{\q}{\epsilon}

\theoremstyle{definition} 

\theoremstyle{definition} 

\theoremstyle{definition} 
\newtheorem{definition}{\normalfont\textbf{Definition}}
\theoremstyle{definition} 
\newtheorem{theorem}{\normalfont\textbf{Theorem}}
\theoremstyle{definition} 
\newtheorem{lemma}{\normalfont\textbf{Lemma}}
\theoremstyle{definition} 
\newtheorem{proposition}{\normalfont\textbf{Proposition}}
\theoremstyle{definition} 
\newtheorem{corollary}{\normalfont\textbf{Corollary}}
\theoremstyle{definition} 
\newtheorem{remark}{\normalfont\textbf{Remark}}

\newenvironment{proof}[1][Proof]{\textbf{#1.} }{\ \rule{0.5em}{0.5em}}

\newcommand{\uu}{\textbf{u}}
\newcommand{\XX}{\textbf{X}}
\newcommand{\SXX}{\textbf{S}}
\newcommand{\xx}{\textbf{x}}
\newcommand{\ttt}{\textbf{t}}
\newcommand{\bb}{\textbf{b}}
\newcommand{\FF}{\underline{\textbf{F}}}
\newcommand{\SSS}{\underline{\textbf{T}}}
\newcommand{\EEE}{\underline{\textbf{E}}}

\begin{document}

\begin{frontmatter}

\title{Strong imposition of Dirichlet boundary velocities in structure-preserving discretizations of elastodynamics}
\thanks[footnoteinfo]{Corresponding author: C.Ponce. E-mail: cristobal.ponces@usm.cl}
\author[First]{Cristobal Ponce}$^{,\star}$, 
\author[First]{Hector Ramirez}, 
\author[Second]{Yongxin Wu}, 
\author[Second]{Ning Liu}, 
\author[Second]{Yann Le Gorrec} 

\address[First]{Universidad Tecnica Federico Santa Maria, Valparaiso, Chile. \\
Department of Mechanical Engineering (e-mail: cristobal.ponces@usm.cl) \\
Department of Electronic Engineering (e-mail: hector.ramireze@usm.cl)}  

\address[Second]{Universite Marie et Louis Pasteur, SUPMICROTECH, CNRS, Institut FEMTO-ST, Besancon, France. \\  (e-mails: yongxin.wu@femto-st.fr, legorrec@femto-st.fr, ning.liu@femto-st.fr)}            

\begin{abstract}                          
The imposition of boundary velocities in finite element models of port-Hamiltonian elastodynamics typically relies on Lagrange multipliers, yielding Differential-Algebraic Equations (DAEs). Alternatively, weak imposition methods that maintain an Ordinary Differential Equation (ODE) structure often exhibit poor accuracy at Dirichlet boundaries. To address these limitations, this paper introduces an additive kinematic decomposition at the continuous level, splitting the displacement and velocity fields into a relative dynamic component that vanishes on the boundary and a prescribed lifting function extending into the interior domain. This decomposition induces a distributed port that maps the effects of the boundary actuation inside the domain. By incorporating this mapping into suitable virtual power principles, we derive \textit{lifted} port-Hamiltonian system (PHS) models that, upon finite element discretization, reduce to ODE {systems in which} Dirichlet boundary velocities are strongly imposed. The framework is applied to derive 2-field and 4-field formulations suited to distinct PHS geometric representations. Furthermore, we show that under specific shape functions, standard FEM schemes are recovered, demonstrating that the lifting framework in the discrete models is equivalent to the classic algebraic matrix partitioning in computational mechanics practice.  The energy-balance properties and computational performance of the proposed methodology are verified through numerical simulations.
\end{abstract}


\end{frontmatter}

\endNoHyper

\section{Introduction}\label{sec:Introduction}

The simulation of elastodynamic phenomena is used in engineering applications such as wave propagation, vibroacoustics, and flexible robotics \cite{liu2026elastodynamic,banerjee2013large,omisore2020review}. Classical spatial discretization techniques based on the finite element method (FEM) rely on Lagrangian frameworks and second-order differential equations to model these continuous systems \cite{bathe2006finite,belytschko2014nonlinear}. These standard formulations are typically derived from virtual work principles where the system state is defined by work-conjugate variables; consequently, boundary conditions are enforced as prescribed boundary displacements and tractions \cite{reddy2017energy}. Alternatively, first-order Hamiltonian representations can be obtained by applying Legendre transformations to the Lagrangian system \cite{simo1988hamiltonian,safko2002classical,yao2009symplectic}. While classical Hamiltonian mechanics describes closed systems and still relies on boundary displacements to satisfy variational principles \cite{sanchez2021symplectic}, opening them to evaluate the energy exchange with the environment or other physical systems requires analyzing the power flow through the boundaries \cite{van2002hamiltonian}. This motivates a shift towards describing boundary interactions using power-conjugate variables, i.e., velocities and tractions.

Hamiltonian systems have been generalized to open physical systems through the addition of power-conjugate ports to enable the exchange of energy with other systems and the environment. This generalization has been established in both finite- and infinite-dimensional systems and defines the port-Hamiltonian system (PHS) framework \cite{maschke1993port,van2002hamiltonian}. PHS models explicitly separate the interconnection topology from the constitutive relations and the dissipative phenomena. This structural separation guarantees clear power balances and passivity. These physical properties make the PHS framework suitable for multiphysics modeling, energy-consistent simulation, and energy-based control \cite{duindam2009modeling}. The PHS formalism has been widely applied to formulate electrical \cite{bartel2022port,abdolmaleki2022distributed,gernandt2024port}, mechanical \cite{macchelli2009port,ferguson2024port,brugnoli2021portA,latussek2026port}, fluid \cite{cisneros2020port,mora2021port,califano2021geometric,cardoso2024port}, thermodynamic \cite{eberard2004port,jaschke2022two,ramirez2022overview}, and multiphysics systems \cite{cardoso2015modeling,rizzello2017thermodynamically,brugnoli2021port,cisneros2025dynamic,rashad2025port}, among others. In the domain of elastodynamics, infinite-dimensional PHS models have been defined on different geometric structures. These formulations were originally established on Stokes-Dirac structures \cite{van2002hamiltonian} and subsequently extended to jet-bundle structures \cite{nishida2005formal,schoberl2014jet,schoberl2015port}. Furthermore, the framework has been adapted for systems with algebraic constraints, leading to descriptor PHS formulations governed by differential-algebraic equations (DAEs), commonly referred to as PH-DAE systems \cite{van2013port,beattie2018linear}. Recently, PHS and PH-DAE have been defined on Stokes-Lagrange structures, which combine Stokes-Dirac structures for their conservation properties with Lagrangian subspaces to enable the definition of systems with implicitly defined energy \cite{van2018generalized,maschke2020linear,maschke2023linear}. These structures facilitate the formulation of generalized PHS models capable of accommodating differential and nonlocal constitutive relations \cite{bendimerad2026structure}, as well as diverse algebraic constraints defined in both the interconnection and energy structures \cite{van2020dirac,bendimerad2025stokes}. For a comprehensive overview of the PHS framework, including modeling, spatial discretization techniques (such as finite differences, spectral methods, finite elements), model order reduction, and control design, the reader is referred to \cite{van2014port,rashad2020twenty,cardoso2024port}.

To use infinite-dimensional models for simulation and energy-based control, finite-dimensional approximations that preserve the underlying geometric structure are preferred. Structure-preserving finite element methods have been developed to preserve the PHS structure and maintain passivity at the discrete level, addressing linear and nonlinear phenomena \cite{brugnoli2020partitioned, thoma2022explicit,thoma2024velocity, kinon2024generalized, ponce2024structure}. Despite these developments, simultaneously preserving the PHS structure at the semi-discrete level while ensuring accurate behavior at Dirichlet boundaries, remains a challenge. In classical Galerkin finite element methods, non-homogeneous Dirichlet boundary conditions are addressed at the continuous level through lifting operators \cite[Ch. 3.3.3]{quarteroni2009numerical}. This approach decomposes the field into an unknown relative component that satisfies homogeneous conditions, and a known lifting function distributed over the domain that satisfies the original non-homogeneous conditions. Supported by Sobolev trace theorems \cite[Ch. 2.4.3]{quarteroni2009numerical}, this decomposition extend the prescribed boundary values into the interior domain. This homogenizes the Dirichlet boundary conditions, allowing the variational problem to be solved over a subspace of test functions that vanish at the boundary. At the discrete level, this translates into the standard matrix partitioning schemes used in computational mechanics. However, to the best of the authors' knowledge, such continuous lifting has been avoided in structure-preserving methods for PHS. A plausible explanation for this is that PHS treat boundaries not as spatial constraints to be homogenized, but as active ports essential for energy exchange. In addition, other challenges are present. First, imposing a lifting function forces the state variables into an affine space $\mathcal{V}_D$, which breaks the linear topology required to preserve geometric structures, since $\mathcal{V}_D$ is not closed under linear combinations, as discussed in \cite{quarteroni2009numerical}. Second, from a physical point of view, evaluating the Hamiltonian over a relative state variable shifts the energy definition away from the physical value. Lastly, the known lifting function acts a distributed port, destroying the physical power balance at the boundaries.

Consequently, to preserve these geometric and physical properties without resorting to continuous lifting, the usual practice in structure-preserving FEM for PHS has shifted towards alternative imposition strategies. Discretizing the jet-bundle PHS representation via Hamilton's principle while imposing Dirichlet boundary velocities using Lagrange multipliers enforces kinematic compliance, but modifies the standard ODE topology, yielding a PH-DAE system \cite[Prop. III.1]{ponce2024portThesis}. The simulation and control of systems governed by DAEs require specific numerical and mathematical tools, which are often computationally more demanding than their ODE counterparts. To avoid DAE structures, penalty methods are frequently employed for the weak imposition of Dirichlet boundary conditions. However, this approach produces a finite-dimensional model driven by imposed boundary displacements, which violates the power conjugation and introduces numerical degeneration in the discrete elastic energy due to the penalty factor \cite[Prop. III.2]{ponce2024portThesis}. Alternative variational frameworks targeting the Stokes-Dirac structure have been proposed to bypass penalty parameters while preserving the geometric properties. These approaches use the stress field to weakly impose Dirichlet velocities and include methods based on the Hellinger-Reissner principle \cite{thoma2022explicit, thoma2024velocity}, modified linked Lagrange multiplier method \cite{ponce2024structure}, and the generalized Hamilton's principle \cite{ponce2025port}. Despite preserving the ODE and the PHS structure, the weak imposition in mixed formulations frequently exhibit poor accuracy at the Dirichlet boundaries, as discussed in \cite{brugnoli2022explicit}. To overcome this lack of accuracy in mixed formulations, Hu-Washizu-like methods \cite{kinon2023port,hille2026port,kinon2026mixed} impose Dirichlet boundary velocities by expanding the state space using Lagrange multipliers. This guarantees exact boundary behavior but again, this transforms the model into a PH-DAE system. To the best of the authors' knowledge, a structure-preserving finite element framework that enables the strong imposition of Dirichlet boundary velocities yielding an ODE system topology, has not yet been proposed in the literature.

To bridge this gap, the contribution of this manuscript is the development of a continuous-to-discrete lifting framework that strongly imposes Dirichlet boundary velocities in port-Hamiltonian elastodynamics while yielding a finite-dimensional model with an ODE topology. In a first step, by embedding the continuous kinematic decomposition of displacement and velocity fields into proposed rate-form variational principles, we yield continuous-level models, denoted as \textit{lifted} PHS. In this formulation, the Hamiltonian does not degenerate because the displacement lifting is included as a state variable, allowing its evaluation in the absolute fields. Furthermore, the original power balance is preserved since the output conjugate to the distributed port evaluates to zero, ensuring power neutrality. In a second step, the continuous lifting framework is used to derive structure-preserving finite element discretizations. The methodology is applied to a 2-field jet-bundle formulation inspired by the Hamilton-Pontryagin principle, and a 4-field Stokes-Dirac formulation inspired by the Hu-Washizu principle. The resulting discretizations yield finite-dimensional PHS models that strongly impose Dirichlet velocities while avoiding differential-algebraic structures to attain an ODE topology. Furthermore, we show that under certain compatibility conditions, these formulation are mathematically equivalent to applying classical finite element discretizations followed by an algebraic partitioning of the consistent mass matrix.

The article is structured as follows. Section \ref{sec:2} outlines the theoretical foundations of elastodynamics, port-Hamiltonian elastodynamics, and the lifting approach. Section \ref{sec:Continuos} introduces the continuous lifted PHS framework. Section \ref{sec:FEM} presents the structure-preserving finite element approaches with strong imposition of Dirichlet boundary velocities. Section \ref{sec:Simulations} shows numerical simulations, and Section \ref{sec:Conclusion} provides conclusions and discusses future work.


\section{Background} \label{sec:2}

To establish the theoretical baseline, this section first outlines the governing equations of nonlinear elastodynamics. Next, the fundamental properties of PHS are introduced through a finite-dimensional mechanical example. These concepts are then extended to the continuous level to define infinite-dimensional PHS representations for geometrically nonlinear and hyperelastic systems defined over multidimensional spatial domains $\Omega \subset \mathbb{R}^{\ell}$, where $\ell \in \lbrace 1,2,3 \rbrace$. The foundational formulations presented up to this point build upon previous work detailed in \cite{ponce2025port}. Lastly, the concept of kinematic lifting is introduced to provide the mathematical foundation for the strong imposition of Dirichlet boundary velocities developed in subsequent sections. For notational clarity, explicit spatial and temporal dependencies are omitted where appropriate.

\subsection{Nonlinear elastodynamics}\label{ssec:elastodynamics}

Let $\mathcal{B}_0 \subset \mathbb{R}^{3}$ be the volume of an elastic body in the reference configuration, with $\p\mathcal{B}_0^D$ and $\p\mathcal{B}_0^N$ denoting its disjoint Dirichlet and Neumann boundary surfaces. The motion of a hyperelastic solid is described by the displacement field $\uu(\XX,t) \in \mathbb{R}^{3}$, defined as:
\begin{equation}
	\uu(\XX,t)= \xx(\XX,t) -\XX,
\end{equation}
which assigns to each material point $\XX \in \mathcal{B}_0$ a displacement vector specifying its position $\xx(\XX,t)$ at time $t$ in the deformed configuration. The deformation gradient tensor $\FF \in \mathbb{R}^{3\times 3}$, which characterizes local deformations, is given by:
\begin{equation}
	\FF = \frac{\p \xx}{\p \XX} = \mathbf{I} + \nabla \uu,
\end{equation}
where $\mathbf{I} \in \mathbb{R}^{3 \times 3}$ is the second-order identity tensor, and $\nabla(\cdot)$ represents the gradient operator with respect to the material coordinates $\XX$. The Green-Lagrange strain tensor $\EEE \in \mathbb{R}^{3\times 3}$, which measures strain in the reference configuration, is defined as:\\[-2mm]
\begin{equation}
	\EEE = \frac{1}{2} \left( \FF^\top \FF - \mathbf{I} \right) = \frac{1}{2} \left( \nabla \uu + (\nabla \uu)^\top + (\nabla \uu)^\top \nabla \uu \right).
	\label{eq:def_Green_Lagrange}
\end{equation}

The second Piola-Kirchhoff stress tensor $\SSS \in \mathbb{R}^{3 \times 3}$, which represents stress power-conjugated to the strain rate $\dot{\EEE}$, is related to the strain energy density function $W(\EEE) \in \mathbb{R}$ by the hyperelastic constitutive relation:  
\begin{equation}
	\SSS = \frac{\p W}{\p \EEE},
\end{equation}
where $W(\EEE)$ represents the deformation energy per unit reference volume. Therefore, the elastic energy $\mathcal{U} \in \mathbb{R}$ is defined as:
\begin{equation}
	\mathcal{U} = \int_{\mathcal{B}_0} W(\EEE)\, d\XX.
	\label{eq:def_U_elastic}
\end{equation}
The governing equations of motion for total Lagrangian nonlinear elastodynamics, are formulated as:
\begin{align}
	\begin{array}{r} 
		\mbox{for all } \XX \in \mathcal{B}_0:\\
	\end{array} & \left\lbrace \begin{array}{rl}
		\rho_0 \ddot{\uu} = & \mbox{Div}(\FF \SSS) + \bb, \\
		\EEE = & \EEE(\uu), \\
		\SSS =  &\frac{\p W}{\p \EEE}, 
	\end{array} \right. \label{eq:GHP_elastodynamics}\\[1mm]
	\mbox{for all } \SXX \in \p\mathcal{B}_0^D: & \hspace{5.5mm} \uu_D = \uu, \label{eq:GHP_DirBC}\\[0mm]
	\mbox{for all } \SXX \in \p\mathcal{B}_0^N: & \hspace{5.5mm} \ttt_N = \FF \SSS \,\hat{\textbf{N}},  \label{eq:GHP_NeuBC}
\end{align}
where $\bb$ is the volumetric body force, $\rho_0$ is the material density, and $\hat{\textbf{N}} \in \mathbb{R}^{3}$ is the outward unit vector normal to the reference boundary $\p\mathcal{B}_0$. The terms $\uu_D$ and $\ttt_N$ denote the prescribed displacement and traction enforced on the Dirichlet and Neumann boundaries, respectively.

\subsection{Finite-dimensional PHS for linear mechanical systems}

To introduce the fundamental structure and properties of PHS, the formulation of a finite-dimensional linear mechanical system provides an illustrative example. Let $q(t) \in \R{n}$ represent the vector of generalized coordinates, $M = M^\top > 0 \in \R{n \times n}$ the mass matrix, $K = K^\top > 0 \in \R{n \times n}$ the stiffness matrix, $C = C^\top \geq 0 \in \R{n \times n}$ the damping matrix, and $F(t) = B u(t) \in \R{n} $ the applied loads, where $B \in \R{n\times m}$ is an input map matrix that distributes the arbitrary input signals $u(t) \in \R{m}$ to the corresponding degrees of freedom. This results in the standard second-order dynamic model:
$$
M\ddot{q}(t) + C\dot{q}(t) + Kq(t) = B u(t).
$$

By defining the generalized momentum vector as $p(t) = M\dot{q}(t)$ and assembling the state vector $x(t) = [p(t)^\top\; q(t)^\top]^\top$, this second-order differential equation can be rewritten as a PHS as follows:
\begin{align*}
	\underbrace{\begin{bmatrix} \dot{p}(t) \\ \dot{q}(t) \end{bmatrix}}_{\dot{x}} = & \, \left( \begin{matrix} \\ \\
	\end{matrix} \right. \underbrace{\begin{bmatrix} 0 & -I \\ I & 0 \end{bmatrix}}_{J} - \underbrace{\begin{bmatrix} C & 0 \\ 0 & 0 \end{bmatrix}}_{R} \left. \begin{matrix} \\ \\
	\end{matrix} \right) \underbrace{\begin{bmatrix} M^{-1} p(t) \\ K q(t) \end{bmatrix}}_{\nabla_x H(x)} + \underbrace{\begin{bmatrix} B \\ 0 \end{bmatrix}}_{G} u(t), \\[1mm]
	y(t) = & \; G^\top \nabla_x H(x).
\end{align*} 
Here, $J = -J^\top$ is the skew-symmetric interconnection matrix, $R = R^\top \geq 0$ is the dissipation matrix, and $G$ is the input map. The term $e_x = \nabla_x H(x)$ represents the gradient of the Hamiltonian function $H(x)$, where $H(x)$ represents the total stored energy of the system, given by:
$$
H(x) = \frac{1}{2} p(t)^\top M^{-1} p(t) + \frac{1}{2} q(t)^\top K q(t). 
$$
Within the PHS terminology, the components of the state vector $x(t)$ are defined as energy variables, while their time derivatives $\dot{x}(t)$ constitute the flows. The elements of the gradient vector $e_x = \nabla_x H(x)$ are defined as the efforts, also referred to as co-energy variables. The input and output variables $u$ and $y$ are defined as ports, which are the abstract points from which the system exchanges energy with other systems or the environment. By taking the time derivative of the Hamiltonian, the PHS structure reveals the  power exchange as:
$$
\dot{H} = \nabla_x H(x)^\top \dot{x} = y^\top u - \nabla_x H(x)^\top R \,\nabla_x H(x).
$$
The resultant balance equation highlights the passivity of PHS. This property dictates that energy is not generated internally. The total energy is  conserved in the absence of dissipative effects and external inputs. Furthermore, the energy of the system increases exclusively upon active interactions, specifically when the product $y^\top u > 0$.

\subsection{Infinite-dimensional PHS for nonlinear elastodynamics}

In the context of nonlinear elastodynamics, the fundamental differential operator arises from the kinematic equations. In the geometrically nonlinear setting, displacements and strains are related through a first-order differential operator devoid of cross derivatives, as defined in \eqref{eq:def_Green_Lagrange}. Consequently, the infinite-dimensional PHS formulations are structured upon a specific class of operators.

\begin{definition}[Differential operators \cite{ponce2024portBEAM}] \label{def:operadores_F} Let $\X = [X_1,\dots,X_\ell]^\top$ be a vector of orthogonal material coordinates, $\Omega \subset \mathbb{R}^{\ell}$ an open set, and let $v(\X) \in \mathbb{R}^{m} $ and $w(\X) \in \mathbb{R}^{n} $ be two sufficiently smooth vector functions. The first-order differential operator $\F_\X$ and its formal adjoint $\F_\X^*$ are given by:
	\begin{align}
		\F_\X \,w(\X)= &\, F_0(\X)\, w(\X) + \textstyle\sum_{k=1}^\ell  F_k(\X) \, \p_k  w(\X), \notag
		\\[3mm]
		\F_\X^* v(\X) =& \,  F_0(\X)^\top  v(\X) - \textstyle\sum_{k=1}^\ell  \p_k (F_k(\X)^\top \!v(\X)), \notag
	\end{align}
	with $\p_k = \p / \p X_k$, and matrices $F_0(\X),\,F_k(\X) \in \mathbb{R}^{m \times n}$.
\end{definition}

\begin{lemma}[Integration by parts \cite{ponce2024portBEAM}] \label{lemma:integration}
	For any pair of smooth functions $v(\X) \in \mathbb{R}^{m}$ and $w(\X) \in \mathbb{R}^{n}$ defined on the closure $\bar{\Omega} = \Omega \cup \p\Omega$, the following integral identity holds:
	\begin{equation}
		\displaystyle\int_\Omega  \left[ v(\X)^{\top} \F_\X\, w(\X) -  w(\X)^{\top} \F_\X^* \,v(\X) \right] d\X = \displaystyle\int_{\p \Omega} w(\SX)^{\top} {F}_\p(\SX) \, v(\SX) \,  d\SX,
		\notag
	\end{equation}
	where $ F_\p(\SX) \in \mathbb{R}^{n\times m}$ is a boundary-valued matrix induced by $\F_\X$ and is given by:
	\begin{equation}
		F_\p(\SX) =   \textstyle\sum_{k=1}^\ell  F_k(\SX)^\top\, \hat{n}_k(\SX),
		\notag
	\end{equation}
	with $\hat{n}_k(\SX)$ denoting the $k$-th component of the pointing outward unit vector $\hat{n}(\SX)\in \mathbb{R}^{\ell}$ normal to the boundary $\p\Omega$.
\end{lemma}

Before defining the infinite-dimensional PHS representations, we introduce the fundamental physical quantities and energy density functions that constitute the models. For the scope of this work, the elastodynamic system is assumed to be conservative, therefore, dissipative phenomena are not considered.

\textit{Kinematics:} The system configuration is described by the generalized displacement field $r(\X,t) \in \mathbb{R}^n$, with its temporal derivative defining the velocity vector $\dot{r}(\X,t) \in \mathbb{R}^n$. The internal deformation is captured by the generalized Voigt strain vector $\q(\X,t) \in \mathbb{R}^m$, which contains the independent components of the Green-Lagrange strain tensor in Voigt notation. The strain and strain-rate are kinematically related to the displacement and velocity, respectively, through $\q(\X,t) = \mathcal{L}_\X(r)r$ and $\dot{\q}(\X,t) = \F_\X(r)\dot{r}$. Here, $\mathcal{L}_\X(r)(\cdot)$ and $\F_\X(r)(\cdot)$ are nonlinear first-order differential operators modulated by $r(\X,t)$, that belong to the class in Definition \ref{def:operadores_F}.

\textit{Mechanical energy:} The kinetic energy is characterized by the mass density matrix $\mathcal{M}(\X) = \mathcal{M}(\X) ^\top >0  \in \mathbb{R}^{n \times n}$, which defines the generalized momentum density $p(\X,t) = \mathcal{M}(\X)\dot{r}(\X,t) \in \mathbb{R}^n$. The elastic energy is governed by the generalized strain energy density function $\Psi(\q) \in \mathbb{R}$, which yields the generalized Voigt stress vector $e_\q(\X,t) = \partial \Psi(\q)/\partial \q \in \mathbb{R}^m$, representing the components of the second Piola-Kirchhoff stress tensor in Voigt notation. Finally, the gravitational potential energy is characterized by the generalized body force vector $b(\X) \in \mathbb{R}^n$.

The first PHS representation is stated in Definition \ref{def:jet_bundle_phs} and is denoted as a jet-bundle PHS, sometimes referred to as port-Lagrangian system \cite{nishida2006field,nishida2005formal,schoberl2014jet,schoberl2015port}. A key characteristic of this model is that the interconnection operator is purely algebraic, i.e., $J = -J^\top$, while all spatial differential operators are encapsulated within the Hamiltonian through the elastic energy density $\Psi(\q)$ and the kinematic relationship $\q(r) = \mathcal{L}_\X(r)r$.

\begin{definition}[Jet-bundle formulation] \label{def:jet_bundle_phs}
	Let $z(\X,t) = [p(\X,t)^\top \, r(\X,t)^\top]^\top \in \mathbb{R}^{2n}$ denote the energy variables, and let $\var_z H(z) = [e_p(p)^\top \, e_r(r)^\top]^\top \in \mathbb{R}^{2n}$ denote the corresponding co-energy variables. The infinite-dimensional PHS governing nonlinear elastodynamics is given by:
	\begin{equation}
		\underbrace{\begin{bmatrix}
				\dot{p}(\X,t) \\ \dot{r}(\X,t)
		\end{bmatrix}}_{\dot{z}(\X,t)} = 
		\underbrace{\begin{bmatrix}
				0 & -I_n \\ I_n & 0 
		\end{bmatrix}}_{J=-J^\top}
		\underbrace{\begin{bmatrix}
				e_p(p) \\ e_r(r)
		\end{bmatrix}}_{\var_z H(z)},
		\label{eq:def_PLS_inf}
	\end{equation}
	$ $\\[-7mm]
	\begin{equation}
		H(z) = \int_{\Omega} \left[ \frac{1}{2} p(\X,t)^\top \mathcal{M}(\X)^{-1} p(\X,t) + \Psi(\q(r)) - r(\X,t)^\top b(\X) \right] d\X,
		\label{eq:HAMILTONIAN_PLS_inf}
	\end{equation}
	
	where $e_p(p) = \mathcal{M}(\X)^{-1}p(\X,t) = \dot{r}(\X,t)$ is the velocity field, $e_r(r) = \F_\X(r)^*e_\q(\q(r)) - b(\X)$ represents the sum of  internal and gravitational loads, and $I_n \in \R{n \times n}$ is the identity matrix. In this formulation, the generalized stress vector $e_\q(\q)$ is dependent on $r(\X,t)$ through the strong imposition of the kinematic equation $\q(r) = \mathcal{L}_\X(r)r$.
\end{definition}

In the jet-bundle PHS in Definition \ref{def:jet_bundle_phs}, the upper block of \eqref{eq:def_PLS_inf}, given by $\dot{p} = - \F_\X(r)^*e_\q(\q(r)) + b$, recovers the local balance of linear momentum. Here, $\dot{p}$ represents the inertial forces, $-\F_\X(r)^*e_\q(\q(r))$ embodies the internal elastic forces, and $b$ accounts for the body forces due to gravity. The lower block, $\dot{r} = e_p(p) = \mathcal{M}^{-1}p$, simply enforces the kinematic identity mapping momentum to velocity.

\begin{remark}
	Note that the term $-\F_\X(r)^*e_\q$ is equivalent to the Divergence of the first Piola-Kirchhoff stress tensor $\underline{\mathbf{P}} = \underline{\mathbf{F}} \underline{\mathbf{T}} $. Therefore, the operators $-\F_\X(r)^*(\cdot)$ and $\text{Div}(\underline{\mathbf{F}}\,\cdot)$ are mathematically analogous.
\end{remark}

While the jet-bundle PHS model embeds the spatial differentiation within the Hamiltonian functional, an alternative geometric representation can be derived by transferring these differential operators directly into the interconnection structure. To achieve this, the generalized strain $\q(\X,t)$ is considered as an energy variable. This choice yields an infinite-dimensional PHS defined on a Stokes-Dirac structure, characterized by a nonlinear and state-modulated formal skew-adjoint differential operator $\mathcal{J}(x) = -\mathcal{J}(x)^*$, and a Hamiltonian functional that is algebraic with respect to its energy variables.

\begin{definition}[Stokes-Dirac formulation] \label{def:PHS_stokesDirac}
	Let  $x(\X,t) = [p(\X,t)^\top \; \q(\X,t)^\top \; r(\X,t)^\top]^\top \in \R{2n+m}$ denote the energy variables, and let $\var_x H(x) = [e_p(p)^\top \; e_\q(\q)^\top -b(\X)^\top]^\top \in \R{2n+m}$ denote the corresponding co-energy variables. The infinite-dimensional PHS governing nonlinear elastodynamics is given by:
	\begin{equation}
			\ub{\begin{bmatrix}
					\dot{p}(\X,t) \\ \dot{\q}(\X,t) \\ \dot{r}(\X,t) \end{bmatrix} }{\dot{x}(\X,t)} 
			= \ub{\begin{bmatrix}
					0 & -\F_\X(r)^* &  -I_n  \\ 
					\F_\X(r) & 0 & 0  \\
					I_n & 0 & 0 
			\end{bmatrix}}{\mathcal{J}(x)\,=-\mathcal{J}(x)^*}  \hspace{-1.0mm}
			\ub{\begin{bmatrix}
					e_p(p) \\ e_\q(\q) \\ -b(\X)
			\end{bmatrix}}{\var_x H(x)} , 
		\label{eq:theo2:PHS_explicito}
	\end{equation}
	$ $\\[-7mm]
	\begin{equation}
		H(x) =  \dint_\Omega \left[ \frac{1}{2}\, p(\X,t)^\top \mathcal{M}(\X)^{-1}p(\X,t) + \Psi(\q) - r(\X,t)^\top b(\X) \right] d\X,
	\end{equation}
	
	where the generalized stress vector $e_\q(\q)$ is dependent on $\q(\X,t)$ through the strong imposition of the hyperelastic constitutive equation $e_\q(\q) = \p\Psi(\q)/\p \q$.
\end{definition}

To complete the infinite-dimensional PHS formulations in Definitions \ref{def:jet_bundle_phs} and \ref{def:PHS_stokesDirac}, the energy exchange across the boundary and the corresponding boundary input and output ports must be defined.

Assume that the boundary is partitioned as $\p \Omega = \p\Omega_N \cup \p\Omega_D$, with $\p\Omega_N \cap \p\Omega_D = \emptyset$, where $\p\Omega_N$ and $\p\Omega_D$ denote the disjoint subsets where Neumann and Dirichlet boundary conditions are applied, respectively. While Neumann and Dirichlet boundary conditions are classically imposed as prescribed tractions and displacements in Lagrangian settings, respectively, the PHS framework requires boundary ports to be power-conjugated. 

\begin{definition}[Boundary ports] \label{def:boundaryPorts}
Based on the integral identity from Lemma \ref{lemma:integration}, the power exchange derived from the time derivative of the Hamiltonian is given by:
\begin{equation*}
	\dot{H} = \int_{\p \Omega} \!\!\! e_p(p)^\top  F_\p(r) \, e_\q(\q) \, d\SX = 
	\int_{\p \Omega_N} \!\!\!\!\!\! v_N(\SX,t)^\top  \tau_N(\SX,t) \, d\SX +
	\int_{\p \Omega_D} \!\!\!\!\!\! \tau_D(\SX,t)^\top  v_D(\SX,t) \, d\SX = \int_{\p \Omega}\!\!\! y_\p(\SX,t)^\top  u_\p(\SX,t) \, d\SX.
\end{equation*}
The boundary inputs $u_\p(\SX,t)\in \mathbb{R}^{2n}$ and outputs $y_\p(\SX,t)\in \mathbb{R}^{2n}$ are selected as:
\begin{equation}
u_\p(\SX,t) = [\tau_N(\SX,t)^\top \, v_D(\SX,t)^\top]^\top \quad , \quad y_\p(\SX,t) = [v_N(\SX,t)^\top \, \tau_D(\SX,t)^\top]^\top,
\end{equation}
where $\tau_D(\SX,t), \tau_N(\SX,t) \in \mathbb{R}^n$ represent the generalized boundary tractions, and $v_D(\SX,t), v_N(\SX,t) \in \mathbb{R}^n$ denote the generalized boundary velocities, defined respectively as:
\begin{align}
	\tau_N(\SX,t) &=  F_\p(r) e_\q(\q), \qquad  v_N(\SX,t) = e_p(p) \qquad\qquad (\text{on } \p\Omega_N), \\[1mm]
	\tau_D(\SX,t) &=  F_\p(r) e_\q(\q), \qquad  v_D(\SX,t) = e_p(p) \qquad\qquad (\text{on } \p\Omega_D).
\end{align}
Here, $F_\p(r) \in \R{n \times m}$ is the boundary-valued matrix induced by the differential operator $\F_\X(r)(\cdot)$. Consequently, $F_\p(r)$ is modulated by $r(\SX,t)$. 
\end{definition}

\begin{remark}
	Note that the generalized boundary traction $\tau_N(\SX,t) = F_\p(r) e_\q(\q)$ is equivalent to the classical boundary traction vector $\textbf{t}_N  =\underline{\mathbf{FT}}\hat{\textbf{N}}$. Since $e_\q(\q)$ represents the second Piola-Kirchhoff stress tensor $\underline{\mathbf{T}}$, the nonlinear boundary matrix $F_\p(r)$ is mathematically analogous to the geometric projection operator defined by $\underline{\mathbf{F}}(\cdot)\hat{\textbf{N}}$.
\end{remark}

\subsection{Kinematic lifting}

To handle non-homogeneous Dirichlet boundary conditions, we employ an additive decomposition of relevant fields. In the context of classical elastodynamics, this approach partitions the generalized displacement field into a known, quasi-static, sufficiently smooth lifting component $r_L(\X,t)$ defined on the closure $\bar{\Omega} = \Omega \cup \p\Omega$, and an unknown relative dynamic component $r_r(\X,t)$ defined on $\Omega$. This allows the total displacement field to be expressed as:
\begin{equation}
	r(\X,t) = r_r(\X,t) + r_L(\X,t). \label{eq:displacement_decomposition}
\end{equation}
Here, the lifting function $r_L(\X,t)$ is further partitioned. Specifically, one part matches the prescribed boundary displacement on the Dirichlet boundary $\p\Omega_D$, while the other acts as an arbitrary field distributed inside the domain $\Omega$. This allows for the transformation of the original non-homogeneous Dirichlet boundary-value problem into a homogeneous one, driven instead by distributed-domain known fields. To deal with prescribed boundary velocities, we extend this decomposition to the velocity field. 

\begin{definition}[Kinematic lifting] \label{def:kinematic_lifting}
	Let $H^1(\Omega; \mathbb{R}^n)$ denote the Sobolev space of square-integrable vector fields with square-integrable first weak spatial derivatives. The displacement field $r(\X,t)$ and the  velocity field $\dot{r}(\X,t)$ are additively decomposed into their unknown relative components and the prescribed, quasi-static, sufficiently smooth, known lifting contributions as: 
	\begin{align}
		r(\X,t) &= r_r(\X,t) + r_L(\X,t), \label{eq:total_r} \\
		\dot{r}(\X,t) &= \dot{r}_r(\X,t) + v_L(\X,t). \label{eq:total_rdot}
	\end{align}
	The relative dynamic fields $r_r(\X,t)$ and $\dot{r}_r(\X,t)$ are constrained to the linear space of kinematically admissible homogeneous fields, defined as:
	\begin{equation}
		\mathcal{V}_0 = \left\{ w \in H^1(\Omega; \mathbb{R}^n) \; \big| \; w|_{\p\Omega_D} = 0 \right\}.
	\end{equation}
	
	The known lifting fields $r_L(\X,t)$ and $v_L(\X,t)$ are defined on $\bar{\Omega}$ and satisfy the following spatial restrictions:
	\begin{itemize}
		\item On the Dirichlet boundary $\p\Omega_D$, they impose the non-homogeneous prescribed boundary conditions:
		\begin{align}
			r_L|_{\p\Omega_D} &= r_D(\SX,t), \\
			v_L|_{\p\Omega_D} &= \dot{r}_D(\SX,t) = v_D(\SX,t),
		\end{align}
		where $r_D(\SX,t)$ and $v_D(\SX,t)$ denote the prescribed boundary displacement and velocity, respectively.\\
		
		\item Inside the domain $\Omega$, they correspond to arbitrary interior extension fields denoted as:
		\begin{align}
			r_L|_\Omega &= r_\Omega(\X,t), \\
			v_L|_\Omega &= \dot{r}_\Omega(\X,t) = v_\Omega(\X,t).
		\end{align}
		Consequently, the total fields $r(\X,t)$ and $\dot{r}(\X,t)$ satisfy the non-homogeneous Dirichlet boundary conditions and are constrained to the affine spaces $\mathcal{V}_D^r$ and $\mathcal{V}_D^v$, respectively, defined as:
		\begin{align}
			\mathcal{V}_D^r &= \left\{ w \in H^1(\Omega; \mathbb{R}^n) \; \big| \; w|_{\p\Omega_D} = r_D(\SX,t) \right\} \\
			\mathcal{V}_D^v &= \left\{ w \in H^1(\Omega; \mathbb{R}^n) \; \big| \; w|_{\p\Omega_D} = v_D(\SX,t) \right\}.
		\end{align}
	\end{itemize}
\end{definition}

From a finite element perspective, the construction of the lifting fields $r_L(\X,t)$ and $v_L(\X,t)$ offers flexibility. A convenient choice is to define these fields such that they decay to zero within the elements adjacent to the Dirichlet boundary $\p\Omega_D$. 
Consequently, the interior extensions $r_\Omega(\X,t)$ and $v_\Omega(\X,t)$ vanish, and the relative dynamic fields $r_r(\X,t)$ and $\dot{r}_r(\X,t)$ match the physical fields at all interior nodes of the mesh, while the nodes at the Dirichlet boundaries are prescribed. This reduces the number of operations and simplifies the interpretation of numerical results, as the relative variables represent the absolute nodal displacements and velocities of the system. 
%



\section{Lifted port-Hamiltonian systems} \label{sec:Continuos}

This section establishes the mathematical foundations of the kinematic lifting framework for port-Hamiltonian elastodynamics. The continuous models are developed in two distinct parts, each addressing a specific geometric representation of the system. First, the lifted jet-bundle PHS formulation is presented, where the spatial differential operator is embedded within the energy functional. Second, the lifted Stokes-Dirac PHS formulation is derived, where the differential operator is instead  embedded within the interconnection structure. For each framework, specific virtual power-based variational schemes are proposed to consistently define and derive the corresponding continuous lifted PHS.

To derive the continuous lifted models, two virtual power-based variational principles are introduced. These principles are reminiscent to the Hamilton-Pontryagin principle \cite{yoshimura2006dirac} and the Hu-Washizu principle \cite[Ch. 8.5]{belytschko2014nonlinear}, but specifically formulated in terms of virtual power to accommodate the first-order rate topology characteristic of PHS. The first variational principle is developed to recover the jet-bundle structure, while the second principle is formulated to recover the Stokes-Dirac geometric structure. Before stating these variational principles, physically meaningful Lagrangian densities for elastodynamics are defined.

\begin{definition} (Lagrangian densities) \label{def:energies}
	Let $\mathfrak{T}_p(p)$ denote the kinetic energy density, and let $\mathfrak{U}_{rr}(r)$ and $\mathfrak{U}_{\q r}(\q,r)$ denote the total potential energy densities expressed in terms of the generalized displacement and strain fields, respectively:
	\begin{equation}
		\mathfrak{T}_p(p) = \frac{1}{2} p^\top \mathcal{M}^{-1} p, \qquad \mathfrak{U}_{rr}(r) = \Psi(\q(r)) - r^\top b, \qquad \mathfrak{U}_{\q r}(\q,r) = \Psi(\q) - r^\top b.
	\end{equation}
	From these scalar functionals, two distinct Lagrangian densities are constructed depending on the chosen set of state variables:
	\begin{equation}
		\mathfrak{L}_{pr}(p,r) = \mathfrak{T}_p(p) - \mathfrak{U}_{rr}(r), \qquad \mathfrak{L}_{p\q r}(p,\q,r) = \mathfrak{T}_p(p) - \mathfrak{U}_{\q r}(\q,r).
	\end{equation}
\end{definition}

With these definitions, the Lagrangian $\mathfrak{L}_{pr}$ naturally pairs with the jet-bundle PHS formulation, whereas $\mathfrak{L}_{p\q r}$ pairs with the Stokes-Dirac PHS representation. In order to derive the lifted jet-bundle PHS representations, a virtual power approach inspired by the Hamilton-Pontryagin principle is proposed. 

\begin{proposition}[Hamilton-Pontryagin-based virtual power principle] \label{prop:varPPLE1} The stationary condition of the functional: 
	\begin{equation} \label{eq:PHP_standard}
		\mathcal{P}_{\HP}(\dot{p},\dot{r}) = \int_\Omega \left[ \dot{r} \cdot \dot{p} - \dot{\mathfrak{L}}_{pr}(p,r,\dot{p},\dot{r}) \right] d\X - \int_{\p\Omega_N} \!\!\! \dot{r} \cdot \tau_N \, d\SX,
	\end{equation}
	for all independent variations $(\delta\dot{p}, \delta\dot{r})$ subject to the homogeneous Dirichlet boundary constraint:
	\begin{equation}
		\delta \dot{r}(\SX,t) = 0 \quad \forall \SX \in \p\Omega_D,
	\end{equation}
	defines a physically consistent variational principle for port-Hamiltonian elastodynamics in jet-bundles.
\end{proposition}
\begin{proof}
	Demonstrating that the proposed scheme is variationally consistent for jet-bundle port-Hamiltonian elastodynamics is equivalent to proving that the stationarity condition $\delta \mathcal{P}_{\HP} = 0$ with respect to the independent variations $(\delta\dot{p}, \delta\dot{r})$ reproduces the system's governing dynamics. The details are given in Appendix \ref{app:proof_Prop1}.
\end{proof}

\begin{proposition}[Lifted jet-bundle PHS] \label{prop:lifted_jet_bundlePHS}
The kinematic lifting in Definition \ref{def:kinematic_lifting}, featuring a prescribed velocity field $v_L(\X,t)$ that satisfies the Dirichlet boundary condition $v_L|_{\p\Omega_D} = v_D(\SX,t)$, induces, via the variational principle in Proposition \ref{prop:varPPLE1}, the following lifted PHS:
	\begin{align}
		\underbrace{\begin{bmatrix}
				\dot{p}(\X,t) \\ \dot{r}_r(\X,t) \\ \dot{r}_L(\X,t)
		\end{bmatrix}}_{\dot{z}(\X,t)} = & \;
		\underbrace{\begin{bmatrix}
				0 & -I_n & \,0 \\ I_n & 0  & \,0 \\ 0 & 0 & \,0
		\end{bmatrix}}_{J=-J^\top}
		\underbrace{\begin{bmatrix}
				e_p(p) \\ e_{r}(r_r \!+\! r_L) \\ e_{r}(r_r \!+\! r_L)
		\end{bmatrix}}_{\var_z H(z)} + \underbrace{\begin{bmatrix}
		0 \\ -I_n \\ I_n
	\end{bmatrix}}_{\mathcal{G}} \underbrace{\begin{bmatrix} v_L(\X,t)\end{bmatrix} }_{u_\Omega(\X,t)}  \label{eq:lifted_PHS_continuo} \\[1mm]
	y_\Omega(\X,t) = & \; \mathcal{G}^\top \delta_z {H}(z) = 0, \notag 
	\end{align}
	$ $\\[-7mm]
	\begin{equation}
		H(z) = \int_{\Omega} \left[ \frac{1}{2} p(\X,t)^\top \mathcal{M}(\X)^{-1} p(\X,t) + \Psi(\q(r_r + r_L)) - [r_r(\X,t) + r_L(\X,t)]^\top b(\X) \right] d\X,
		\label{eq:HAMILTONIAN_PHS1_Lifted}
	\end{equation}
	\begin{equation}
		\dot{H} = \int_{\p\Omega} y_\p(\SX,t)^\top u_\p(\SX,t) \, d\SX,
		\label{eq:HAMILTONIANDOT_PHS1_Lifted}
	\end{equation}
	
	with $u_\Omega(\X,t)\in \mathbb{R}^{n}$ and $y_\Omega(\X,t)\in \mathbb{R}^{n}$ being the distributed input and output ports, respectively. The boundary inputs $u_\p(\SX,t)\in \mathbb{R}^{2n}$ and outputs $y_\p(\SX,t)\in \mathbb{R}^{2n}$ are defined as:
	\begin{equation}
		u_\p(\SX,t) = [\tau_N(\SX,t)^\top \; v_D(\SX,t)^\top]^\top \quad , \quad y_\p(\SX,t) = [v_N(\SX,t)^\top \; \tau_D(\SX,t)^\top]^\top, \notag
	\end{equation}
	where $\tau_D(\SX,t), \tau_N(\SX,t) \in \mathbb{R}^n$ represent the generalized boundary tractions, and $v_D(\SX,t), v_N(\SX,t) \in \mathbb{R}^n$ denote the generalized boundary velocities, defined as:
	\begin{align}
		\tau_N(\SX,t) &=  F_\p(r_r\!+\!r_L) e_\q(\q(r_r\!+\!r_L)), \qquad  v_N(\SX,t) = e_p(p) \qquad\qquad   (\text{on } \p\Omega_N), \\[1mm]
		\tau_D(\SX,t) &=  F_\p(r_r\!+\!r_L) e_\q(\q(r_r\!+\!r_L)), \qquad  v_D(\SX,t) =  e_p(p)   \qquad\qquad   (\text{on } \p\Omega_D),
	\end{align}
    where $v_D(\SX,t) = e_p(p)$ on $\p\Omega_D$ follows from $e_p(p)|_{\p\Omega_D} = (\dot{r}_r + v_L)|_{\p\Omega_D} $ with  $\dot{r}_r|_{\p\Omega_D} = 0$ and $v_L|_{\p\Omega_D} = v_D(\SX,t)$.
\end{proposition}

\begin{proof}
The proof of this proposition is provided in Appendix \ref{app:proof_Prop2}.
\end{proof}

Moving us to derive the Stokes-Dirac representation, a virtual power principle inspired by the Hu-Washizu formulation is proposed to treat $\dot{p}$, $\dot{\q}$, $\dot{r}$, and $e_\q$ as independent variables.

\begin{proposition}[Hu-Washizu-based virtual power principle] \label{prop:varPPLE2} The stationary condition of the functional: 
	\begin{equation} \label{eq:PHP_standard}
		\mathcal{P}_{\HW}(\dot{p}, \dot{\q}, \dot{r}, e_\q) = \int_\Omega \left[ \dot{r} \cdot \dot{p} - \dot{\mathfrak{L}}_{p\q r}(p,\q,r, \dot{p},\dot{\q}, \dot{r}) -  e_\q \cdot (\dot{\q} -  \mathcal{F}_\X(r)\dot{r} ) \right] d\X - \int_{\p\Omega_N} \!\!\! \dot{r} \cdot \tau_N \, d\SX,
	\end{equation}
	for all independent variations $(\delta\dot{p}, \delta\dot{\q}, \delta\dot{r}, \delta{e_\q})$ subject to the homogeneous Dirichlet boundary constraint:
	\begin{equation}
		\delta \dot{r}(\SX,t) = 0 \quad \forall \SX \in \p\Omega_D,
	\end{equation}
	defines a physically consistent variational principle for port-Hamiltonian elastodynamics in Stokes-Dirac structure.
\end{proposition}

\begin{proof}
	Demonstrating that the proposed scheme is variationally consistent for the Stokes-Dirac port-Hamiltonian representation is equivalent to proving that the stationarity condition $\delta \mathcal{P}_{\HW} = 0$ with respect to the independent variations $(\delta\dot{p}, \delta\dot{r}, \delta\dot{\q}, \delta e_\q)$  reproduces the system's governing dynamics. The details are given in Appendix \ref{app:proof_Prop3}.
\end{proof}

\begin{proposition}[Lifted Stokes-Dirac PHS] \label{prop:lifted_StokesPHS2}
	The kinematic lifting in Definition \ref{def:kinematic_lifting}, featuring a prescribed velocity field $v_L(\X,t)$ that satisfies the Dirichlet boundary condition $v_L|_{\p\Omega_D} = v_D(\SX,t)$, induces, via the variational principle in Proposition \ref{prop:varPPLE2}, the following lifted PHS:
	\begin{align}
		\underbrace{\begin{bmatrix}
				\dot{p}(\X,t) \\ \dot{\q}(\X,t) \\ \dot{r}_r(\X,t) \\ \dot{r}_L(\X,t)
		\end{bmatrix}}_{\dot{x}(\X,t)} = & \;
		\underbrace{\begin{bmatrix}
				0 & -\F_\X(r_r \!+\! r_L)^* & -I_n & \;\;0 \\ \F_\X(r_r \!+\! r_L) & 0 & 0  & \;\;0\\ I_n & 0 & 0  & \;\;0 \\ 0 & 0 & 0 & \;\;0
		\end{bmatrix}}_{\mathcal{J}(x)=-\mathcal{J}(x)^*}
		\underbrace{\begin{bmatrix}
				e_p(p) \\ e_\q(\q) \\ -b(\X) \\ -b(\X)
		\end{bmatrix}}_{\var_x H(x)} + \underbrace{\begin{bmatrix}
				0 \\ 0 \\  -I_n \\ I_n
		\end{bmatrix}}_{\mathcal{G}(x)} \underbrace{\begin{bmatrix} v_L(\X,t)\end{bmatrix} }_{u_\Omega(\X,t)}  \label{eq:lifted_PHS_continuo2} \\[1mm]
		y_\Omega(\X,t) = & \; \mathcal{G}(x)^\top \delta_x {H}(x) = 0, \notag
	\end{align}
	$ $\\[-7mm]
	\begin{equation}
		H(x) = \int_{\Omega} \left[ \frac{1}{2} p(\X,t)^\top \mathcal{M}(\X)^{-1} p(\X,t) + \Psi(\q) - [r_r(\X,t) + r_L(\X,t)]^\top b(\X) \right] d\X,
		\label{eq:HAMILTONIAN_PHS2_Lifted}
	\end{equation}
	\begin{equation}
		\dot{H} = \int_{\p\Omega} y_\p(\SX,t)^\top u_\p(\SX,t) \, d\SX,
		\label{eq:HAMILTONIANDOT_PHS2_Lifted}
	\end{equation}
	
	with $u_\Omega(\X,t)\in \mathbb{R}^{n}$ and $y_\Omega(\X,t)\in \mathbb{R}^{n}$ being the distributed input and output ports, respectively. The boundary inputs $u_\p(\SX,t)\in \mathbb{R}^{2n}$ and outputs $y_\p(\SX,t)\in \mathbb{R}^{2n}$ are defined as:
	\begin{equation}
		u_\p(\SX,t) = [\tau_N(\SX,t)^\top \; v_D(\SX,t)^\top]^\top \quad , \quad y_\p(\SX,t) = [v_N(\SX,t)^\top \; \tau_D(\SX,t)^\top]^\top, \notag
	\end{equation}
	where $\tau_D(\SX,t), \tau_N(\SX,t) \in \mathbb{R}^n$ represent the generalized boundary tractions, and $v_D(\SX,t), v_N(\SX,t) \in \mathbb{R}^n$ denote the generalized boundary velocities, defined as:
	\begin{align}
		\tau_N(\SX,t) &=  F_\p(r_r \!+\! r_L) e_\q(\q), \qquad  v_N(\SX,t) = e_p(p) \qquad \qquad  (\text{on } \p\Omega_N), \\[1mm]
		\tau_D(\SX,t) &=  F_\p(r_r \!+\! r_L) e_\q(\q), \qquad  v_D(\SX,t) = e_p(p)   \qquad\qquad   (\text{on } \p\Omega_D),
	\end{align}
	where $v_D(\SX,t) = e_p(p)$ on $\p\Omega_D$ follows from $e_p(p)|_{\p\Omega_D} = (\dot{r}_r + v_L)|_{\p\Omega_D} $ with  $\dot{r}_r|_{\p\Omega_D} = 0$ and $v_L|_{\p\Omega_D} = v_D(\SX,t)$.
\end{proposition}

\begin{proof}
	The proof of this proposition is provided in Appendix \ref{app:proof_Prop4}.
\end{proof}

Before proceeding with the finite element discretization, two remarks regarding the physical interpretation and general applicability of the proposed framework are discussed. 

First, the distributed conjugate pair $(u_\Omega, y_\Omega)$ demonstrates the energy neutrality of the lifting framework. As derived in the lifted PHS models, the distributed output vanishes ($y_\Omega = 0$), resulting in zero power exchange within the domain ($\int_\Omega y_\Omega^\top u_\Omega \, d\X = 0$). Physically, this implies that the arbitrary interior lifting velocity $v_L|_\Omega = v_\Omega$ does not inject or extract energy from the system. Thus, the power balance remains governed  by the physical boundary ports. 

Second, although this lifting methodology has been derived using virtual power principles inspired by the Hamilton-Pontryagin and Hu-Washizu formulations, its core mechanism is independent of these specific functionals. Consequently, this strategy can be integrated into alternative variational schemes. The theoretical guarantee of yielding a structure-preserving lifted PHS model relies entirely on the capacity of the base (unlifted) variational formulation to establish a PHS structure. Therefore, the proposed lifting framework is not restricted to the jet-bundle and Stokes-Dirac representations presented herein, but stands as a versatile theoretical tool capable of being coupled with other geometric formalisms in the literature.

\section{Structure-preserving discretization with strong imposition of Dirichlet boundary velocities} \label{sec:FEM}

Building upon the variational framework, the naturally induced mixed finite element schemes are used to construct structure-preserving spatial discretizations. These discrete models guarantee the strong imposition of non-homogeneous Dirichlet boundary velocities while avoiding DAEs, strictly retaining a pure ODE system topology.

\subsection{Finite element considerations} \label{sec:FEM_definitions}

To discretize the continuous lifted PHS models, the spatial domain $\Omega$ is partitioned into $n_e$ non-overlapping elements $\Omega^e$, such that $\Omega \approx \bigcup_{e=1}^{n_e} \Omega^e$. Let $N\!N$ denote the total number of nodes in the mesh. These nodes are topologically partitioned into three mutually exclusive sets: interior nodes ($N\!N_\Omega$), nodes strictly located on the Neumann boundary ($N\!N_{\p\Omega_N}$), and nodes strictly located on the Dirichlet boundary ($N\!N_{\p\Omega_D}$). 

The kinematic degrees of freedom (DOFs) are categorized into Dirichlet DOFs, $N_D = n N\!N_{\p\Omega_D}$, Neumann DOFs, $N_N = n N\!N_{\p\Omega_N}$, and \textit{free} DOFs, $N_F = n(N\!N_\Omega + N\!N_{\p\Omega_N})$. For the Stokes-Dirac  formulation, the generalized strain and stress fields require $N_S$ global DOFs. To maintain generality, this dimension supports both continuous nodal interpolations ($N_S = m N\!N$) and discontinuous element-wise interpolations ($N_S = m \sum_{e=1}^{n_e} n_s^e$, where $n_s^e$ is the number of local evaluation points per element).

At the element level, the continuous fields are approximated using local shape functions. Let $N\!N_F^e$, $N\!N_D^e$, and $N\!N_N^e$ denote the number of free nodes, Dirichlet nodes, and Neumann nodes per element, respectively. To accommodate the kinematic lifting within the discretized domain, the discrete lifting functions are formulated first by separating their interior extensions from their boundary values:
\begin{align}
	\tilde{r}_L^e(\X,t) &= N_{r_\Omega}^e(\X) \hat{r}_\Omega^e(t) + N_{r_D}^e(\X) \hat{r}_D^e(t), \notag \\[1mm]
	\tilde{v}_L^e(\X,t) &= N_{r_\Omega}^e(\X) \hat{v}_\Omega^e(t) + N_{r_D}^e(\X) \hat{v}_D^e(t), \notag
\end{align}
where $N_{r_\Omega}^e(\X) \in \mathbb{R}^{n \times n N\!N_F^e}$ and $N_{r_D}^e(\X) \in \mathbb{R}^{n \times n N\!N_D^e}$ are the shape function associated with the free nodes and Dirichlet nodes, respectively. The local coefficient vectors $\hat{r}_\Omega^e(t), \hat{v}_\Omega^e(t) \in \mathbb{R}^{n N\!N_F^e}$ represent the arbitrary extensions of the prescribed boundary kinematics into the interior domain, while $\hat{r}_D^e(t), \hat{v}_D^e(t) \in \mathbb{R}^{n N\!N_D^e}$ represent the prescribed boundary values. Note that at the Dirichlet boundary ($\SX \in \p\Omega_D$), the interior shape functions vanish ($N_{r_\Omega}^e(\SX) = 0$). Consequently, the lifting only enforces the prescribed Dirichlet boundary conditions $\tilde{r}_L^e(\SX,t) = N_{r_D}^e(\SX) \hat{r}_D^e(t) = \tilde{r}_D^e(\SX,t)$ and $\tilde{v}_L^e(\SX,t) = N_{r_D}^e(\SX) \hat{v}_D^e(t) = \tilde{v}_D^e(\SX,t)$.

Building upon this lifting, the total element displacement $\tilde{r}^e(\X,t)$, total velocity $\dot{\tilde{r}}^e(\X,t)$, and generalized momentum $\tilde{p}^e(\X,t)$ are decomposed into the unknown relative dynamics and the prescribed lifting fields:
\begin{align}
	\tilde{p}^e(\X,t) &= N_p^e(\X) \hat{p}^e(t), \notag \\[1mm]
	\tilde{r}^e(\X,t) &= N_r^e(\X) \hat{r}_r^e(t) + \tilde{r}_L^e(\X,t) = N_r^e(\X) \hat{r}_r^e(t) + N_{r_\Omega}^e(\X) \hat{r}_\Omega^e(t) + N_{r_D}^e(\X) \hat{r}_D^e(t), \notag \\[1mm]
	\dot{\tilde{r}}^e(\X,t) &= N_r^e(\X) \dot{\hat{r}}_r^e(t) + \tilde{v}_L^e(\X,t) = N_r^e(\X) \dot{\hat{r}}_r^e(t) + N_{r_\Omega}^e(\X) \hat{v}_\Omega^e(t) + N_{r_D}^e(\X) \hat{v}_D^e(t), \notag
\end{align}
where $N_p^e(\X), N_r^e(\X) \in \mathbb{R}^{n \times n N\!N_F^e}$ are the shape functions, and $\hat{p}^e(t), \hat{r}_r^e(t) \in \mathbb{R}^{n N\!N_F^e}$ are the local time-dependent coefficients of the respective Galerkin approximations. 
For the strain, stress, and traction fields, the local approximations are defined as:
\begin{equation}
	\tilde{\q}^e(\X,t) = N_\q^e(\X) \hat{\q}^e(t), \qquad \tilde{e}_\q^e(\X,t) = N_{e_\q}^e(\X) \hat{e}_\q^e(t), \qquad \tilde{\tau}_N^e(\SX,t) = N_{\tau_N}^e(\SX) \hat{\tau}_N^e(t), \notag
\end{equation}
where $N_\q^e(\X), N_{e_\q}^e(\X) \in \mathbb{R}^{m \times m n_s^e}$ are the local shape functions for the $n_s^e$ evaluation points in the element, with corresponding  local time-dependent coefficient vectors $\hat{\q}^e(t), \hat{e}_\q^e(t) \in \mathbb{R}^{m n_s^e}$. In addition,  $N_{\tau_N}^e(\SX) \in \mathbb{R}^{n \times n N\!N_N^e}$ denotes the shape functions for imposed traction on $\p\Omega_N$, with $\hat{\tau}_N^e(t) \in \mathbb{R}^{n N\!N_N^e}$ being the prescribed coefficients.

The global vectors of time-dependent coefficients $\hat{p}(t), \hat{r}_r(t), \hat{r}_\Omega(t), \hat{v}_\Omega(t) \in \mathbb{R}^{N_F}$, $\hat{r}_D(t), \hat{v}_D(t) \in \mathbb{R}^{N_D}$, $\hat{\q}(t), \hat{e}_\q(t) \in \mathbb{R}^{N_S}$, and $\hat{\tau}_N(t) \in \mathbb{R}^{N_N}$ are related to their local counterparts through the standard boolean assembly matrices $L_p^e, L_r^e, L_{r_\Omega}^e \in \mathbb{R}^{n N\!N_F^e \times N_F}$, $L_{r_D}^e \in \mathbb{R}^{n N\!N_D^e \times N_D}$, $L_\q^e, L_{e_\q}^e \in \mathbb{R}^{m n_s^e \times N_S}$, and $L_{\tau_N}^e \in \mathbb{R}^{n N\!N_N^e \times N_N}$.

\begin{remark} \label{rem:Galerkin_strains}
	To guarantee the PHS structure at the discrete level, the shape functions of power-conjugate variables cannot be chosen arbitrarily. They must be linearly related through symmetric positive-definite mapping matrices to ensure that the resulting discrete mass and constitutive operators preserve symmetry and positive definiteness. Specifically, the local shape functions must satisfy $N_p^e(\X) = A_p^e N_r^e(\X)$ and $N_\q^e(\X) = A_\q^e N_{e_\q}^e(\X)$, where $A_p^e = (A_p^e)^\top > 0 \in \mathbb{R}^{n \times n}$ and $A_\q^e = (A_\q^e)^\top > 0 \in \mathbb{R}^{m \times m}$. In standard Galerkin finite element schemes, these mappings are conventionally chosen as identity matrices ($A_p^e = I_n$, $A_\q^e = I_m$). Furthermore, to ensure mathematical consistency in the lifting definition, the shape functions for the relative dynamics and the interior lifting must be distributed over the same spatial subspace, making it mandatory that $N_r^e(\X) = N_{r_\Omega}^e(\X)$.
\end{remark}

\subsection{Finite-dimensional port-Hamiltonian models}

This subsection develops the finite element spatial discretization of the continuous models. The procedure is structured in two stages for each representation. First, the local stationary conditions are established over an individual element. Applying integration by parts (Lemma \ref{lemma:integration}) to these conditions yields the standard local weak forms required for continuous Galerkin discretizations. Second, the local approximations are substituted into the weak forms and assembled over the global domain, yielding the finite-dimensional PHS models. The formulation is presented first for the jet-bundle approach, followed by the Stokes-Dirac approach.

\begin{proposition}[Local stationary condition for the lifted jet-bundle PHS] \label{prop:local_weak_form_jet}
	The local stationary condition of the lifted jet-bundle PHS, induced by the Hamilton-Pontryagin-based virtual power principle, is defined over an element domain $\Omega^e$ and its Neumann boundary $\p\Omega_N^e$ as:
	\begin{align}
		\delta_{\dot{r}_r} \mathcal{P}^e_{\HP} &= \int_{\Omega^e} \delta \dot{r}_r^e\cdot \big[\dot{p}^e + \F_\X(r^e)^*e_\q(\q(r^e)) -b\big] d\X - \int_{\p\Omega^e_N} \!\!\! \delta{\dot{r}_r^e}\cdot \big[\tau^e_N - F_\p(r^e) e_\q(\q(r^e))\big] d\SX , \label{eq:weak_r_jet} \\[1mm]
		\delta_{\dot{p}} \mathcal{P}^e_{\HP} &= \int_{\Omega^e} \delta \dot{p}^e\cdot \big[\dot{r}_r^e + v^e_L - \mathcal{M}^{-1}p^e\big] d\X, \label{eq:weak_p_jet}
	\end{align}
	for all kinematically admissible local variations $\delta\dot{r}_r^e$ and arbitrary local variations $\delta\dot{p}^e$.
\end{proposition}

\begin{proof}
	The derivation follows directly from the proof of Proposition \ref{prop:lifted_jet_bundlePHS} in Appendix \ref{app:proof_Prop2}.
\end{proof}

\begin{theorem}[Discrete lifted jet-bundle PHS] \label{theo:discrete_lifted_jet}
	Let the continuous fields be discretized via the local Galerkin approximations according with Remark \ref{rem:Galerkin_strains}: $\tilde{p}^e(\X,t) = N_p^e(\X) \hat{p}^e(t)$ and $\tilde{r}^e(\X,t) = N_r^e(\X) \hat{r}_r^e(t) + N_{r_\Omega}^e(\X) \hat{r}_\Omega^e(t) + N_{r_D}^e(\X) \hat{r}_D^e(t)$, together with $\tilde{v}_L^e(\X,t) = N_{r_\Omega}^e(\X) \hat{v}_\Omega^e(t) + N_{r_D}^e(\X) \hat{v}_D^e(t)$ and $\tilde{\tau}_N^e(\SX,t) = N_{\tau_N}^e(\SX) \hat{\tau}_N^e(t)$. The spatial discretization following Proposition \ref{prop:local_weak_form_jet} yields the finite-dimensional PHS with strong imposition of Dirichlet boundary velocities:
	\begin{align}
		\underbrace{\begin{bmatrix} \dot{\hat{p}}(t) \\ \dot{\hat{r}}_r(t) \\ \dot{\hat{r}}_D(t) \\ \dot{\hat{r}}_\Omega(t) \end{bmatrix}}_{\dot{\hat{z}}(t)} &= 
		\underbrace{\begin{bmatrix} 0 & -\hat{M}_{pr}^{-\top} & 0 & \;\;0 \\ \hat{M}_{pr}^{-1} & 0 &  0 & \;\;0 \\ 0 & 0 & 0 & \;\;0 \\ 0  & 0 & 0 & \;\;0 \end{bmatrix}}_{\hat{J} = -\hat{J}^\top} 
		\underbrace{\begin{bmatrix} \hat{e}_p(\hat{p}) \\ \hat{e}_{r_r}(\hat{r}) \\ \hat{e}_{r_D}(\hat{r}) \\ \hat{e}_{r_\Omega}(\hat{r}) \end{bmatrix}}_{\nabla_{\hat{z}} \hat{H}} + 
		\underbrace{\begin{bmatrix} \hat{M}_{pr}^{-\top}\hat{B}_N & 0 & 0\\ 0 & -\hat{M}_{pr}^{-1}\hat{M}_{pD} & -I_{N_F} \\ 0 & I_{N_D} & 0 \\ 0 & 0 & I_{N_F} \end{bmatrix}}_{\hat{G}} 
		\underbrace{\begin{bmatrix} \hat{\tau}_N(t) \\ \hat{v}_D(t) \\  \hat{v}_\Omega(t) \end{bmatrix}}_{\hat{u}(t)}, \label{eq:theo_discrete_jet_dyn} \\[1mm]
		\hat{y}(t) &= \hat{G}^\top \nabla_{\hat{z}} \hat{H}(\hat{z}) =
		\begin{bmatrix} \hat{v}_N(t)^\top & \hat{\tau}_D(t)^\top & 0 \end{bmatrix}^\top, \label{eq:theo_discrete_jet_out} \\[1mm]
		\hat{H}(\hat{z}) &= \frac{1}{2} \hat{p}(t)^\top \hat{M}_{pp} \hat{p}(t) + \hat{\Psi}(\hat{r}) - [\hat{r}_r(t) + \hat{r}_\Omega(t)]^\top \hat{b}_r - \hat{r}_D(t)^\top \hat{b}_D, \label{eq:discrete_hamiltonian}
	\end{align}
	with the discrete strain energy defined as $\hat{\Psi}(\hat{r})  = \sum_{e=1}^{n_e} \int_{\Omega^e} \Psi(\q(\tilde{r}^e)) \, d\X$, and the notation $\hat{r} = (\hat{r}_r,\hat{r}_D,\hat{r}_\Omega)$ introduced for clarity. The global matrices and vectors are assembled from their local contributions as:
	\begin{align}
		\hat{M}_{pr} &= \sum_{e=1}^{n_e} (L_p^e)^\top \!\! \int_{\Omega^e} (N_p^e)^\top N_r^e \, d\X \, L_r^e,  \quad & \hat{M}_{pp} &= \sum_{e=1}^{n_e} (L_p^e)^\top \!\! \int_{\Omega^e} (N_p^e)^\top \mathcal{M}^{-1} N_p^e \, d\X \, L_p^e,  \notag \\
		\hat{M}_{pD} &= \sum_{e=1}^{n_e} (L_p^e)^\top \!\! \int_{\Omega^e} (N_p^e)^\top N_{r_D}^e \, d\X \, L_{r_D}^e, \quad & \hat{B}_N &= \sum_{e=1}^{n_e} (L_r^e)^\top \!\! \int_{\p\Omega_N^e} \!\!\! (N_r^e)^\top N_{\tau_N}^e \, d\SX \, L_{\tau_N}^e, \notag \\
		\hat{F}_{e_\q r}(\hat{r}) &= \sum_{e=1}^{n_e} (L_r^e)^\top \int_{\Omega^e} ({\F}_\X(\tilde{r}^e)N_r^e)^\top e_\q(\q(\tilde{r}^e)) \, d\X,  \quad &  \hat{b}_r &= \sum_{e=1}^{n_e} (L_r^e)^\top \!\! \int_{\Omega^e} (N_r^e)^\top b \, d\X, \notag \\
		\hat{F}_{e_\q D}(\hat{r}) &= \sum_{e=1}^{n_e} (L_{r_D}^e)^\top \int_{\Omega^e} ({\F}_\X(\tilde{r}^e)N_{r_D}^e)^\top e_\q(\q(\tilde{r}^e)) \, d\X,  \quad &  \hat{b}_D &= \sum_{e=1}^{n_e} (L_{r_D}^e)^\top \!\! \int_{\Omega^e} (N_{r_D}^e)^\top b \, d\X, \notag 
	\end{align}
	together with $I_{N_F} \in \mathbb{R}^{N_F \times N_F}$ and $I_{N_D} \in \mathbb{R}^{N_D \times N_D}$ being identity matrices.  The discrete co-energy variables evaluate to $\hat{e}_p(\hat{p}) = \hat{M}_{pp} \hat{p}$, $\hat{e}_{r_D}(\hat{r}) = \hat{F}_{e_\q D}(\hat{r}) - \hat{b}_D$, and $\hat{e}_{r_r}(\hat{r}) = \hat{e}_{r_\Omega}(\hat{r}) = \hat{F}_{e_\q r}(\hat{r}) - \hat{b}_r$. 
\end{theorem}

\begin{proof}
		The proof of this Theorem is provided in Appendix \ref{app:Theo1}.
\end{proof}

\begin{corollary} \label{cor:FEM_equivalence1}
	Let the interior lifting extensions be strictly zero ($\hat{r}_\Omega(t) = 0$ and $\hat{v}_\Omega(t) = 0$). Furthermore, let the momentum shape functions be defined as $N_p^e(\X) = \mathcal{M}(\X) N_r^e(\X)$. Under these choices, the spatial discretization from Theorem \ref{theo:discrete_lifted_jet} reduces to an equivalent single-field displacement-based finite element formulation with strong imposition of Dirichlet boundary velocities. By defining $\hat{p}_{r}(t) = \hat{M}_{rr} \hat{p}(t)$, the discrete PHS takes the canonical structure:
	\begin{align}
		\underbrace{\begin{bmatrix} \dot{\hat{p}}_{r}(t) \\ \dot{\hat{r}}_r(t) \\ \dot{\hat{r}}_D(t) \end{bmatrix}}_{\dot{\hat{z}}(t)} &= 
		\underbrace{\begin{bmatrix} 0 & -I_{N_F} & 0 \\ I_{N_F} & 0 &  0 \\ 0 & 0 & 0 \end{bmatrix}}_{\hat{J} = -\hat{J}^\top} 
		\underbrace{\begin{bmatrix} \hat{e}_{p_{r}}(\hat{p}_{r}) \\ \hat{e}_{r_r}(\hat{r}) \\ \hat{e}_{r_D}(\hat{r}) \end{bmatrix}}_{\nabla_{\hat{z}} \hat{H}} + 
		\underbrace{\begin{bmatrix} \hat{B}_N & 0 \\ 0 & -\hat{M}_{{rr}}^{-1}\hat{M}_{{r D}} \\ 0 & I_{N_D} \end{bmatrix}}_{\hat{G}} 
		\underbrace{\begin{bmatrix} \hat{\tau}_N(t) \\ \hat{v}_D(t) \end{bmatrix}}_{\hat{u}_\p(t)}, \label{eq:FEM_ODE_PHS} \\[1mm]
		\hat{y}_\p(t) &= \hat{G}^\top \nabla_{\hat{z}} \hat{H}(\hat{z}) =
		\begin{bmatrix} \hat{v}_N(t)^\top & \hat{\tau}_D(t)^\top \end{bmatrix}^\top, \\[1mm]
		\hat{H}(\hat{z}) &= \frac{1}{2} \hat{p}_{r}(t)^\top \hat{M}_{rr}^{-1} \hat{p}_{r}(t) + \hat{\Psi}(\hat{r}) - \hat{r}_r(t)^\top \hat{b}_r - \hat{r}_D(t)^\top \hat{b}_D. \label{eq:discrete_hamiltonianFEM}
	\end{align}
	The additional involved matrices are defined as: $ $\\[-7mm]
	\begin{align}
		\hat{M}_{rr} &= \sum_{e=1}^{n_e} (L_r^e)^\top \!\! \int_{\Omega^e} (N_r^e)^\top \mathcal{M} N_r^e \, d\X \, L_r^e,  \quad & \hat{M}_{rD} &= \sum_{e=1}^{n_e} (L_{r_D}^e)^\top \!\! \int_{\Omega^e} (N_{r_D}^e)^\top \mathcal{M} N_r^e \, d\X \, L_r^e,  \notag \\[-7mm] \notag
	\end{align}
	with the discrete co-energies variables $\hat{e}_{p_{r}}(\hat{p}_{r}) = \hat{M}_{rr}^{-1} \hat{p}_{r}$, $\hat{e}_{r_r}(\hat{r}) = \hat{F}_{e_\q r}(\hat{r}) - \hat{b}_r$, and $\hat{e}_{r_D}(\hat{r}) = \hat{F}_{e_\q D}(\hat{r}) - \hat{b}_D$.
\end{corollary}

It follows from Corollary \ref{cor:FEM_equivalence1} that the derived matrices $\hat{M}_{rr}$ and $\hat{M}_{rD}$ correspond to the partitions of the consistent mass matrix from a standard displacement-based finite element approach, obtained when the global nodal displacement vector (including all nodes of the mesh) is partitioned as $\hat{r}(t) = [{\hat{r}}_r(t)^\top \; {\hat{r}}_D(t)^\top]^\top$. In this regime, since $\hat{r}_\Omega(t) = 0$, the field $\hat{r}_r(t)$ ceases to represent a relative displacement and recovers its significance as the absolute nodal displacement of the system. The practical implications of this equivalence will be briefly discussed in Section \ref{ssec:FEM_algebra}.

Having established the finite-dimensional jet-bundle representation, we now continue with the formulation of the lifted PHS in Stokes-Dirac structure, presenting first its local stationary condition and subsequently its spatial discretization.


\begin{proposition}[Local stationary condition for the lifted Stokes-Dirac PHS] \label{prop:local_weak_form_StokesDirac}
	The local stationary condition of the lifted Stokes-Dirac PHS, induced by the Hu-Washizu-based virtual power principle, is defined over an element domain $\Omega^e$ and its Neumann boundary $\p\Omega_N^e$ as:
	\begin{align}
		\delta_{\dot{r}_r} \mathcal{P}^e_{\HW} &= \int_{\Omega^e} \delta \dot{r}_r^e\cdot \big[\dot{p}^e + \F_\X(r^e)^*e_\q^e-b\big] d\X - \int_{\p\Omega^e_N} \!\!\!\delta{\dot{r}_r^e}\cdot \big[\tau^e_N - F_\p(r^e) e_\q^e\big] d\SX, \label{eq:weak_r_stokes} \\[1mm]
		\delta_{{e}_\q} \mathcal{P}^e_{\HW} &= \int_{\Omega^e}  \delta e_\q^e\cdot \big[\dot{\q}^e  - \F_\X(r^e)\dot{r}^e \big] 	d\X, \label{eq:weak_e_stokes}  \\[1mm]
		\delta_{\dot{p}} \mathcal{P}^e_{\HW} &= \int_{\Omega^e} \delta \dot{p}^e\cdot \big[\dot{r}_r^e + v^e_L - \mathcal{M}^{-1}p^e\big] d\X, \label{eq:weak_p_stokes} \\[1mm]
		\delta_{\dot{\q}} \mathcal{P}^e_{\HW} &= \int_{\Omega^e}  \delta \dot{\q}^e\cdot \left(e_\q^e -  \frac{\p \Psi}{\p\q}\Big|_{\q=\q^e}\right) d\X, \label{eq:weak_q_stokes}
	\end{align}
	for all kinematically admissible local variations $\delta\dot{r}_r^e$ and arbitrary local variations $\delta\dot{p}^e$, $\delta\dot{\q}^e$ and $\delta{e}_\q^e$.
\end{proposition}

\begin{proof}
	The derivation follows directly from the proof of Proposition \ref{prop:lifted_StokesPHS2} in Appendix \ref{app:proof_Prop4}.
\end{proof}

\begin{theorem}[Discrete lifted Stokes-Dirac PHS] \label{theo:discrete_lifted_stokes}
	Let the continuous fields be discretized via the local Galerkin approximations according with Remark \ref{rem:Galerkin_strains}: $\tilde{p}^e(\X,t) = N_p^e(\X) \hat{p}^e(t)$, $\tilde{\q}^e(\X,t) = N_\q^e(\X) \hat{\q}^e(t)$, $\tilde{e}_\q^e(\X,t) = N_{e_\q}^e(\X) \hat{e}_\q^e(t)$, and $\tilde{r}^e(\X,t) = N_r^e(\X) \hat{r}_r^e(t) + N_{r_\Omega}^e(\X) \hat{r}_\Omega^e(t) + N_{r_D}^e(\X) \hat{r}_D^e(t)$, together with $\tilde{v}_L^e(\X,t) = N_{r_\Omega}^e(\X) \hat{v}_\Omega^e(t) + N_{r_D}^e(\X) \hat{v}_D^e(t)$ and $\tilde{\tau}_N^e(\SX,t) = N_{\tau_N}^e(\SX) \hat{\tau}_N^e(t)$. The spatial discretization following Proposition \ref{prop:local_weak_form_StokesDirac} yields the finite-dimensional PHS with strong imposition of Dirichlet boundary velocities:
	\begin{align}
		\underbrace{\begin{bmatrix} \!\dot{\hat{p}}(t) \! \\ \!\dot{\hat{\q}}(t)\! \\ \!\dot{\hat{r}}_r(t)\! \\ \!\dot{\hat{r}}_D(t)\! \\ \!\dot{\hat{r}}_\Omega(t) \! \end{bmatrix}}_{\dot{\hat{x}}(t)} & \!\!=\!\! 
		\underbrace{\begin{bmatrix} 0 & \!\!\!\!\!\!\!-\!\hat{M}_{pr}^{-\!\top}\!\hat{F}_{r}(\hat{r})^{\!\top}\!\hat{M}_{e\q}^{-\!\top} & \!-\!\hat{M}_{pr}^{-\!\top} & 0 & 0 \\ \!\hat{M}_{e\q}^{-1}\!\hat{F}_{r}(\hat{r})\hat{M}_{pr}^{-1} & \!\!\!\!\!\!\!0 & 0 & 0 & 0 \\ \hat{M}_{pr}^{-1} & \!\!\!\!\!\!\!0 & 0 & 0 & 0 \\ 0 & \!\!\!\!\!\!\!0 & 0 & 0 & 0 \\ 0 & \!\!\!\!\!\!\!0 & 0 & 0 & 0 \end{bmatrix}}_{\hat{J}(\hat{x}) = -\hat{J}(\hat{x})^\top} \!\!
		\underbrace{\begin{bmatrix} \!\hat{e}_p(\hat{p})\! \\ \!\check{e}_\q(\hat{\q})\! \\ \!-\hat{b}_r\! \\ \!-\hat{b}_D\! \\ \!-\hat{b}_r\! \end{bmatrix}}_{\nabla_{\hat{x}}\hat{H}} 
		\!+ \!
		\underbrace{\begin{bmatrix} \!\hat{M}_{pr}^{-\!\top}\!\hat{B}_N & 0 & 0 \\ 0 & \!\!\!\!\hat{M}_{e\q}^{-1}\big(\!\hat{F}_{D}(\hat{r})\! -\! \hat{F}_r(\hat{r})\hat{M}_{pr}^{-1}\!\hat{M}_{pD}\!\big) & 0 \\ 0 & -\hat{M}_{pr}^{-1}\hat{M}_{pD} & \!\!-\!I_{N_F} \! \\ 0 & I_{N_D} & 0 \\ 0 & 0 & \!\!I_{N_F}\! \end{bmatrix}}_{\hat{G}(\hat{x})} \!\!
		\underbrace{\begin{bmatrix} \hat{\tau}_N(t) \\ \hat{v}_D(t) \\ \hat{v}_\Omega(t) \end{bmatrix}}_{\hat{u}(t)}\!, \label{eq:theo_discrete_stokes_dyn} \\[-2mm]
		\hat{y}(t) &= \hat{G}(\hat{x})^\top \nabla_{\hat{x}}\hat{H}(\hat{x}) = 
		\begin{bmatrix} \hat{v}_N(t)^\top &  \hat{\tau}_D(t)^\top & 0 \end{bmatrix}^\top, \label{eq:theo_discrete_stokes_out} \\[2mm]
		\hat{H}(\hat{x}) &= \frac{1}{2} \hat{p}(t)^\top \hat{M}_{pp} \hat{p}(t) + \hat{\Psi}(\hat{\q}) - [\hat{r}_r(t) + \hat{r}_\Omega(t)]^\top \hat{b}_r - \hat{r}_D(t)^\top \hat{b}_D, \label{eq:discrete_hamiltonian_stokes}
	\end{align}
	with the discrete strain energy defined as $\hat{\Psi}(\hat{\q}) = \sum_{e=1}^{n_e} \int_{\Omega^e} \Psi(\tilde{\q}^e) \, d\X$, and the notation $\hat{r} = (\hat{r}_r, \hat{r}_D, \hat{r}_\Omega)$ used for clarity. The global matrices and vectors are assembled from their local contributions as:
	\begin{align}
		\hat{M}_{pr} &= \sum_{e=1}^{n_e} (L_p^e)^\top \!\! \int_{\Omega^e} (N_p^e)^\top N_r^e \, d\X \, L_r^e,  \quad & \hat{M}_{pp} &= \sum_{e=1}^{n_e} (L_p^e)^\top \!\! \int_{\Omega^e} (N_p^e)^\top \mathcal{M}^{-1} N_p^e \, d\X \, L_p^e,  \notag \\
		\hat{M}_{pD} &= \sum_{e=1}^{n_e} (L_p^e)^\top \!\! \int_{\Omega^e} (N_p^e)^\top N_{r_D}^e \, d\X \, L_{r_D}^e, \quad & \hat{B}_N &= \sum_{e=1}^{n_e} (L_r^e)^\top \!\! \int_{\p\Omega_N^e} \!\!\! (N_r^e)^\top N_{\tau_N}^e \, d\SX \, L_{\tau_N}^e, \notag \\
		\hat{M}_{e\q} &= \sum_{e=1}^{n_e} (L_{e_\q}^e)^\top \!\! \int_{\Omega^e} (N_{e_\q}^e)^\top N_\q^e \, d\X \, L_\q^e, \quad & \check{e}_\q(\hat{\q}) &= \sum_{e=1}^{n_e} (L_\q^e)^\top \int_{\Omega^e} (N_\q^e)^\top \frac{\p \Psi}{\p\q}\bigg|_{\q(\tilde{\q}^e)} d\X, \notag \\
		\hat{F}_{r}(\hat{r}) &= \sum_{e=1}^{n_e} (L_{e_\q}^e)^\top \!\! \int_{\Omega^e} (N_{e_\q}^e)^\top \F_\X(\tilde{r}^e)N_r^e \, d\X \, L_r^e, \quad & \hat{b}_r &= \sum_{e=1}^{n_e} (L_r^e)^\top \!\! \int_{\Omega^e} (N_r^e)^\top b \, d\X, \notag \\
		\hat{F}_{D}(\hat{r}) &= \sum_{e=1}^{n_e} (L_{e_\q}^e)^\top \!\! \int_{\Omega^e} (N_{e_\q}^e)^\top \F_\X(\tilde{r}^e)N_{r_D}^e \, d\X \, L_{r_D}^e, \quad & \hat{b}_D &= \sum_{e=1}^{n_e} (L_{r_D}^e)^\top \!\! \int_{\Omega^e} (N_{r_D}^e)^\top b \, d\X, \notag
	\end{align}
	where the discrete co-energy momentum evaluates to $\hat{e}_p(\hat{p}) = \hat{M}_{pp} \hat{p}$.
	\end{theorem}

\begin{proof}
	The proof of this Theorem is provided in Appendix \ref{app:Theo2}.
\end{proof}

\begin{corollary} \label{cor:HWFEM_equivalence1}
	Let the interior lifting extensions be strictly zero ($\hat{r}_\Omega(t) = 0$ and $\hat{v}_\Omega(t) = 0$). Furthermore, let the momentum shape functions be defined as $N_p^e(\X) = \mathcal{M}(\X) N_r^e(\X)$. Under these choices, the spatial discretization from Theorem \ref{theo:discrete_lifted_stokes} reduces to an equivalent three-field mixed finite element formulation with strong imposition of Dirichlet boundary velocities. By defining $\hat{p}_{r}(t) = \hat{M}_{rr} \hat{p}(t)$, the discrete PHS takes the structure:
		\begin{align}
			\underbrace{\begin{bmatrix} \dot{\hat{p}}_{r}(t) \\ \dot{\hat{\q}}(t) \\ \dot{\hat{r}}_r(t) \\ \dot{\hat{r}}_D(t) \end{bmatrix}}_{\dot{\hat{x}}(t)} &= 
			\underbrace{\begin{bmatrix} 0 & -\hat{F}_{r}(\hat{r})^\top \hat{M}_{e\q}^{-\top} & -I_{N_F} & 0 \\ \hat{M}_{e\q}^{-1}\hat{F}_{r}(\hat{r}) & 0 &  0 & 0 \\ I_{N_F} & 0 & 0 & 0 \\ 0 & 0 & 0 & 0 \end{bmatrix}}_{\hat{J}(\hat{x}) = -\hat{J}(\hat{x})^\top} 
			\underbrace{\begin{bmatrix} \hat{e}_{p_{r}}(\hat{p}_{r}) \\ \check{e}_\q(\hat{\q}) \\ -\hat{b}_r \\ -\hat{b}_D \end{bmatrix}}_{\nabla_{\hat{x}} \hat{H}} + 
			\underbrace{\begin{bmatrix} \hat{B}_N & 0 \\ 0 & \hat{M}_{e\q}^{-1}\big(\hat{F}_{D}(\hat{r}) - \hat{F}_r(\hat{r})\hat{M}_{rr}^{-1}\hat{M}_{rD}\big) \\ 0 & -\hat{M}_{rr}^{-1}\hat{M}_{rD} \\ 0 & I_{N_D} \end{bmatrix}}_{\hat{G}(\hat{x})} 
			\underbrace{\begin{bmatrix} \hat{\tau}_N(t) \\ \hat{v}_D(t) \end{bmatrix}}_{\hat{u}_\p(t)}, \label{eq:FEM_ODE_PHSHW} \\[1mm]
			\hat{y}_\p(t) &= \hat{G}(\hat{x})^\top \nabla_{\hat{x}} \hat{H}(\hat{x}) =
			\begin{bmatrix} \hat{v}_N(t)^\top & \hat{\tau}_D(t)^\top \end{bmatrix}^\top, \\[1mm]
			\hat{H}(\hat{x}) &= \frac{1}{2} \hat{p}_{r}(t)^\top \hat{M}_{rr}^{-1} \hat{p}_{r}(t) + \hat{\Psi}(\hat{\q}) - \hat{r}_r(t)^\top \hat{b}_r - \hat{r}_D(t)^\top \hat{b}_D, \label{eq:discrete_hamiltonianFEMHW}
		\end{align}
		where $\hat{r}_r(t)$ represents the absolute nodal displacement vector, and the discrete interconnection matrix $\hat{J}(\hat{x}) = -\hat{J}(\hat{x})^\top$ recovers the same skew-symmetric structure as in \cite{ponce2025port} via the generalized Hamilton's principle.
\end{corollary}

Similar to the generalized Hamilton's principle approach in \cite{ponce2025port}, a limitation of the discretized models in Theorem \ref{theo:discrete_lifted_stokes} and Corollary \ref{cor:HWFEM_equivalence1} is their inability to directly yield static solutions in their current form. To illustrate this, let $\hat{x}^\star = [\hat{p}^{\star \top} \, \hat{\q}^{\star \top} \,  \hat{r}_r^{\star \top} \, \hat{r}_D^{\star \top} \, \hat{r}_\Omega^{\star \top}  ]^\top$ and $\hat{u}^\star = [\hat{\tau}_N^{\star \top} \, \hat{v}_D^{\star \top} \, \hat{v}_\Omega^{\star \top}  ]^\top$ denote the equilibrium state and input vectors, respectively, assuming $\dot{\hat{x}}=0$. Applying this to the model in Theorem \ref{theo:discrete_lifted_stokes} yields the following equations:
\begin{align}
	0 = &  -\!\hat{M}_{pr}^{-\top}\hat{F}_r(\hat{r}^\star)^\top \hat{M}_{e\q}^{-\top}\check{e}_\q(\hat{\q}^\star) + \hat{M}_{pr}^{-\top}\hat{b}_r +  \hat{M}_{pr}^{-\top}\hat{B}_N \hat{\tau}_N^\star \label{eq:equilbrio_1} \\
	0 = & \; \hat{M}_{e\q}^{-1}\hat{F}_r(\hat{r}^\star) \hat{M}_{pr}^{-1} \hat{M}_{pp} \hat{p}^\star + \hat{M}_{e\q}^{-1}\big(\!\hat{F}_{D}(\hat{r}^\star) \!-\! \hat{F}_r(\hat{r}^\star)\hat{M}_{pr}^{-1}\!\hat{M}_{pD}\!\big) \hat{v}_D^\star  \label{eq:equilbrio_2}\\
	0 = & \; \hat{M}_{pr}^{-1}\hat{M}_{pp} \hat{p}^\star -\hat{M}_{pr}^{-1}\hat{M}_{pD} \hat{v}_D^\star - \hat{v}_\Omega^\star  \label{eq:equilbrio_3} \\[0.25mm]
	0 = & \; \hat{v}_D^\star \label{eq:equilibrio_4} \\[0.5mm]
	0 = & \; \hat{v}_\Omega^\star. \label{eq:equilibrio_5}
\end{align}
As expected, equilibrium is achieved when $\hat{v}_D^\star = \hat{v}_\Omega^\star = 0$, which implies $\hat{p}^\star = 0$ from \eqref{eq:equilbrio_3}. Consequently, equation \eqref{eq:equilbrio_2} reduces to the trivial identity $0=0$. Therefore, solving the static problem requires \eqref{eq:equilbrio_1}, alongside an additional kinematic equation of the form $\hat{\q}^\star = \check{\q}(\hat{r}^\star)$.

\begin{proposition}[Kinematic equation] \label{prop:kinematicEQ}
	The local stationary condition of the kinematic equation, induced by the Hu-Washizu virtual work principle, is given by: \\[-1mm]
	\begin{equation}
		\var_{e_\q} \mathcal{W}_{\HW}^e = \dint_{\Omega^e}   \var e_\q^e \cdot [ {\q}^e - \q(r^e)] \, d\X. \label{eq:weak_Work_GHP}
	\end{equation}
	This represents the virtual work associated with the variation $\var e_\q^e$, where $\q(r^e) \in \R{m}$ denotes the continuous generalized strain expressed in terms of $r^e(\X,t)$. Applying the approximations from Theorem \ref{theo:discrete_lifted_stokes} yields the following relation:
	\begin{align}
		\hat{\q}(t)  = & \; \hat{M}_{e\q}^{-1} \sum_{e=1}^{n_e} (L_{e_\q}^e)^\top \!\! \dint_{\Omega^e} (N_{e_\q}^e)^\top \q(\tilde{r}^e)\, d\X = \check{\q}(\hat{r}), \label{eq:equilbrio_2v2}
	\end{align}  
	which holds at all time instants, including at equilibrium, i.e., $\q^\star = \check{\q}(\hat{r}^\star)$.
\end{proposition}
\begin{proof}
	The proof is detailed in Appendix \ref{app:prop7}.
\end{proof}

\subsection*{Discussion about the computational implementation}

	As proven in Theorems \ref{theo:discrete_lifted_jet} and \ref{theo:discrete_lifted_stokes}, the output port associated with the interior extension evaluates to zero ($\hat{y}_\Omega = 0$). This guarantees that the input $\hat{v}_\Omega(t)$ is neutral with respect to the power balance. 

While the arbitrary definition $\hat{v}_\Omega(t)$ expands the state space, its power neutrality hints at a structural redundancy. From a computational perspective, explicitly integrating this expanded system introduces unnecessary overhead. The following corollary demonstrates that the finite-dimensional PHS models naturally collapse into a reduced state of absolute nodal displacements, decoupling the time integration from $\hat{v}_\Omega(t)$.

\begin{corollary}[Reduction to absolute displacements] \label{cor:reduction_absolute}
	Let $\hat{d}(t) = \hat{r}_r(t) + \hat{r}_\Omega(t) \in \mathbb{R}^{N_F}$ define the absolute displacement vector for the free degrees of freedom. Since the local shape functions satisfy the spatial compatibility condition $N_r^e(\X) = N_{r_\Omega}^e(\X)$ established in Remark \ref{rem:Galerkin_strains}, the continuous-time dynamic evolution of both the discrete lifted jet-bundle PHS (Theorem \ref{theo:discrete_lifted_jet}) and the discrete lifted Stokes-Dirac PHS (Theorem \ref{theo:discrete_lifted_stokes}) is invariant to the choice of $\hat{v}_\Omega(t)$, and can be solved in the reduced states $\hat{z}(t) = [\hat{p}(t)^\top \, \hat{d}(t)^\top \, \hat{r}_D(t)^\top]^\top$ and $\hat{x}(t) = [\hat{p}(t)^\top \, \hat{\q}(t)^\top \, \hat{d}(t)^\top \, \hat{r}_D(t)^\top]^\top$, respectively.
\end{corollary}

\begin{proof}
	Corollary \ref{cor:reduction_absolute} establishes that the reduction to an absolute displacement state, previously achieved in Corollaries \ref{cor:FEM_equivalence1} and \ref{cor:HWFEM_equivalence1} by strictly enforcing $\hat{v}_\Omega(t) = 0$, is an intrinsic structural property of the discrete formulation that holds true even when $\hat{v}_\Omega(t) \neq 0$. To demonstrate this, summing $\dot{\hat{r}}_r(t)$ and the interior lifting $\dot{\hat{r}}_\Omega(t)$ from either \eqref{eq:theo_discrete_jet_dyn} or \eqref{eq:theo_discrete_stokes_dyn} yields the absolute velocity:
	\begin{equation}
		\dot{\hat{d}}(t) = \dot{\hat{r}}_r(t) + \dot{\hat{r}}_\Omega(t) = \left[ \hat{M}_{pr}^{-1}\hat{e}_p(\hat{p}) - \hat{M}_{pr}^{-1}\hat{M}_{pD} \hat{v}_D(t) - \hat{v}_\Omega(t) \right] + \hat{v}_\Omega(t) = \hat{M}_{pr}^{-1}\hat{e}_p(\hat{p}) - \hat{M}_{pr}^{-1}\hat{M}_{pD} \hat{v}_D(t),
	\end{equation}
	where $\hat{v}_\Omega(t)$ is canceled. Furthermore, since $N_r^e(\X) = N_{r_\Omega}^e(\X)$, the total element approximation of the displacement field yields $\tilde{r}^e(\X,t) =  N_r^e(\X) \hat{r}_r^e(t) + N_{r_\Omega}^e(\X) \hat{r}_\Omega^e(t) + N_{r_D}^e(\X) \hat{r}_D^e(t) =  N_r^e(\X) \hat{d}^e(t) + N_{r_D}^e(\X) \hat{r}_D^e(t)$. Therefore, all discrete matrices and vectors from Theorems \ref{theo:discrete_lifted_jet} and \ref{theo:discrete_lifted_stokes} that are evaluated in  $\hat{r} = (\hat{r}_r, \hat{r}_\Omega, \hat{r}_D)$, actually evaluate in $\hat{r} = (\hat{r}_r + \hat{r}_\Omega, \hat{r}_D) = (\hat{d},\hat{r}_D)$. Consequently, the continuous-time discretized PHS models can be integrated in time with absolute independence from $\hat{v}_\Omega(t)$. If required for analysis or control purposes, the relative dynamic state can be recovered a posteriori by direct temporal integration of the interior lifting velocity, i.e.,  $\hat{r}_r(t) = \hat{d}(t) - \int_0^t \hat{v}_\Omega(\tau) \, d\tau - \hat{r}_\Omega(0)$.
\end{proof}


\subsection{Connections with other finite element approaches} \label{ssec:FEM_algebra}



As established in Corollaries \ref{cor:FEM_equivalence1}, \ref{cor:HWFEM_equivalence1}, and \ref{cor:reduction_absolute}, the proposed lifting framework inherently reduces to finite-dimensional models governed by absolute nodal displacements, irrespective of the interior velocity field $\hat{v}_\Omega(t)$. Interestingly, this structural reduction can be algebraically reconstructed from classical spatial discretizations. This section demonstrates this equivalence, establishing a practical methodology to transition directly from standard or mixed FEM models into structure-preserving PHS formulations with strongly imposed Dirichlet boundary velocities, bypassing the need for penalty factors, Lagrange multipliers, or non-trivial kinematic liftings.

\subsubsection{Standard displacement-based finite element formulation}

Consider the spatial discretization of an elastodynamic system using a displacement-based finite element method, prior to the imposition of Dirichlet boundary conditions. This leads to the second-order unconstrained system:
\begin{equation}
	\underbrace{\begin{bmatrix} M_{11} & M_{12} \\ M_{21} & M_{22} \end{bmatrix}}_{M} 
	\underbrace{\begin{bmatrix} \ddot{q}_1(t) \\ \ddot{q}_2(t) \end{bmatrix}}_{\ddot{q}(t)} + 
	\underbrace{\begin{bmatrix} \nabla_{q_1} V(q)\\ \nabla_{q_2} V(q) \end{bmatrix}}_{\nabla_q V(q)} = 
	\underbrace{\begin{bmatrix} B_N \\ 0 \end{bmatrix}}_{G_N} \tau_N(t),
	\label{eq:standard_FEM_remark}
\end{equation}
where $M = M^\top>0$ is the mass matrix, $V(q)$ is the potential energy, and $\tau_N(t)$ are the imposed boundary tractions. The nodal coordinates $q(t)$ are partitioned such that $q_2(t) \in \mathbb{R}^{N_D}$ corresponds to the degrees of freedom on the Dirichlet boundary $\p\Omega_D$ (where $\dot{q}_2(t) = v_D(t)$ is strongly imposed), and $q_1(t) \in \mathbb{R}^{N_F}$ represents the nodes in the domain interior and Neumann boundary (the \textit{free} nodes).

The discrete kinetic energy of the unconstrained system is $T(\dot{q}) = \frac{1}{2}\dot{q}(t)^\top M \dot{q}(t)$. In the PHS framework, the state variable is the momentum $p_1(t)$ energy-conjugated to the velocity $\dot{q}_1(t)$. Taking the partial derivative with respect to $\dot{q}_1(t)$ yields:
\begin{equation}
	p_1(t) = \frac{\p T}{\p \dot{q}_1} = M_{11}\dot{q}_1(t) + M_{12}\dot{q}_2(t) = M_{11}\dot{q}_1(t) + M_{12}v_D(t). \notag
\end{equation}
Since $M$ is symmetric and positive definite, the principal submatrix $M_{11}$ is invertible. Thus, the interior velocity is isolated:
\begin{equation}
	\dot{q}_1(t) = M_{11}^{-1}p_1(t) - M_{11}^{-1}M_{12}v_D(t). \label{eq:FEM_q1_dot}
\end{equation}

From the first block-row of the dynamic equation \eqref{eq:standard_FEM_remark}, the balance equation for the free degrees of freedom is $M_{11} \ddot{q}_1(t) + M_{12}\ddot{q}_2(t) + \nabla_{q_1} V(q) = B_N \tau_N(t)$. Recognizing that the strong imposition implies $\ddot{q}_2(t) = \dot{v}_D(t)$, the time derivative of the momentum is $\dot{p}_1(t) = M_{11}\ddot{q}_1(t) + M_{12}\dot{v}_D(t)$. Substituting this equivalence simplifies the momentum balance to:
\begin{equation}
	\dot{p}_1(t) = -\nabla_{q_1} V(q) + B_N \tau_N(t). \label{eq:FEM_p1_dot}
\end{equation}

Gathering the imposed boundary velocities $\dot{q}_2(t) = v_D(t)$, the kinematic relation \eqref{eq:FEM_q1_dot}, and the momentum balance \eqref{eq:FEM_p1_dot}, the FEM-based model is transformed into the PHS form:
\begin{align}
	\underbrace{\begin{bmatrix} \dot{p}_1(t) \\ \dot{q}_1(t) \\ \dot{q}_2(t) \end{bmatrix}}_{\dot{z}(t)} &= 
	\underbrace{\begin{bmatrix} 0 & -I_{N_F} & 0 \\ I_{N_F} & 0 & 0 \\ 0 & 0 & 0 \end{bmatrix}}_{J = -J^\top} 
	\underbrace{\begin{bmatrix} M_{11}^{-1}p_1(t) \\ \nabla_{q_1} V(q) \\ \nabla_{q_2} V(q) \end{bmatrix}}_{\nabla_z H(z)} + 
	\underbrace{\begin{bmatrix} B_N & 0 \\ 0 & -M_{11}^{-1}M_{12} \\ 0 & I_{N_D} \end{bmatrix}}_{G} 
	\underbrace{\begin{bmatrix} \tau_N(t) \\ v_D(t) \end{bmatrix}}_{u_\p(t)}, \notag \\[2mm]
	y_\p(t) &= G^\top \nabla_z H(z), \notag \\[2mm]
	H(z) &= \frac{1}{2} p_1(t)^\top M_{11}^{-1} p_1(t) + V(q), \notag
\end{align}
which mirrors the structure derived in Corollary \ref{cor:FEM_equivalence1}.

\subsubsection{Mixed finite element formulations}

The methodology described above extends to mixed finite element models derived from alternative variational approaches, such as the Hellinger-Reissner principle \cite{thoma2022explicit,thoma2024velocity,brugnoli2022explicit}, the Hu-Washizu-based methods \cite{kinon2026mixed,hille2026port}, the generalized Hamilton's principle \cite{ponce2025port}, the linked Lagrange multiplier method \cite{ponce2024structure}, among others. 

Simplifying the presentation, applying these variational approaches while initially discarding the terms associated with the imposition of Dirichlet boundary conditions yield models as:
\begin{equation}
	\begin{bmatrix} M & 0 & 0 \\ 0 & I & 0 \\ 0 & 0 & I \end{bmatrix} 
	\begin{bmatrix} \dot{v}(t) \\ \dot{\varepsilon}(t) \\ \dot{r}(t) \end{bmatrix} = 
	\begin{bmatrix} 0 & -D(r)^\top & -I \\ D(r) & 0 & 0 \\ I & 0 & 0 \end{bmatrix} 
	\begin{bmatrix} v(t) \\ \sigma(\varepsilon) \\ e_r \end{bmatrix} + 
	\begin{bmatrix} G_N \\ 0 \\ 0 \end{bmatrix} \tau_N(t),
	\label{eq:mixed_FEM_unconstrained}
\end{equation}
where $M = M^\top>0$ is the mass matrix, $\varepsilon(t)$ is the discrete strain vector, $\sigma(\varepsilon)$ is the discrete stress, and $e_r$ is the gradient of the potential with respect to the displacement vector $r(t)$. The identity matrices $I$ are of appropriate dimensions. Partitioning the velocity and displacement vectors into components associated with the interior domain and Neumann boundary (subscript $1$) and the Dirichlet boundary (subscript $2$) induces the block structures:
\begin{equation}
 \quad v(t) = \begin{bmatrix} v_{1}(t) \\ v_{2}(t) \end{bmatrix} , \quad  r(t) = \begin{bmatrix} r_1(t) \\ r_2(t) \end{bmatrix}, \quad	M = \begin{bmatrix} M_{11} & M_{12} \\ M_{21} & M_{22} \end{bmatrix}, \quad G_N = \begin{bmatrix} B_N \\ 0 \end{bmatrix}, \quad e_r = \begin{bmatrix} e_{r_1} \\ e_{r_2} \end{bmatrix}. \notag
\end{equation}
The discrete version of the differential operator $D(r)$ is partitioned such that $D(r)v(t) = D_\Omega(r)v_1(t) + D_D(r)v_2(t)$. By enforcing $v_2(t) = v_D(t)$, the momentum for the free velocities is extracted as $p_1(t) = M_{11}v_1(t) + M_{12}v_D(t)$. Isolating the interior velocity yields:
\begin{equation}
	v_1(t) = M_{11}^{-1}p_1(t) - M_{11}^{-1}M_{12}v_D(t). \label{eq:mixed_FEM_v1}
\end{equation}

Substituting \eqref{eq:mixed_FEM_v1} into the second block-row of \eqref{eq:mixed_FEM_unconstrained} defines the strain dynamics:
\begin{equation}
	\dot{\varepsilon}(t) = D_\Omega(r)v_1(t) + D_D(r)v_D(t) = D_\Omega(r)M_{11}^{-1}p_1(t) + [D_D(r) - D_\Omega(r)M_{11}^{-1}M_{12}] v_D(t).
\end{equation}
From the first block-row of \eqref{eq:mixed_FEM_unconstrained}, the momentum balance is $M_{11}\dot{v}_1(t) + M_{12}\dot{v}_D(t) = -D_\Omega(r)^\top \sigma(\varepsilon) - e_{r_1} + B_N \tau_N(t)$. Recognizing the left-hand side as $\dot{p}_1(t)$, the dynamic equation reduces to:
\begin{equation}
	\dot{p}_1(t) = -D_\Omega(r)^\top \sigma(\varepsilon) - e_{r_1} + B_N \tau_N(t).
\end{equation}
Gathering the momentum balance, the strain dynamics, and the kinematics, the mixed FEM model maps into an ODE-PHS structure:
\begin{align}
	\underbrace{\begin{bmatrix} \dot{p}_1(t) \\ \dot{\varepsilon}(t) \\ \dot{r}_1(t) \\ \dot{r}_2(t) \end{bmatrix}}_{\dot{x}(t)} &= 
	\underbrace{\begin{bmatrix} 0 & -D_\Omega(r)^\top & -I_{N_F} & 0 \\ D_\Omega(r) & 0 & 0 & 0 \\ I_{N_F} & 0 & 0 & 0 \\ 0 & 0 & 0 & 0 \end{bmatrix}}_{J(x) = -J(x)^\top} 
	\underbrace{\begin{bmatrix} M_{11}^{-1}p_1(t) \\ \sigma(\varepsilon) \\ e_{r_1} \\ e_{r_2} \end{bmatrix}}_{\nabla_x H(x)} + 
	\underbrace{\begin{bmatrix} B_N & 0 \\ 0 & (D_D(r) - D_\Omega(r)M_{11}^{-1}M_{12}) \\ 0 & -M_{11}^{-1}M_{12} \\ 0 & I_{N_D} \end{bmatrix}}_{G(x)} 
	\underbrace{\begin{bmatrix} \tau_N(t) \\ v_D(t) \end{bmatrix}}_{u_\p(t)}, \notag \\[2mm]
	y_\p(t) &= G(x)^\top \nabla_x H(x), \notag \\[2mm]
	H(x) &= \frac{1}{2} p_1(t)^\top M_{11}^{-1} p_1(t) + U(\varepsilon) + r(t)^\top e_{r}, \notag
\end{align}
which mirrors the structure derived in Corollary \ref{cor:HWFEM_equivalence1}.

The methodology presented above demonstrates that the lifted PHS framework is consistent with standard computational mechanics practices, namely, the algebraic partitioning of matrices and vectors to strongly enforce Dirichlet boundary conditions. The advantage of the proposed framework is that it provides a systematic approach to achieve this while preserving the underlying geometric structure of the continuous models. Furthermore, it enables the definition of a relative displacement state and a distributed input port, which could be useful for system analysis and control design.

\section{Numerical experiments} \label{sec:Simulations}

This section validates the proposed structure-preserving lifting framework through numerical benchmarks. The simulations evaluate the finite-dimensional PHS models derived in Theorems \ref{theo:discrete_lifted_jet} and \ref{theo:discrete_lifted_stokes} by exploiting the algebraic reduction to the absolute displacements established in Corollary \ref{cor:reduction_absolute}, thereby avoiding the integration of relative fields.

\subsection{Spatial discretization accuracy: Shear locking and static equilibrium}

The first benchmark evaluates the spatial accuracy of the finite element schemes and their susceptibility to shear locking. A one-dimensional linear Timoshenko beam is considered, due to its well known shear locking issues. The generalized displacement vector contains the cross-section rotation and vertical deflection, $r(X,t) = [\psi(X,t) \; w(X,t)]^\top$, with material coordinate $X \in (0, L_0)$. 

The linear strain vector $\q(r) \in \mathbb{R}^2$ and the differential operator $\F_\X$ governing the strain rate $\dot{\q} = \F_\X \,\dot{r}$ are given by:
\begin{equation}
	\q(r) = \begin{bmatrix}
		\p_X \psi \\
		\p_X w - \psi
	\end{bmatrix}, \qquad
	\F_\X = \begin{bmatrix}
		\,\p_X & \;0 \\
		-1 & \; \p_X
	\end{bmatrix}.
	\notag
\end{equation}
The strain energy density is $\Psi(\q) = \frac{1}{2} \, \q^\top \mathcal{K}_\q \q$, with the constitutive stiffness matrix $\mathcal{K}_\q = \text{diag}(E I_0, \kappa G A_0)$, where $E$ is the Young's modulus, $G$ the shear modulus, $A_0$ the cross-sectional area, $I_0$ the second moment of inertia, and $\kappa$ the shear correction factor.

The geometric and physical parameters are: initial beam length $L_0 = 0.5$ [m], width $b = 0.03$ [m], thickness $h = 0.001$ [m] (yielding $A_0 = bh$ [m$^2$] and $I_0 = bh^3/12$ [m$^4$]), Young's modulus $E = 210$ [GPa], Poisson's ratio $\nu = 0.3$ [--], shear modulus $G = E/(2(1+\nu))$, and shear correction factor $\kappa = 5/6$ [--]. 

The beam is clamped at $X = 0$, imposed via the strong Dirichlet boundary condition $v_D(0,t) = 0$. A transverse point load is applied at the free end $X = L_0$ through the non-homogeneous Neumann boundary condition $\tau_N(L_0,t) = [0 \; P^\star]^\top$. Simulations are performed for five load levels: $P^\star \in \{1, 2, 3, 4, 5\}$ [N]. The analytical solution for the vertical deflection $w(X)$ under these conditions is:
\begin{equation}
	w(X) = \frac{P^\star}{E I_0} \left( \frac{L_0 X^2}{2} - \frac{X^3}{6} \right) + \frac{P^\star X}{\kappa G A_0}.
	\notag
\end{equation}
To isolate the spatial discretization error, the static equilibrium equations of the finite-dimensional models are solved. The discrete lifted jet-bundle formulation (Theorem \ref{theo:discrete_lifted_jet}) yields the static solution directly from its equilibrium conditions. In contrast, evaluating the static equilibrium for the discrete lifted Stokes-Dirac formulation (Theorem \ref{theo:discrete_lifted_stokes}) requires the supplementary kinematic equation established in Proposition \ref{prop:kinematicEQ}. Both formulations are evaluated over the domain using uniform meshes of $n_e = 3$ and $n_e = 10$ elements.

For the formulation in Theorem \ref{theo:discrete_lifted_jet}, continuous Lagrangian shape functions of degree 1 (P1) and 2 (P2) are implemented. The computed beam configurations are presented in Fig. \ref{fig:static_method1}. Under full integration, the P1 discretization exhibits severe shear locking, yielding an artificially rigid response that drastically underpredicts the deflection. However, this phenomenon is effectively alleviated by applying classical reduced integration, as demonstrated in the case with $n_e = 10$. Alternatively, increasing the interpolation order to P2 polynomials also mitigates the locking effect, even when using full integration; yet, the coarse 3-element mesh still displays a visible deviation from the exact equilibrium path, requiring further spatial refinement to accurately capture the analytical solution.

\begin{figure}[b]
	\centering
	\begin{tabular}{cc}
		\includegraphics[width=0.5\textwidth]{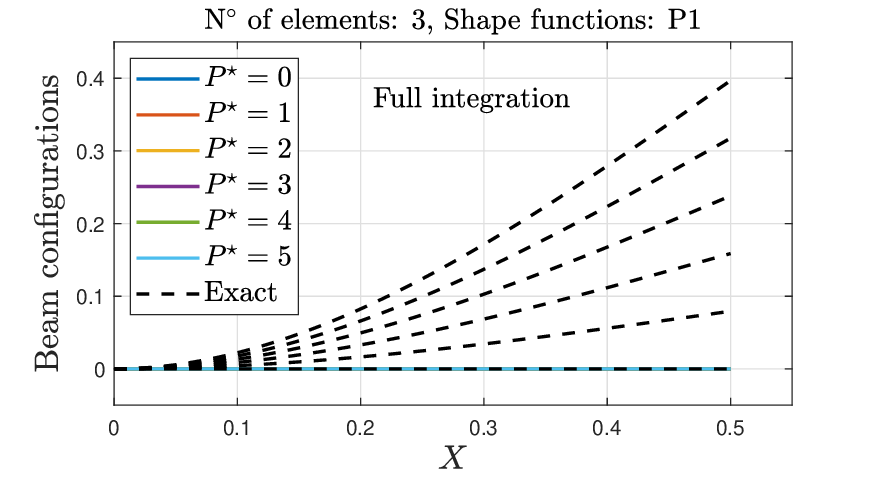} & \includegraphics[width=0.5\textwidth]{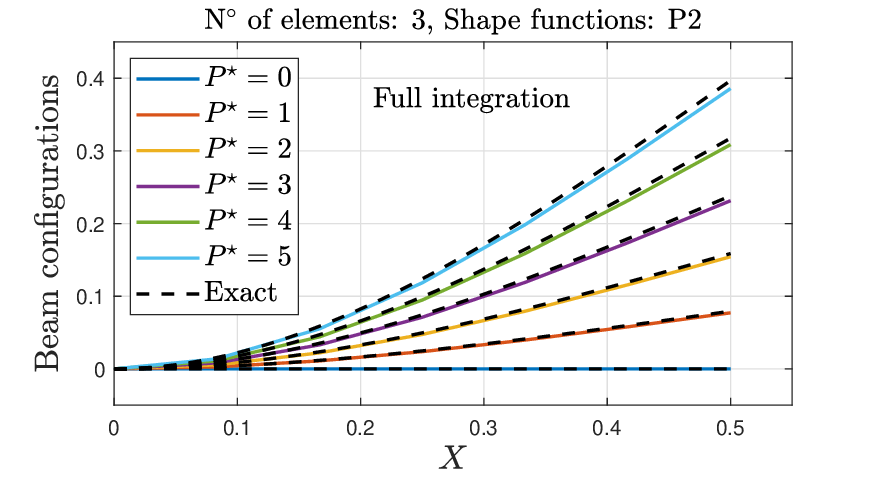} \\
		\includegraphics[width=0.5\textwidth]{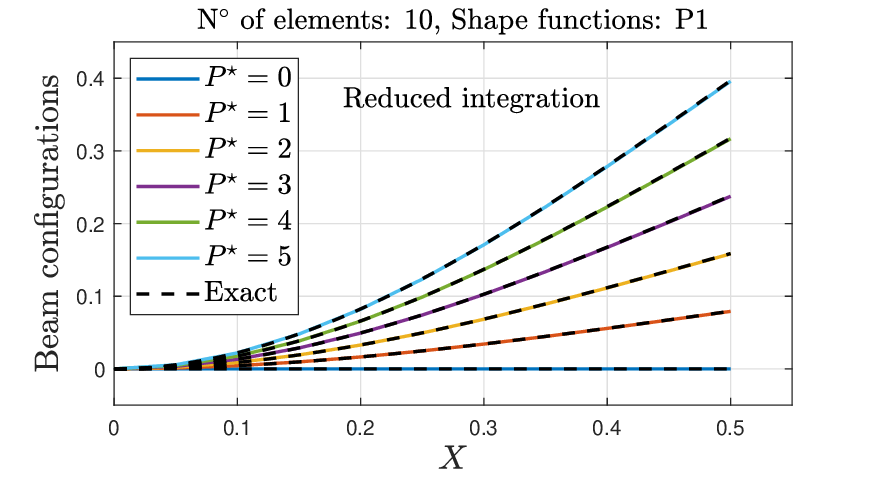} & \includegraphics[width=0.5\textwidth]{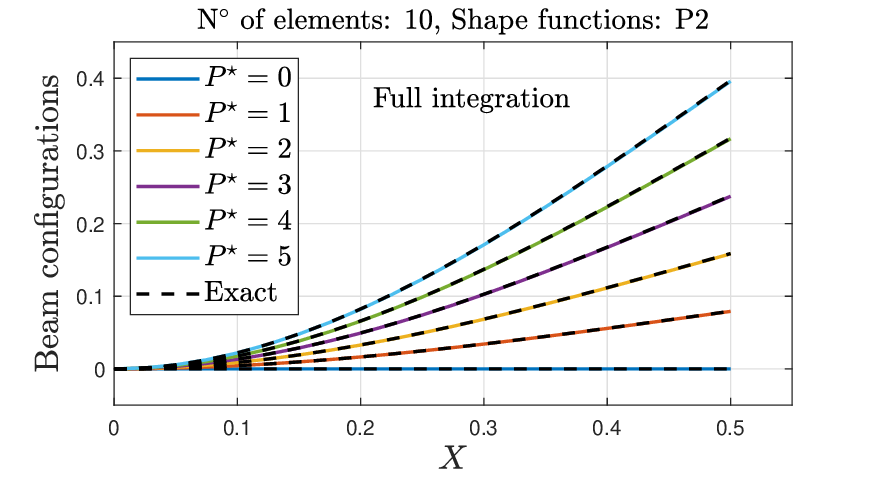}
	\end{tabular}\\[-3mm]
	\caption{Static deflections computed via the discrete lifted jet-bundle formulation (Theorem \ref{theo:discrete_lifted_jet}). Dashed lines represent the analytical solutions.}
	\label{fig:static_method1}
\end{figure}

%

For the formulation in Theorem \ref{theo:discrete_lifted_stokes}, the displacements and strain/stress interpolations are paired as P1-P0 (continuous linear displacements, discontinuous piecewise constant strains/stresses) and P2-P2 (continuous quadratic fields for all variables). The beam configurations are shown in Fig. \ref{fig:static_method2}. As expected from Hu-Washizu-type formulations, the model is inherently locking-free. The P1-P0 discretization successfully avoids artificial stiffness and closely approximates the analytical reference while employing full integration. Furthermore, the P2-P2 discretization yields a visibly superior approximation compared to the P2 scheme from Theorem \ref{theo:discrete_lifted_jet} on the coarse $n_e = 3$ mesh, as the introduction of additional degrees of freedom for the strain fields provides enhanced kinematic richness. As a preliminary conclusion, the discrete lifted jet-bundle formulation is susceptible to shear locking but can be alleviated with reduced integration, whereas the lifted Stokes-Dirac formulation preserves the locking-free properties of mixed schemes.

\begin{figure}[t]
	\centering
	\begin{tabular}{cc}
	\includegraphics[width=0.5\textwidth]{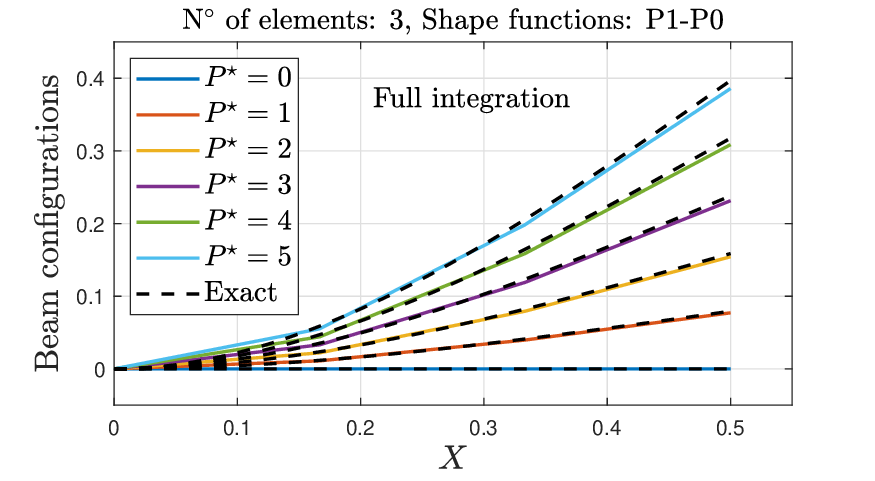} & \includegraphics[width=0.5\textwidth]{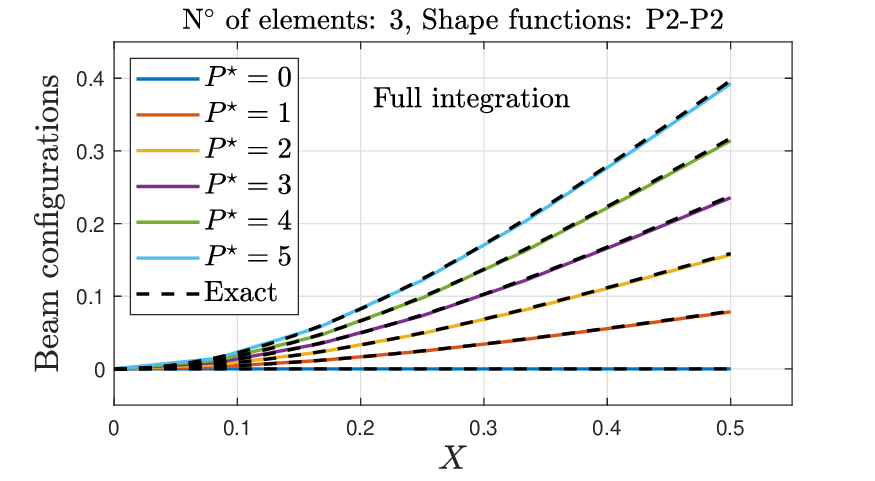} \\
	\includegraphics[width=0.5\textwidth]{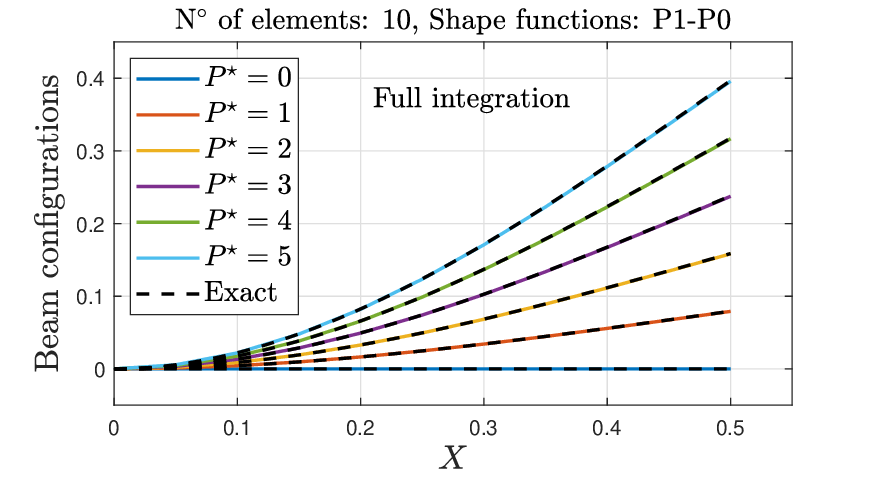} & \includegraphics[width=0.5\textwidth]{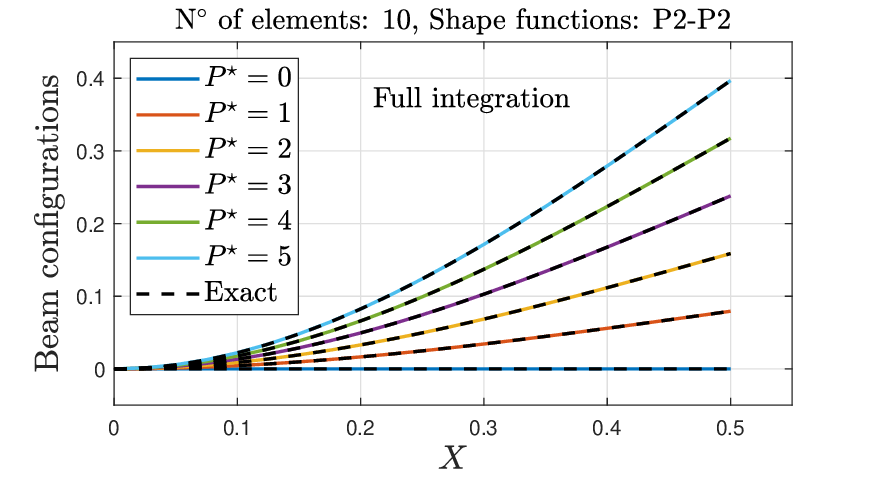}
\end{tabular}\\[-4mm]
\caption{Static deflections computed via the discrete lifted Stokes-Dirac formulation (Theorem \ref{theo:discrete_lifted_stokes}). Dashed lines represent the analytical solutions.}
	\label{fig:static_method2}
\end{figure}

\subsection{Energy conservation: homogeneous boundary conditions}

The second benchmark validates the energy-preserving properties of the discretized formulation. A one-dimensional spatial geometrically exact string is used considering a hyperelastic Saint Venant-Kirchhoff material under the action of a constant gravitational field.

Let $s \in [0, L_0]$ denote the arc-length coordinate in the undeformed state. The reference configuration is defined by the position vector $d_0(s) = [x_0(s) \; y_0(s) \; z_0(s)]^\top$. To satisfy geometric compatibility, the initial tangent vector $\p_s d_0(s)$ is constrained to be unitary throughout the domain, ensuring $\|\p_s d_0(s)\| = 1$.

The current spatial position is given by $d(s,t) = d_0(s) + r(s,t)$, where $r(s,t) = [r_x(s,t) \; r_y(s,t) \; r_z(s,t)]^\top$ is the generalized displacement vector. By quantifying the metric difference between the current and reference configurations, the generalized Green-Lagrange scalar strain $\q(s,t)$ is defined as:
\begin{equation}
	\q(r) = \frac{1}{2} \big( \p_s d \cdot \p_s d - \p_s d_0 \cdot \p_s d_0 \big) = \p_s d_0 \cdot \p_s r + \frac{1}{2} \p_s r \cdot \p_s r.
	\notag
\end{equation}
Following the Stokes-Dirac geometric structure established in Theorem \ref{theo:discrete_lifted_stokes}, the generalized state vector is assembled as $x = [p_x \; p_y \; p_z \; \q \; r_x \; r_y \; r_z]^\top \in \mathbb{R}^7$, where $p = [p_x \;\; p_y \;\; p_z]^\top = \rho_0 A_0 \dot{r}$ is the generalized momentum, with $\rho_0$ the material density and $A_0$ the cross-sectional area. The underlying formally skew-adjoint differential operator $\mathcal{J}(x) = -\mathcal{J}^*(x)$ coupling the momentum balance and the kinematic compatibility is given by:
\begin{equation}
	\mathcal{J}(x) = \begin{bmatrix}
		0 & 0 & 0 & \p_s \big( (\p_s x_0 + \p_s r_x) \, \cdot \big) & \;\;-1 & \;\;0 & \;\;0 \\[-0.75mm]
		0 & 0 & 0 & \p_s \big( (\p_s y_0 + \p_s r_y) \, \cdot \big) & \;\;0 & \;\;-1 & \;\;0 \\[-0.75mm]
		0 & 0 & 0 & \p_s \big( (\p_s z_0 + \p_s r_z) \, \cdot \big) & \;\;0 & \;\;0 & \;\;-1 \\[-0.75mm]
		(\p_s x_0 + \p_s r_x) \p_s & (\p_s y_0 + \p_s r_y) \p_s & (\p_s z_0 + \p_s r_z) \p_s & 0 & \;\;0 & \;\;0 & \;\;0 \\[-0.75mm]
		1 & 0 & 0 & 0 & \;\;0 & \;\;0 & \;\;0 \\[-0.75mm]
		0 & 1 & 0 & 0 & \;\;0 & \;\;0 & \;\;0 \\[-0.75mm]
		0 & 0 & 1 & 0 & \;\;0 & \;\;0 & \;\;0
	\end{bmatrix},
	\notag
\end{equation}
where the placeholder $(\cdot)$ denotes the argument upon which the differential operator acts. 
The Hamiltonian $H(x)$ encapsulates the kinetic energy ($E_{kin}$), the elastic strain energy ($E_{elas}$), and the total gravitational potential energy ($E_g$). The strain energy density function is given by $\Psi(\q) = \frac{1}{2} E A_0 \q^2$, where $E$ denotes the Young's modulus. The gravitational body force vector per unit undeformed length is $b = [0 \;\; -\rho_0 A_0 g \;\; 0]^\top$, where $g$ is the gravitational acceleration. The Hamiltonian is thus expressed as: 
\begin{equation}
	H(x) = \int_{0}^{L_0} \left[ \frac{1}{2\rho_0 A_0} p^\top p + \frac{1}{2} E A_0 \q^2 - (d_0 + r)^\top b \right] ds.
	\notag
\end{equation}
To evaluate the energy conservation properties, the string is fixed at $s = 0$ via the Dirichlet boundary condition $v_D(0,t) = 0$, and left unconstrained at the free end $s = L_0$ ($\tau_N(L_0,t) = 0$). The initial displacements and strains are set to zero ($r(s,0) = 0$, $\q(s,0) = 0$). To induce the dynamic response, a linearly increasing initial transverse velocity is prescribed along the string, reaching $2$ [m/s] at the free end: $\dot{r}_x(s,0) = 2s/L_0$. This velocity condition is imposed through the initial generalized momentum density as $p_x(s,0) = \rho_0 A_0 \dot{r}_x(s,0)$, with all other initial momentum components ($p_y, p_z$) set to zero. Under these conditions, the Hamiltonian must remain constant at the initial value: 
\begin{equation}
	H(x)\big|_{t=0} = \int_{0}^{L_0} \left[ \frac{1}{2 \rho_0 A_0} p_x(s,0)^2 -   d_0(s)^\top b \, \right] ds .
	\notag
\end{equation}
The geometric and physical parameters of the string are: length $L_0 = 1.5$ [m], cross-sectional area $A_0 = 1 \times 10^{-6}$ [m$^2$], density $\rho_0 = 2800$ [kg/m$^3$], Young's modulus $E = 2$ [MPa], and gravitational acceleration $g = 9.806$ [m/s$^2$]. The reference configuration is defined as a straight line oriented in space, given by $d_0(s) = s \frac{\sqrt{2}}{2} [1 \;\; -1 \;\; 0]^\top$. The time integration is performed using the implicit midpoint rule with a constant time step $\Delta t = 5 \times 10^{-3}$ [s].

The discrete lifted jet-bundle formulation (Theorem \ref{theo:discrete_lifted_jet}) is evaluated using continuous linear shape functions (P1) on uniform meshes of $n_e = 3$ and $n_e = 10$ elements, performing full integration. Figure \ref{fig:string_snapshots_M1} presents the spatial configuration of the string at five distinct time instances for both mesh densities. 

\begin{figure}[b!]
	\centering
	\hspace*{-1.5mm}\begin{tabular}{ccccc}
		\includegraphics[width=0.18\textwidth]{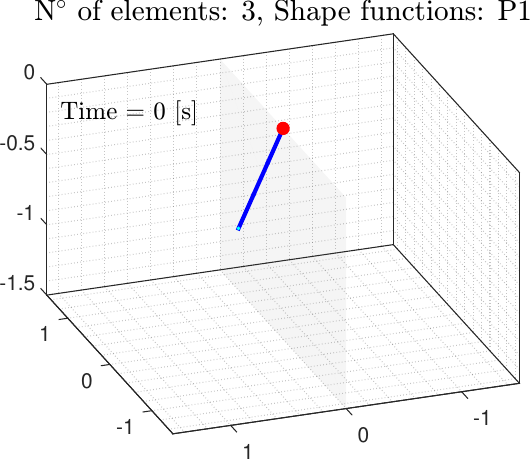} & \includegraphics[width=0.18\textwidth]{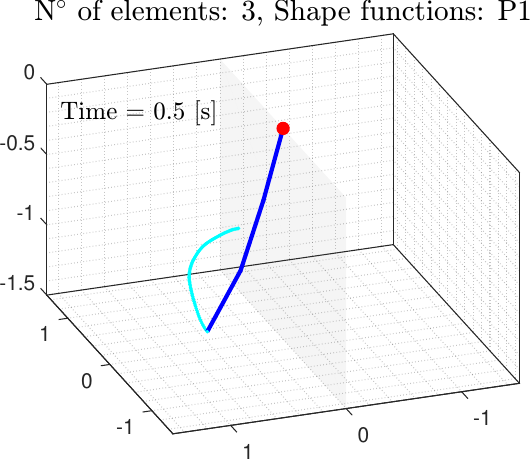} & \includegraphics[width=0.18\textwidth]{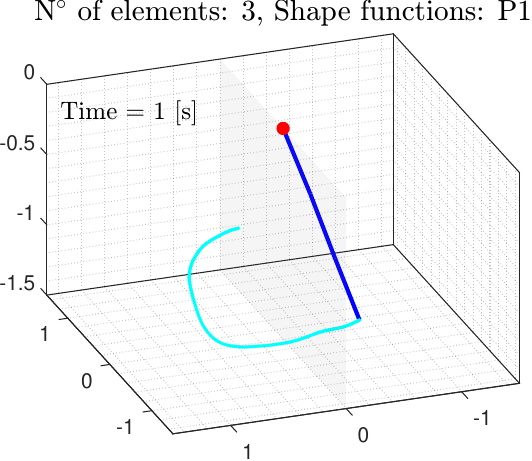} &
		\includegraphics[width=0.18\textwidth]{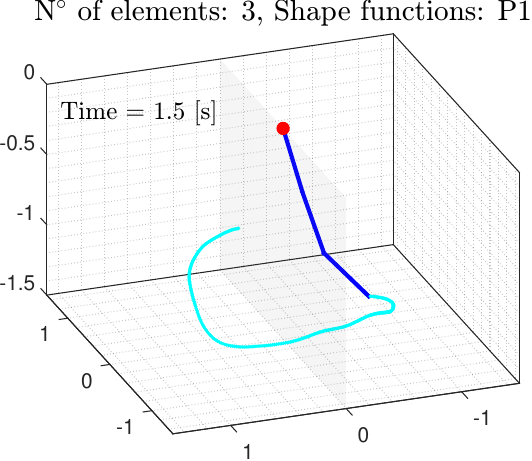} & \includegraphics[width=0.18\textwidth]{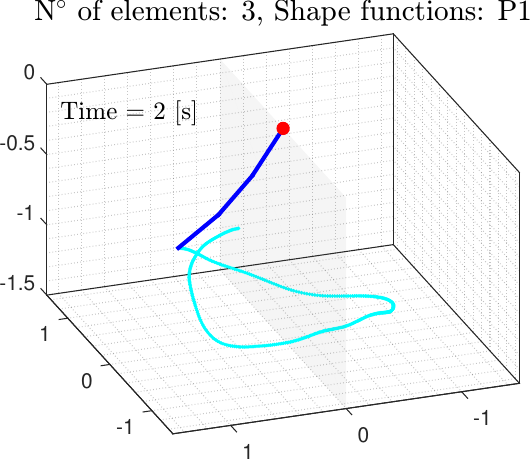} \\
		\includegraphics[width=0.18\textwidth]{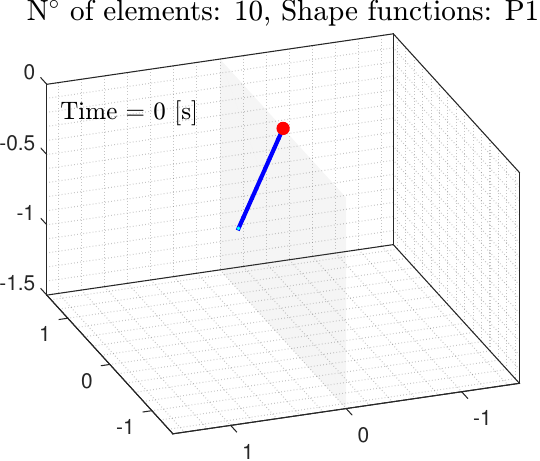} & \includegraphics[width=0.18\textwidth]{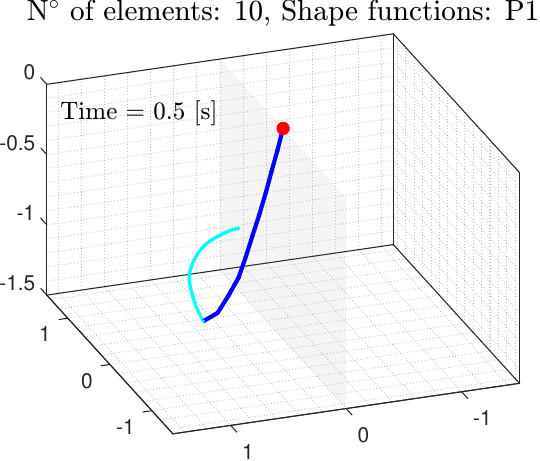} & \includegraphics[width=0.18\textwidth]{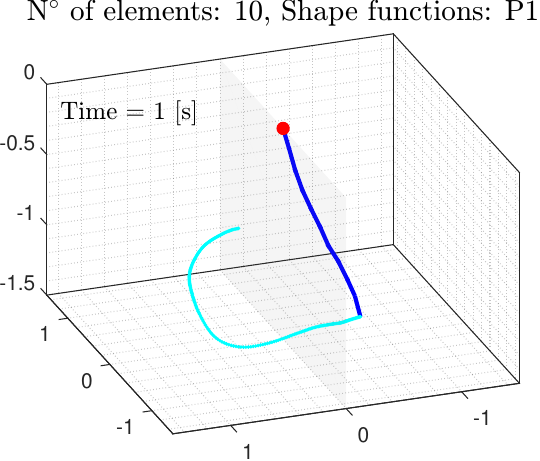} &
		\includegraphics[width=0.18\textwidth]{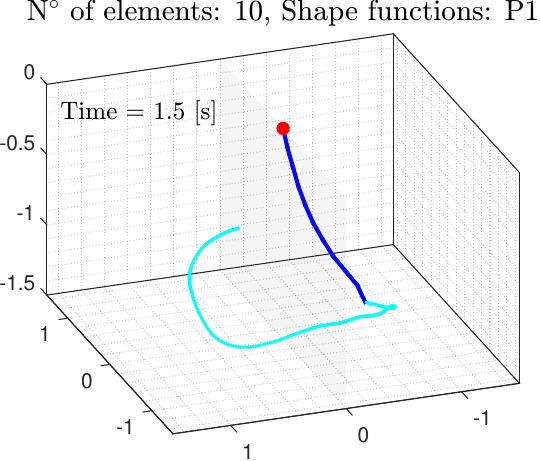} & \includegraphics[width=0.18\textwidth]{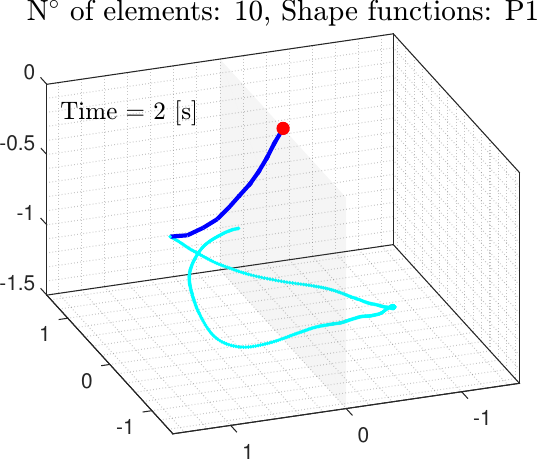}\\[-3mm]
	\end{tabular}
	\caption{Spatial configuration of the string computed via the discrete lifted jet-bundle formulation (Theorem \ref{theo:discrete_lifted_jet}).}
	\label{fig:string_snapshots_M1}
\end{figure}

The time evolution of the individual system energies, i.e., elastic ($\hat{E}_{elas}$), kinetic ($\hat{E}_{kin}$), and gravitational ($\hat{E}_g$), is presented in Fig. \ref{fig:string_energy_M1}. The latter two start with a non-zero value corresponding to the prescribed initial velocity and initial configuration. For all time, the sum of these energies remains constant, conserving the total initial energy independent of the number of elements used for the discretization.

\begin{figure}[t]
	\centering
	\begin{tabular}{cc}
		\includegraphics[width=0.45\textwidth]{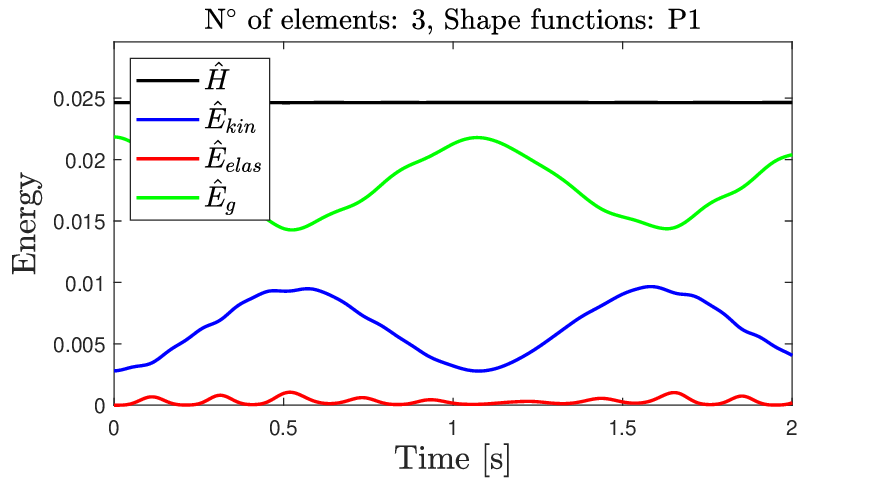} & \includegraphics[width=0.45\textwidth]{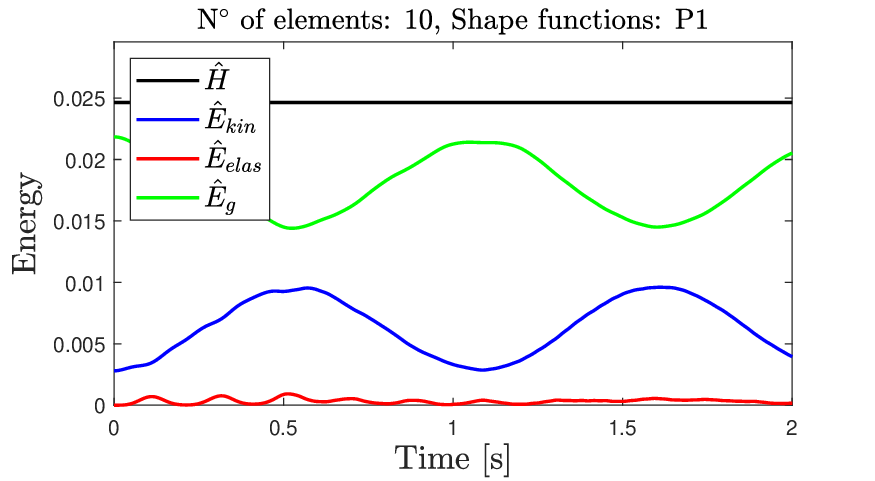}\\[-3mm]
	\end{tabular}
	\caption{Energy evolution of the string computed via the discrete lifted jet-bundle formulation (Theorem \ref{theo:discrete_lifted_jet}).}
	\label{fig:string_energy_M1}
\end{figure}

To validate the mixed formulation, the identical dynamic test is performed using the discrete lifted Stokes-Dirac structure (Theorem \ref{theo:discrete_lifted_stokes}). The string is discretized using P1-P0 elements (continuous linear displacements and discontinuous piecewise constant strains/stresses) on uniform meshes of $n_e = 3$ and $n_e = 10$ elements. Figure \ref{fig:string_snapshots_M2} displays the spatial configuration of the string at the same five time instances, while Fig. \ref{fig:string_energy_M2} presents the corresponding time evolution of the individual and total system energies.

\begin{figure}[b]
	\centering
		\hspace*{-1.5mm}\begin{tabular}{ccccc}
		\includegraphics[width=0.18\textwidth]{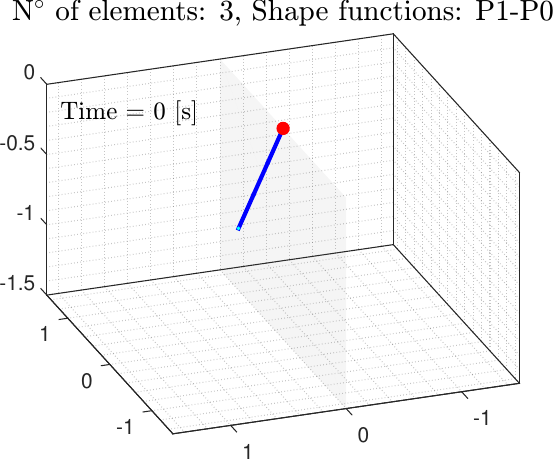} & \includegraphics[width=0.18\textwidth]{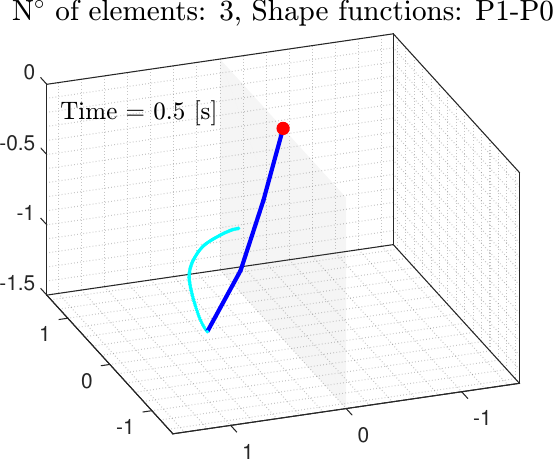} & \includegraphics[width=0.18\textwidth]{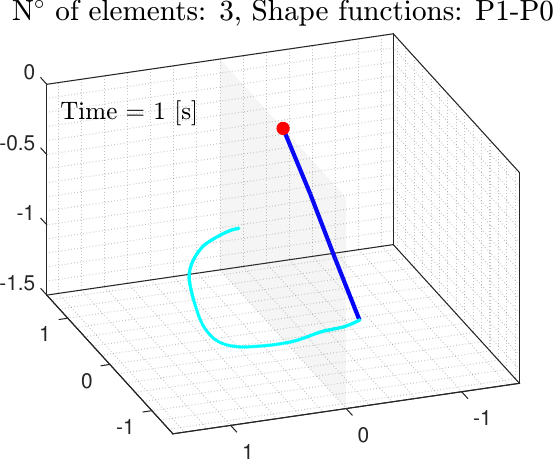} &
		\includegraphics[width=0.18\textwidth]{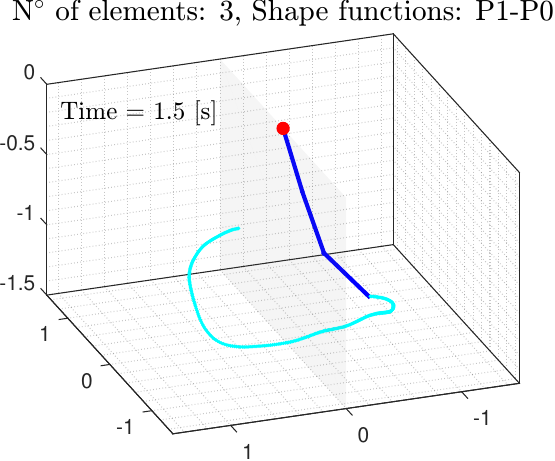} & \includegraphics[width=0.18\textwidth]{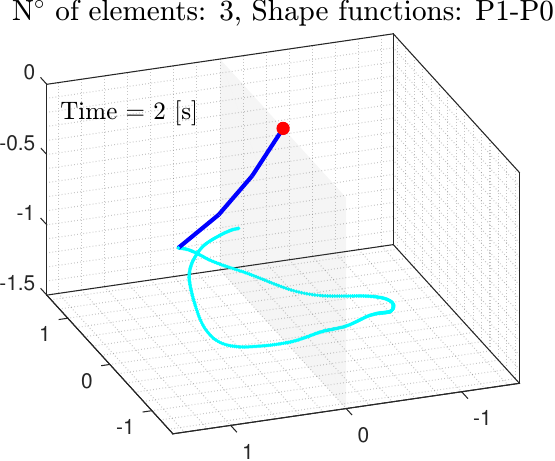} \\
		\includegraphics[width=0.18\textwidth]{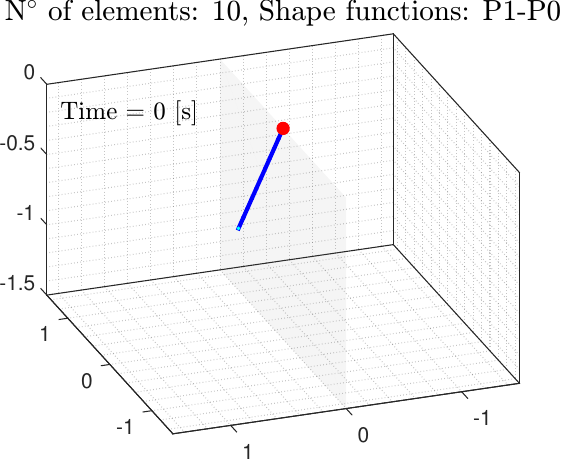} & \includegraphics[width=0.18\textwidth]{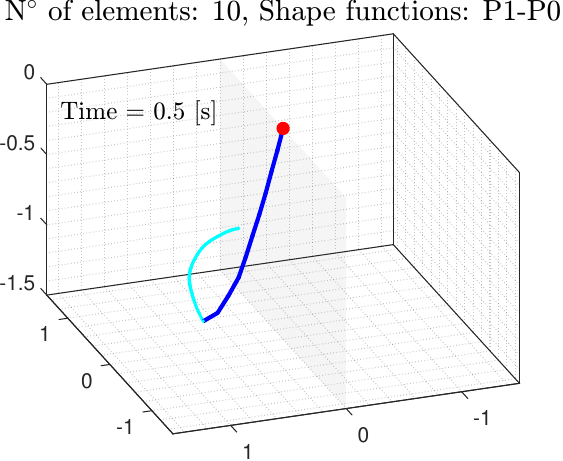} & \includegraphics[width=0.18\textwidth]{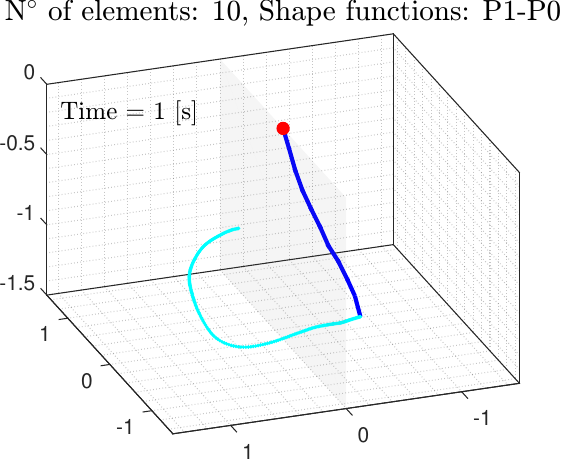} &
		\includegraphics[width=0.18\textwidth]{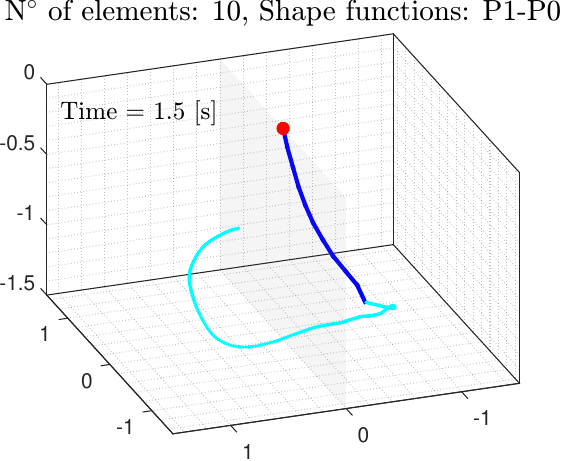} & \includegraphics[width=0.18\textwidth]{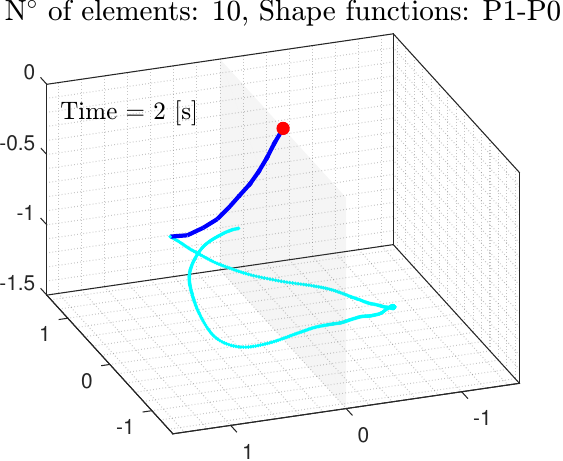}\\[-2mm]
	\end{tabular}
		\caption{Spatial configuration of the string computed via the discrete lifted Stokes-Dirac formulation (Theorem \ref{theo:discrete_lifted_stokes}).}
	\label{fig:string_snapshots_M2}
\end{figure}


\begin{figure}[b]
	\centering
		\begin{tabular}{cc}
		\includegraphics[width=0.45\textwidth]{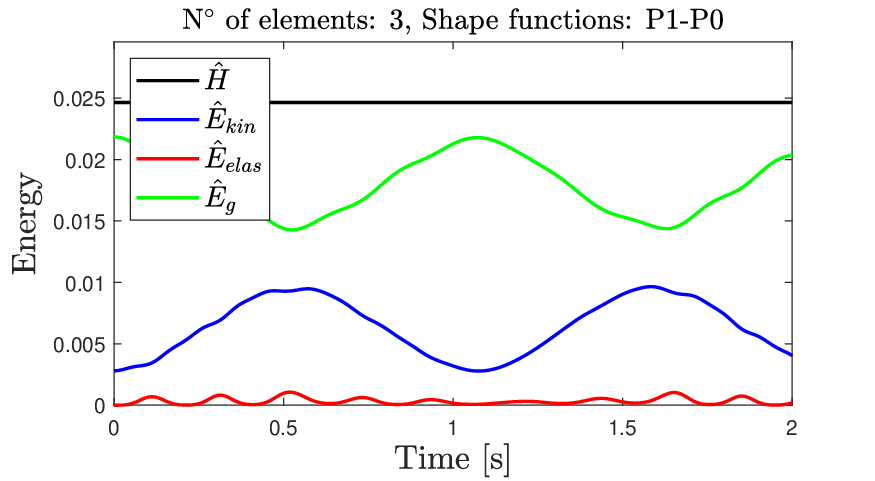} & \includegraphics[width=0.45\textwidth]{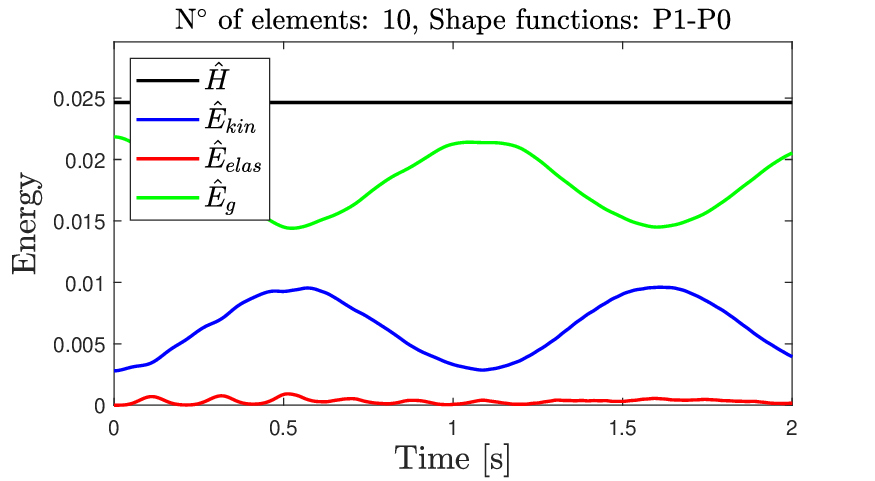}\\[-3mm]
	\end{tabular}
	\caption{Energy evolution of the string computed via the discrete lifted Stokes-Dirac formulation (Theorem \ref{theo:discrete_lifted_stokes}).}
	\label{fig:string_energy_M2}
\end{figure}

As observed in Fig. \ref{fig:string_energy_M2}, the sum of the elastic, kinetic, and gravitational potential energies remains constant over time, conserving the total initial energy of the system independent of the spatial discretization density. Comparing the results from both approaches, the discrete models derived from Theorem \ref{theo:discrete_lifted_jet} and Theorem \ref{theo:discrete_lifted_stokes} behave consistently, capture the same dynamics, and yield equivalent energy-preserving responses.

\subsection{Dynamic response under Dirichlet boundary velocities}

\begin{figure*}[t!]
	\begin{center}
		\includegraphics[scale=0.78]{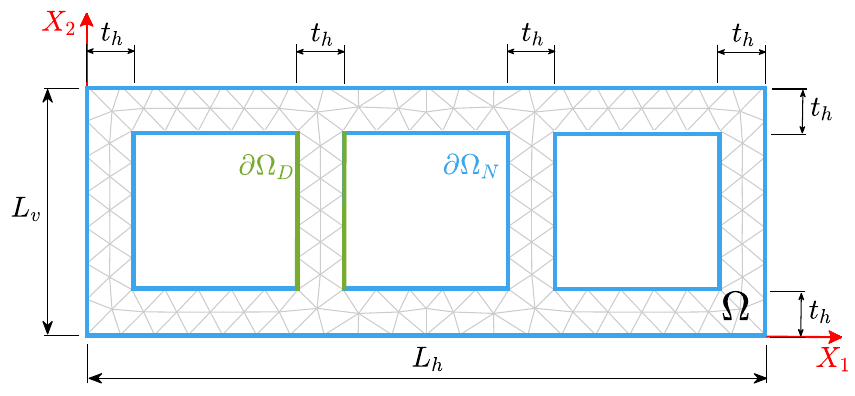}
	\end{center}
	\vspace{-3mm}\caption{Two dimensional frame.}
	\label{fig:Frame2D}
\end{figure*}

The third numerical example evaluates the proposed formulations under imposed Dirichlet boundary velocities. For this purpose, a geometrically nonlinear two-dimensional frame with compressible Neo-Hookean material is used. Gravitational forces are omitted to isolate the dynamic effects induced by the prescribed boundary motion.

Let $\Omega \subset \mathbb{R}^2$ denote the reference domain of the frame with constant thickness $h$, as shown in Fig. \ref{fig:Frame2D}. The spatial material coordinates are denoted by $\X = [X_1 \;\; X_2]^\top$ and the displacement field by $r(\X,t) = [u_1(\X,t) \;\; u_2(\X,t)]^\top$. To capture large deformation kinematics, the non-zero components of the Green-Lagrange strain tensor are mapped into the generalized strain vector $\q(\X,t) = [\q_1 \;\; \q_2 \;\; \q_3]^\top$. The strain vector defined in terms of the spatial derivatives of $r(\X,t)$ is given by:
\begin{equation}
	\q(r) = \begin{bmatrix}
		\p_1 u_1 + \frac{1}{2}(\p_1 u_1)^2 + \frac{1}{2}(\p_1 u_2)^2 \\
		\p_2 u_2 + \frac{1}{2}(\p_2 u_1)^2 + \frac{1}{2}(\p_2 u_2)^2 \\
		\p_2 u_1 + \p_1 u_2 + (\p_1 u_1 \p_2 u_1) + (\p_1 u_2 \p_2 u_2)
	\end{bmatrix}.
	  \notag
\end{equation}

The state-dependent differential operator $\mathcal{F}_\X(r)$, such that $\dot{\q} = \mathcal{F}_\X(r)\dot{r}$, is given by:
\begin{equation}
	\mathcal{F}_\X(r) = \begin{bmatrix}
		\p_1 + \p_1 u_1 \p_1 & \;\; \p_1 u_2 \p_1 \\
		\p_2 u_1 \p_2 & \;\;  \p_2 + \p_2 u_2 \p_2 \\
		\p_2 + \p_2 u_1 \p_1 + \p_1 u_1 \p_2 & \;\; \p_1 + \p_2 u_2 \p_1 + \p_1 u_2 \p_2
	\end{bmatrix}.
	\notag
\end{equation}
The compressible Neo-Hookean material is characterized by the first and third invariants of the right Cauchy-Green deformation tensor, which are expressed in terms of the generalized strains as $I_C(\q) = 2\q_1 + 2\q_2 + 3$ and $III_C(\q) = (2\q_1+1)(2\q_2+1) - \q_3^2$. Then, the generalized strain energy density function $\Psi(\q)$ is given by:
\begin{equation}
	\Psi(\q) = h \left( \frac{\mu_L}{2} \left[ I_C(\q) - 3 - \ln(III_C(\q)) \right] + \frac{\lambda_L}{2} \left(\sqrt{III_C(\q)} - 1 \right)^2 \right),
\end{equation}
where $\mu_L$ and $\lambda_L$ are the Lamé constants of the material, defined in terms of Young's modulus $E$ and Poisson's ratio $\nu$ as:
\begin{equation}
	\mu_L = \frac{E}{2(1+\nu)}, \qquad \lambda_L = \frac{E\nu}{(1+\nu)(1-2\nu)}.
	\notag
\end{equation}

The Hamiltonian $H(x)$ encapsulates the kinetic and elastic strain energies of the system as:
\begin{equation}
	H(x) = \int_{\Omega} \left[ \frac{1}{2\rho_0 h} p^\top p + \Psi(\q) \right] d\X,
	\notag
\end{equation}
where $p = [p_1 \;\; p_2]^\top = \rho_0 h \dot{r}$ is the generalized momentum density, with $\rho_0$ the material density.

The system is initialized from a state of rest, with zero initial displacements, strains, and momenta ($x|_{t=0} = 0$). The Neumann boundaries are left free of loads ($\tau_N(\SX,t) = 0$). To induce motion, a time-dependent velocity profile $v_D(\SX,t) = v_0 \sin(2\pi f t)$, with amplitude $v_0 = 0.2$ [m/s] and frequency $f = 10$ [Hz], is imposed at the Dirichlet boundary for the interval $t \le 0.5$ [s]. For $t > 0.5$ [s], this boundary is fixed, setting $v_D(\SX,t) = 0$ for the remainder of the simulation. To address the higher computational demands of the two-dimensional nonlinear model, two time integration schemes are evaluated and compared: the explicit Störmer-Verlet method and the implicit midpoint rule. Both schemes are implemented with a constant time step $\Delta t = 2 \times 10^{-4}$ [s]. For details on the implementation of these integrators in the context of port-Hamiltonian elastodynamics, the reader is referred to \cite[Appendix B]{ponce2025port}.

The physical and geometric parameters of the frame are: overall length $L_h = 30$ [cm], overall height $L_v = 11$ [cm], strut width $t_h = 2$ [cm], out-of-plane thickness $h = 10$ [mm], density $\rho_0 = 1000$ [kg/m$^3$], Young's modulus $E = 50$ [kPa], and Poisson's ratio $\nu = 0.4$. The spatial domain is discretized using a single unstructured triangular mesh (Fig. \ref{fig:Frame2D}). The discrete lifted jet-bundle formulation (Theorem \ref{theo:discrete_lifted_jet}) is evaluated using continuous linear elements (P1), while the discrete lifted Stokes-Dirac formulation (Theorem \ref{theo:discrete_lifted_stokes}) utilizes the mixed P1-P0 pairing. 

Figure \ref{fig:frame_snapshots} presents the spatial configuration of the frame computed via the formulation from Theorem \ref{theo:discrete_lifted_jet} using the implicit midpoint rule. The dynamics produced by the mixed formulation from Theorem \ref{theo:discrete_lifted_stokes} yield equivalent responses and are therefore omitted. The time evolution of the system energies is presented in Fig. \ref{fig:frame_energies}, comparing the performance of the explicit Störmer-Verlet and the implicit midpoint rule.

\begin{figure}[b]
	\centering
	\hspace*{-1.5mm}\begin{tabular}{ccccc}
		\includegraphics[width=0.17\textwidth]{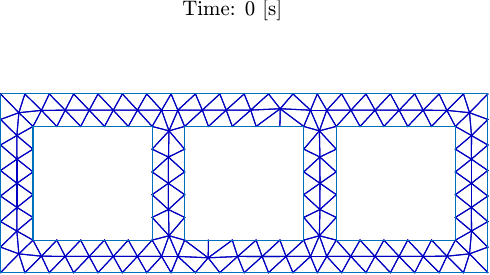} & \includegraphics[width=0.17\textwidth]{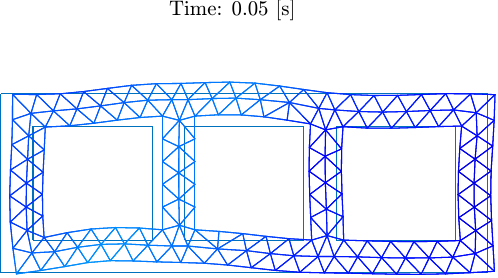} & \includegraphics[width=0.17\textwidth]{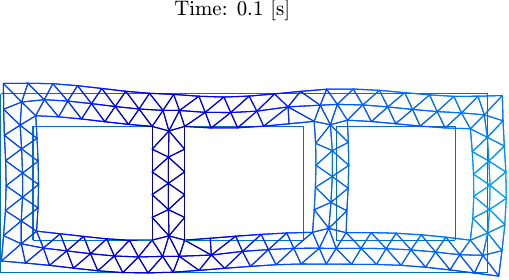} &
		\includegraphics[width=0.17\textwidth]{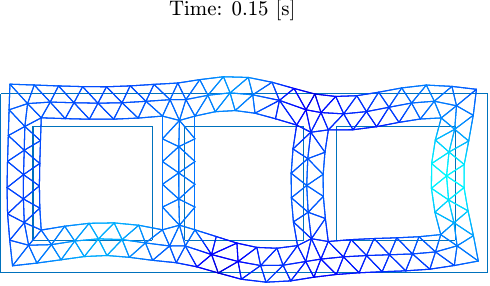} & \includegraphics[width=0.17\textwidth]{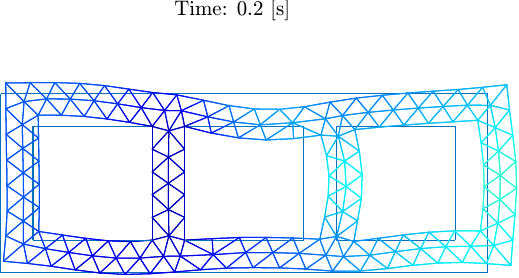} \\
		\includegraphics[width=0.17\textwidth]{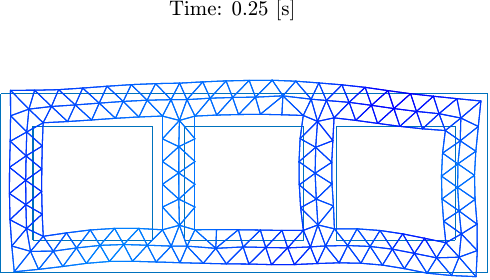} & \includegraphics[width=0.17\textwidth]{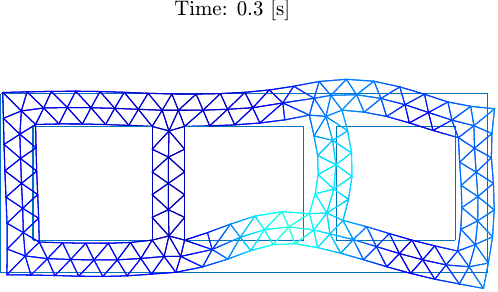} & \includegraphics[width=0.17\textwidth]{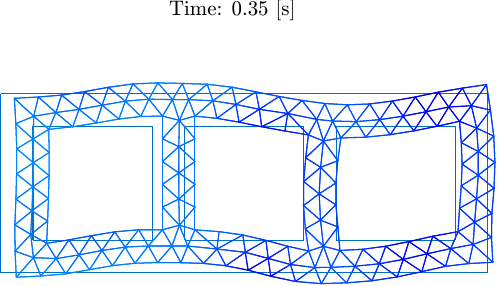} &
		\includegraphics[width=0.17\textwidth]{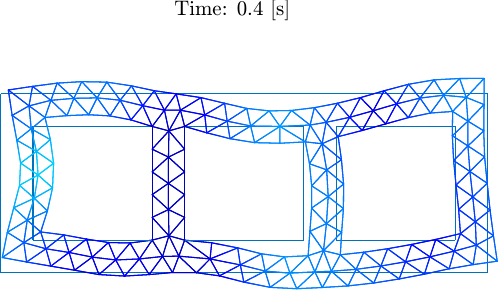} & \includegraphics[width=0.17\textwidth]{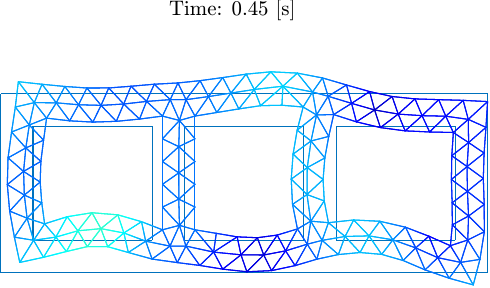}\\
		\includegraphics[width=0.17\textwidth]{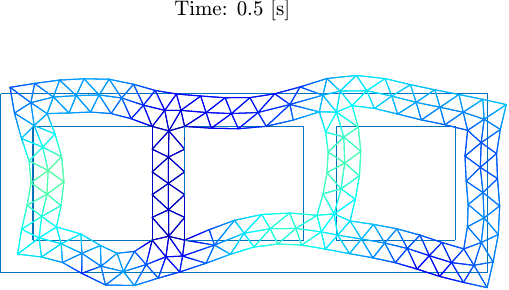} & \includegraphics[width=0.17\textwidth]{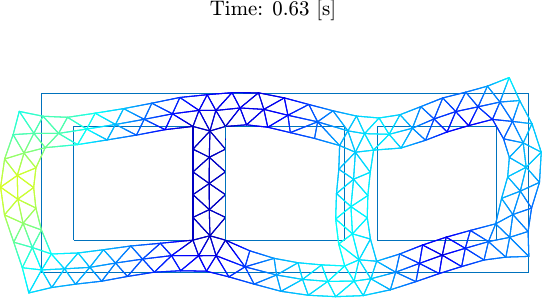} & \includegraphics[width=0.17\textwidth]{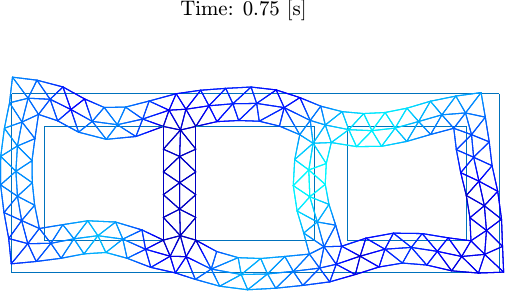} &
		\includegraphics[width=0.17\textwidth]{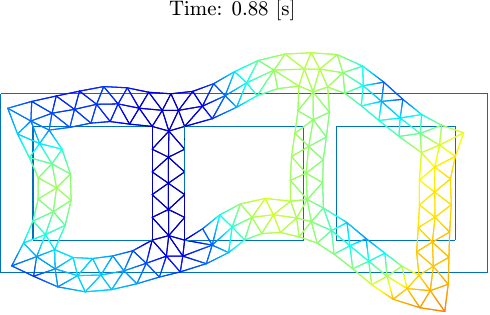} & \includegraphics[width=0.17\textwidth]{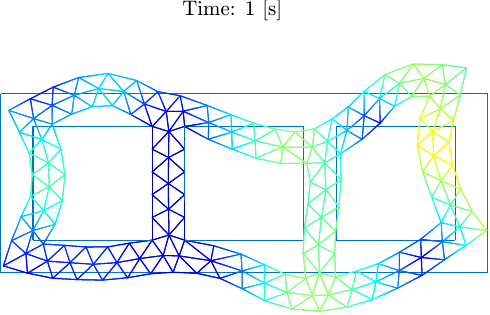}\\
	\end{tabular}
	\includegraphics[width=0.3\textwidth]{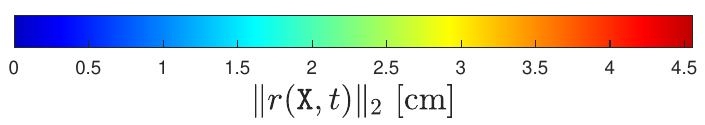}\\[-2mm]
	\caption{Dynamic configuration of the 2D Neo-Hookean frame.}
	\label{fig:frame_snapshots}
\end{figure}

\begin{figure}[b]
	\centering
	\begin{subfigure}[b]{0.48\textwidth}
		\centering
		\includegraphics[width=0.95\textwidth]{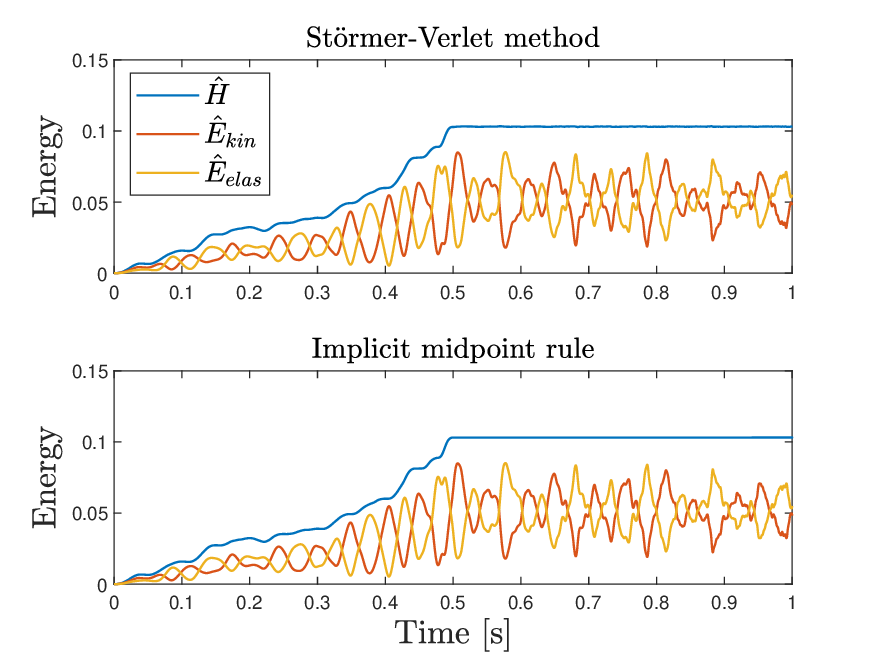}\\[-2mm]
		\caption{Model from Theorem \ref{theo:discrete_lifted_jet} (P1)}
		\label{fig:energy_T1}
	\end{subfigure}
	\hfill
	\begin{subfigure}[b]{0.48\textwidth}
		\centering
		\includegraphics[width=0.95\textwidth]{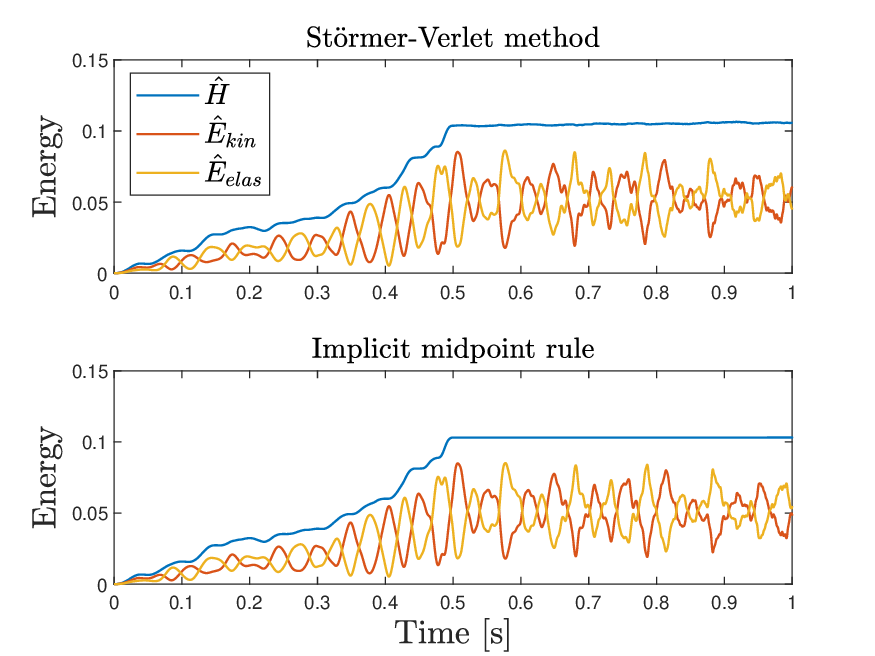}\\[-2mm]
		\caption{Model from Theorem \ref{theo:discrete_lifted_stokes} (P1-P0)}
		\label{fig:energy_T2}
	\end{subfigure}
	\caption{Energy evolution of the 2D frame comparing the explicit Störmer-Verlet and implicit midpoint rule integrators.}
	\label{fig:frame_energies}
\end{figure}

As expected from the imposed velocity condition, it is observed from Fig. \ref{fig:frame_snapshots} that the entire structure displaces diagonally, a motion that initiates at the Dirichlet boundary $\partial\Omega_D$ and subsequently propagates throughout the rest of the domain.

Regarding the energy response, as observed in Fig. \ref{fig:frame_energies}, external power is supplied to the system through the Dirichlet boundary during the first 0.5 [s]. For $t > 0.5$ [s], the boundary is fixed and the system transitions to a conservative regime. In this phase, the implicit midpoint rule maintains an apparently constant discrete Hamiltonian $\hat{H}$, whereas the explicit Störmer-Verlet method exhibits more appreciable oscillations. However, this improved energy conservation comes with a higher computational cost. Since the Störmer-Verlet method is an explicit scheme, it evaluates the nonlinear terms only once per time step, whereas the implicit midpoint rule requires solving a nonlinear system of equations at every increment. For the model based on Theorem \ref{theo:discrete_lifted_jet} (P1), the simulation required 10 [s] using the explicit scheme and 30 [s] using the implicit scheme. For the mixed model based on Theorem \ref{theo:discrete_lifted_stokes} (P1-P0), the explicit scheme took 110 [s], while the implicit scheme required 620 [s]. It is worth mentioning that these computational times were obtained using a sub-optimal custom MATLAB code executed on a standard laptop, as the implementation was intended primarily to illustrate the methodology.

\section{Conclusion and future work}  \label{sec:Conclusion}

This paper presented a continuous-to-discrete kinematic lifting framework to strongly impose Dirichlet boundary velocities in port-Hamiltonian models of elastodynamics. Through an additive decomposition of the displacement and velocity fields at the continuous level, the framework maps boundary actuation effects into the interior domain via a distributed port. This methodology yields finite-dimensional port-Hamiltonian systems with an ordinary differential equation (ODE) topology, bypassing differential-algebraic equations (DAEs) and penalty parameters. Numerical simulations verify the structure-preserving properties of the discrete models, demonstrating exact energy conservation, locking-free spatial accuracy in mixed formulations, and accurate dynamic responses under time-varying Dirichlet boundary velocities.

A practical contribution of this work is the identification of the structural connection between continuous kinematic lifting and usual computational mechanics practices. Treating the velocities of the degrees of freedom at the Dirichlet boundaries as prescribed inputs induces an algebraic partitioning of the consistent mass matrix. This algebraic reduction of classical or mixed finite element schemes is mathematically equivalent to the structure-preserving discretization of a lifted PHS models with specific interpolation shape functions. This establishes a systematic pathway to transition finite element formulations into port-Hamiltonian structures with strongly imposed Dirichlet boundary velocities, eliminating the requirement to construct explicit lifting functions.

Future work will investigate the control-theoretic and analytical properties of the distributed interior lifting velocity $v_\Omega$. Exploring how this energy-neutral distributed port can be used for control design, internal state observers, or trajectory optimization constitutes the next step in extending this port-Hamiltonian elastodynamics framework.

\begin{ack}                               
\noindent The authors acknowledge support from Chilean ANID projects: Fondecyt Postdoctorado 3260380, Fondecyt Regular 1231896, ECOS 220040, CIA250006; the French EIPHI Graduate School ANR-17-EURE-0002, the CNRS funded IRP SMARTER project; and the European MSCA Project MODCONFLEX 101073558. 
\end{ack}

\section*{Code and data availability}
The MATLAB source code and the datasets generated during the current study, which are required to reproduce all numerical examples, are publicly available in the GitHub repository at \url{https://cponces.github.io/Strong_imposition_of_Dirichlet_BC/}.

\section*{Declaration of generative AI and AI-assisted technologies in the manuscript preparation process}

During the preparation of this work the authors used Google's Gemini large language model (LLM) in order to improve wording, punctuation, and grammar. After using this tool, the authors reviewed and edited the content as needed and take full responsibility for the content of the published article.

\appendix

\section{Proofs of Section \ref{sec:Continuos}}

\subsection{Proof of Proposition \ref{prop:varPPLE1}}  \label{app:proof_Prop1}

\begin{proof}
	Demonstrating that the proposed scheme is variationally consistent for jet-bundle port-Hamiltonian elastodynamics is equivalent to proving that the stationarity condition $\delta \mathcal{P}_{\HP} = 0$ with respect to the independent variations $(\delta\dot{p}, \delta\dot{r})$  reproduces the system's governing dynamics.
	
	\textit{Step 1.} The variation of the functional requires evaluating the Lagrangian rate density $\dot{\mathfrak{L}}_{pr}$. Expanding the energy terms from Definition \ref{def:energies} and applying the chain rule to the strain energy rate $\dot{\Psi} = e_\q(\q(r))^\top \dot{\q}$ where $\dot{\q} = \mathcal{F}_\X(r)\dot{r}$, we obtain:
	\begin{equation}
		\dot{\mathfrak{L}}_{pr}(p,r,\dot{p},\dot{r}) = (\mathcal{M}^{-1}p)^\top \dot{p} - e_\q(\q(r))^\top \mathcal{F}_\X(r)\dot{r} + b^\top \dot{r}. \notag
	\end{equation}
	
	\textit{Step 2.} Substituting this rate into the Hamilton-Pontryagin-based functional and setting the total variation $\delta \mathcal{P}_{\HP} = \delta_{\dot{p}} \mathcal{P}_{\HP} + \delta_{\dot{r}} \mathcal{P}_{\HP} = 0$ yields two independent stationary conditions. The variation with respect to the momentum rate $\delta\dot{p}$ gives:
	\begin{equation}
		\delta_{\dot{p}} \mathcal{P}_{\HP} = \int_\Omega \delta\dot{p}^\top \left( \dot{r} - \mathcal{M}^{-1}p \right) d\X = 0, \notag
	\end{equation}
	which enforces the relation $\dot{r} = \mathcal{M}^{-1}p$ in $\Omega$.
	
	\textit{Step 3.} Simultaneously, evaluating the variation with respect to the velocity $\delta\dot{r}$ yields:
	\begin{equation}
		\delta_{\dot{r}} \mathcal{P}_{\HP} = \int_\Omega \left[ \delta\dot{r}^\top \dot{p} + e_\q(\q(r))^\top \mathcal{F}_\X(r)\delta\dot{r} - \delta\dot{r}^\top b \right] d\X - \int_{\p\Omega_N} \!\!\! \delta\dot{r}^\top \tau_N \, d\SX = 0. \notag
	\end{equation}
	To factor out the arbitrary variation $\delta\dot{r}$, Lemma \ref{lemma:integration} is applied to the term involving the differential operator $\mathcal{F}_\X(r)$:
	\begin{equation}
		\int_\Omega e_\q(\q(r))^\top \mathcal{F}_\X(r)\delta\dot{r} \, d\X = \int_\Omega \big(\mathcal{F}_\X(r)^*e_\q(\q(r))\big)^\top \delta\dot{r} \, d\X + \int_{\p\Omega} \delta\dot{r}^\top F_\p(r) e_\q(\q(r)) \, d\SX. \notag
	\end{equation}
	Substituting this identity back and grouping terms yields:
	\begin{equation}
		\int_\Omega \left[ \dot{p} + \mathcal{F}_\X(r)^*e_\q(\q(r)) - b \right]^\top \delta\dot{r} \, d\X - \int_{\p\Omega_N} \!\!\! \left[ \tau_N - F_\p(r) e_\q(\q(r)) \right]^\top \delta\dot{r} \, d\SX + \int_{\p\Omega_D} \!\!\! \big(F_\p(r) e_\q(\q(r))\big)^\top \delta\dot{r} \, d\SX = 0. \notag
	\end{equation}
	By hypothesis,  $\delta\dot{r} = 0$ on $\p\Omega_D$, causing the last boundary integral above to vanish. Since $\delta\dot{r}$ is arbitrary both within $\Omega$ and on $\p\Omega_N$, we deduce the dynamic momentum balance $\dot{p} = -\mathcal{F}_\X(r)^*e_\q(\q(r)) + b$ in $\Omega$ and the corresponding Neumann boundary condition $\tau_N = F_\p(r) e_\q(\q(r))$ on $\p\Omega_N$.
	
	\textit{Step 4.} The state vector is defined as $z = [p^\top \; r^\top]^\top$. The co-energy variables $\var_z H = [e_p^\top \; e_r^\top]^\top$ are defined through the first variation of the Hamiltonian functional $H(z)$ in the direction of the arbitrary variations $\delta z = [\delta p^\top \; \delta r^\top]^\top$:
	\begin{equation}
		\delta H = \int_\Omega \left( \delta p^\top \delta_p H + \delta r^\top \delta_r H \right) d\X. \notag
	\end{equation}
	Computing the first variation of $H$ yields:
	$$
		\delta H = \int_{\Omega} \left[ \delta p^\top (\mathcal{M}^{-1}p) + \big(\mathcal{F}_\X(r)\delta r\big)^\top e_\q(\q(r)) - \delta r^\top b \right] d\X. 
	$$
	Applying Lemma \ref{lemma:integration} to isolate $\delta r$ gives:
	\begin{equation}
		\delta H = \int_{\Omega} \left[ \delta p^\top (\mathcal{M}^{-1}p) + \delta r^\top \big( \mathcal{F}_\X(r)^*e_\q(\q(r)) - b \big) \right] d\X + \int_{\p\Omega} \delta r^\top F_\p(r) e_\q(\q(r)) \, d\SX. \notag
	\end{equation}
	By matching the terms in the integral over $\Omega$ with the definition of co-energy variables, they are  identified as:
	\begin{align}
		e_p(p) = \delta_p H = \mathcal{M}^{-1}p, \notag  & &
		e_r(r) = \delta_r H = \mathcal{F}_\X(r)^*e_\q(\q(r)) - b. \notag
	\end{align}
	Substituting these co-energies into the dynamic and kinematic equations derived in Steps 2 and 3  recovers the  structure presented in Definition \ref{def:jet_bundle_phs}. Lastly, through the specification of the power-conjugate boundary ports in accordance with Definition \ref{def:boundaryPorts}, this dynamic system defines a jet-bundle PHS.
\end{proof}

\subsection{Proof of Proposition \ref{prop:lifted_jet_bundlePHS}}  \label{app:proof_Prop2}

\begin{proof}
	\textit{Step 1.} The functional $\mathcal{P}_{\HP}$ is evaluated by substituting the velocity decomposition $\dot{r} = \dot{r}_r + v_L$. Since the lifting velocity $v_L$ is  prescribed, its independent variation is zero ($\delta v_L = 0$). The stationary condition $\delta \mathcal{P}_{\HP} = 0$ is:
	\begin{equation}
		\delta \mathcal{P}_{\HP} = \int_\Omega \left[ \delta\dot{p} \cdot (\dot{r}_r + v_L) + \dot{p} \cdot \delta\dot{r}_r - \delta\dot{\mathfrak{L}}_{pr} \right] d\X - \int_{\p\Omega_N} \!\!\! \tau_N \cdot \delta\dot{r}_r \, d\SX = 0. \notag
	\end{equation}
	Taking the variation with respect to the momentum rate $\delta\dot{p}$ yields:
	\begin{equation}
		\delta_{\dot{p}} \mathcal{P}_{\HP} = \int_\Omega \delta\dot{p} \cdot \left( \dot{r}_r + v_L - \mathcal{M}^{-1}p \right) d\X = 0, \notag
	\end{equation}
	which enforces the  relation $\dot{r}_r = \mathcal{M}^{-1}p - v_L$ in $\Omega$. The lifting kinematics directly impose $\dot{r}_L = v_L$.

	\textit{Step 2.} The variation of the total potential energy rate density $\dot{\mathfrak{U}}_{rr}$ with respect to  $\dot{r}_r$ in the direction $\delta\dot{r}_r$ is computed via the G\^ateaux derivative. Recalling the strain energy rate $\dot{\Psi} = e_\q(\q(r))^\top \dot{\q}$ where $\dot{\q} = \mathcal{F}_\X(r)\dot{r}$, and applying the velocity decomposition $\dot{r} = \dot{r}_r + v_L$, we introduce the perturbation $\dot{r}_r \to \dot{r}_r + \varepsilon \delta\dot{r}_r$:
	\begin{align}
		\delta_{\dot{r}_r} \dot{\mathfrak{U}}_{rr} &= \lim_{\varepsilon \to 0} \frac{d}{d\varepsilon} \bigg[ e_\q(\q(r))^\top \mathcal{F}_\X(r)(\dot{r}_r + \varepsilon \delta\dot{r}_r + v_L) -  (\dot{r}_r + \varepsilon \delta\dot{r}_r + v_L)^\top b \bigg] \notag \\[1mm]
		&= \lim_{\varepsilon \to 0} \left[ e_\q(\q(r))^\top \mathcal{F}_\X(r)\delta\dot{r}_r - \delta\dot{r}_r ^\top b \right] \notag \\[1mm]
		&= e_\q(\q(r))^\top \mathcal{F}_\X(r)\delta\dot{r}_r - \delta\dot{r}_r ^\top b. \notag
	\end{align}
	Noting that the kinetic energy rate $\dot{\mathfrak{T}}_p$ is independent of $\dot{r}_r$ in this formulation, the variation of the Lagrangian rate density becomes $\delta_{\dot{r}_r} \dot{\mathfrak{L}}_{pr} = - \delta_{\dot{r}_r} \dot{\mathfrak{U}}_{rr}$. The first variation of the functional with respect to $\var \dot{r}_r$ is then given by:
	\begin{equation}
		\delta_{\dot{r}_r} \mathcal{P}_{\HP} = \int_\Omega \left[ \delta\dot{r}_r ^\top \dot{p} + e_\q(\q(r))^\top \mathcal{F}_\X(r)\delta\dot{r}_r - \delta\dot{r}_r ^\top b \right] d\X - \int_{\p\Omega_N} \!\!\! \delta\dot{r}_r ^\top \tau_N \, d\SX = 0. \notag
	\end{equation}
	To factor out the arbitrary variation $\delta\dot{r}_r$, Lemma \ref{lemma:integration} is applied to the term involving the differential operator $\mathcal{F}_\X(r)$:
	\begin{equation}
		\int_\Omega e_\q(\q(r))^\top \mathcal{F}_\X(r)\delta\dot{r}_r \, d\X = \int_\Omega \delta\dot{r}_r ^\top \mathcal{F}_\X(r)^*e_\q(\q(r)) \, d\X + \int_{\p\Omega} \delta\dot{r}_r^\top F_\p(r) e_\q(\q(r)) \, d\SX. \notag
	\end{equation}
	Substituting this identity back and grouping terms yields:
	\begin{equation}
		\int_\Omega  \delta\dot{r}_r ^\top \left[ \dot{p} + \mathcal{F}_\X(r)^*e_\q(\q(r)) - b \right]  d\X - \int_{\p\Omega_N} \!\!\! \delta\dot{r}_r ^\top \left[ \tau_N - F_\p(r) e_\q(\q(r)) \right]  d\SX + \int_{\p\Omega_D} \!\!\! \delta\dot{r}_r ^\top F_\p(r) e_\q(\q(r)) \, d\SX = 0. \notag
	\end{equation}
	Since the admissible variations satisfy $\delta\dot{r}_r \in \mathcal{V}_0$, we have $\delta\dot{r}_r|_{\p\Omega_D} = 0$, causing the integral over $\p\Omega_D$ to vanish. For arbitrary $\delta\dot{r}_r$ in $\Omega$ and on $\p\Omega_N$, the fundamental lemma of calculus of variations enforces $\dot{p} = -\mathcal{F}_\X(r)^*e_\q(\q(r)) + b = -e_r(r_r+r_L)$ in $\Omega$, and the Neumann boundary condition $\tau_N = F_\p(r) e_\q(\q(r))$ on $\p\Omega_N$.

	\textit{Step 3.} The augmented state vector is defined as $z = [p^\top \; r_r^\top \; r_L^\top]^\top$. The co-energy variables $\var_z H = [e_p^\top \; e_{r_r}^\top \; e_{r_L}^\top]^\top$ are  defined through the first variation of the Hamiltonian functional $\delta H$ in the direction of the arbitrary variations $\delta z = [\delta p^\top \; \delta r_r^\top \; \delta r_L^\top]^\top$. This variation is expressed as:
	\begin{equation}
		\delta H = \int_\Omega \left( \delta p^\top \delta_p H + \delta r_r^\top \delta_{r_r} H + \delta r_L^\top \delta_{r_L} H \right) d\X. \notag
	\end{equation}
	Computing the first variation of $H$ yields:
	\begin{align}
		\delta H &= \int_{\Omega} \left[ \delta p^\top (\mathcal{M}^{-1}p) + \delta\q^\top e_\q(\q(r)) - (\delta r_r + \delta r_L)^\top b \right] d\X \notag \\
		&= \int_{\Omega} \left[ \delta p^\top (\mathcal{M}^{-1}p) + \big(\mathcal{F}_\X(r)(\delta r_r + \delta r_L)\big)^\top e_\q(\q(r)) - (\delta r_r + \delta r_L)^\top b \right] d\X. \notag
	\end{align}
	Applying Lemma \ref{lemma:integration} to isolate the arbitrary variations $\delta r_r$ and $\delta r_L$ we obtain:
	\begin{align}
		\delta H =& \int_{\Omega} \left[ \delta p^\top (\mathcal{M}^{-1}p) + \delta r_r^\top \big( \mathcal{F}_\X(r)^*e_\q(\q(r)) - b \big) + \delta r_L^\top \big( \mathcal{F}_\X(r)^*e_\q(\q(r)) - b \big) \right] d\X + \int_{\p\Omega} (\delta r_r + \delta r_L)^\top F_\p(r) e_\q(\q(r)) \, d\SX. \notag
	\end{align}
	By matching the terms in the integral over $\Omega$ with the definition of co-energy variables, they are  identified as:
	\begin{align}
		e_p(p) = \delta_p H = \mathcal{M}^{-1}p, \notag & &
		e_{r_r}(r) = \delta_{r_r} H = \mathcal{F}_\X(r)^*e_\q(\q(r)) - b, \notag  & &
		e_{r_L}(r) = \delta_{r_L} H = \mathcal{F}_\X(r)^*e_\q(\q(r)) - b. \notag
	\end{align}
	This proves the equality $e_{r_r}(r) = e_{r_L}(r) = e_r(r_r + r_L)$. Substituting these co-energies into the dynamic and kinematic equations derived in Steps 1 and 2 recovers the  structure presented in Proposition \ref{prop:lifted_jet_bundlePHS}.  
	
	\textit{Step 4.} The boundary ports are defined by computing the time derivative of $H(z)$:
	\begin{equation}
		\dot{{H}} = \int_\Omega \left( e_p(p)^\top \dot{p} + e_\q(\q(r))^\top \dot{\q} - b^\top \dot{r} \right) d\X, \notag
	\end{equation}
	since $\dot{r} = \dot{r}_r + v_L = e_p(p)$, $\dot{\q} = \mathcal{F}_\X(r)\dot{r}$, and $\dot{p} = -\mathcal{F}_\X(r)^*e_\q(\q(r)) + b$, substituting:
	\begin{align*}
		\dot{{H}} =&  \int_\Omega \left[ e_p(p)^\top \big(-\mathcal{F}_\X(r)^*e_\q(\q(r)) + b\big) + e_\q(\q(r))^\top \mathcal{F}_\X(r)e_p(p) - b^\top e_p(p) \right] d\X  \notag \\
		= & \int_\Omega \left[ e_\q(\q(r))^\top \mathcal{F}_\X(r)e_p(p) - e_p(p)^\top \mathcal{F}_\X(r)^*e_\q(\q(r)) \right]  d\X \notag \\
		= & \int_{\p\Omega} e_p(p)^\top F_\p(r) e_\q(\q(r)) \, d\SX, \notag
	\end{align*}
	where Lemma \ref{lemma:integration} was applied. Recognizing that $e_p(p)|_{\p\Omega_D} = (\dot{r}_r + v_L)|_{\p\Omega_D} = v_D(\SX,t)$ since  $\dot{r}_r|_{\p\Omega_D} = 0$ and $v_L|_{\p\Omega_D} = v_D(\SX,t)$, the power balance is equivalently written as:
	\begin{equation}
		\dot{{H}} = \int_{\p\Omega} e_p(p)^\top F_\p(r) e_\q(\q(r)) \, d\SX = \int_{\p\Omega_N} v_N^\top \tau_N \, d\SX + \int_{\p\Omega_D} v_D^\top \tau_D \, d\SX \notag
	\end{equation}
	with the given definitions of boundary inputs and outputs ports.
\end{proof}

\subsection{Proof of Proposition \ref{prop:varPPLE2}}  \label{app:proof_Prop3}

\begin{proof}
	Demonstrating that the proposed scheme is variationally consistent for the Stokes-Dirac port-Hamiltonian representation is equivalent to proving that the stationarity condition $\delta \mathcal{P}_{\HW} = 0$ with respect to the independent variations $(\delta\dot{p}, \delta\dot{r}, \delta\dot{\q}, \delta e_\q)$  reproduces the system's governing dynamics.
	
	\textit{Step 1.} The variation of the functional requires evaluating the time derivative of the Lagrangian density $\dot{\mathfrak{L}}_{p\q r}$. Expanding the energy terms from Definition \ref{def:energies}, we obtain:
	\begin{equation}
		\dot{\mathfrak{L}}_{p\q r}(p,\q, \dot{p},\dot{\q}, \dot{r}) = (\mathcal{M}^{-1}p)^\top \dot{p} - \left(\frac{\partial \Psi(\q)}{\partial \q}\right)^\top \dot{\q} + b^\top \dot{r}. \notag
	\end{equation}
	\textit{Step 2.} Substituting this rate into the Hu-Washizu-based functional and setting the total variation $\delta \mathcal{P}_{\HW} = \delta_{\dot{p}} \mathcal{P}_{\HW} + \delta_{\dot{\q}} \mathcal{P}_{\HW} + \delta_{e_\q} \mathcal{P}_{\HW} + \delta_{\dot{r}} \mathcal{P}_{\HW} = 0$ yields four independent stationary conditions. Grouping the variations with respect to $\delta\dot{p}$, $\delta\dot{\q}$, and $\delta e_\q$ yields:
	\begin{align}
		\delta_{\dot{p}} \mathcal{P}_{\HW} &= \int_\Omega \delta\dot{p}^\top \left( \dot{r} - \mathcal{M}^{-1}p \right) d\X = 0, \notag \\
		\delta_{\dot{\q}} \mathcal{P}_{\HW} &= \int_\Omega \delta\dot{\q}^\top \left( \frac{\partial \Psi(\q)}{\partial \q} - e_\q \right) d\X = 0, \notag \\
		\delta_{e_\q} \mathcal{P}_{\HW} &= \int_\Omega -\delta e_\q^\top \left( \dot{\q} - \mathcal{F}_\X(r)\dot{r} \right) d\X = 0. \notag
	\end{align}
	Since these variations are arbitrary, they enforce in $\Omega$ the relations $\dot{r} = \mathcal{M}^{-1}p$,  $e_\q = \frac{\partial \Psi(\q)}{\partial \q}$, and $\dot{\q} = \mathcal{F}_\X(r)\dot{r}$.
	
	\textit{Step 3.} Simultaneously, evaluating the variation with respect to the velocity $\delta\dot{r}$ yields:
	\begin{equation}
		\delta_{\dot{r}} \mathcal{P}_{\HW} = \int_\Omega \left[ \delta\dot{r}^\top \dot{p} - \delta\dot{r}^\top b + e_\q(\q)^\top \mathcal{F}_\X(r)\delta\dot{r} \right] d\X - \int_{\p\Omega_N} \!\!\! \delta\dot{r}^\top \tau_N \, d\SX = 0. \notag
	\end{equation}
	To factor out $\delta\dot{r}$ from the differential operator $\mathcal{F}_\X(r)$, the integration Lemma \ref{lemma:integration} is applied:
	\begin{equation}
		\int_\Omega e_\q(\q)^\top \mathcal{F}_\X(r)\delta\dot{r} \, d\X = \int_\Omega \big(\mathcal{F}_\X(r)^*e_\q(\q)\big)^\top \delta\dot{r} \, d\X + \int_{\p\Omega} \delta\dot{r}^\top F_\p(r) e_\q(\q) \, d\SX. \notag
	\end{equation}
	Substituting this identity back into the velocity variation gives:
	\begin{equation}
		\delta_{\dot{r}} \mathcal{P}_{\HW} = \int_\Omega \left[ \dot{p} + \mathcal{F}_\X(r)^*e_\q(\q) - b \right]^\top \delta\dot{r} \, d\X - \int_{\p\Omega_N} \!\!\! \left[ \tau_N - F_\p(r)  e_\q(\q) \right]^\top \delta\dot{r} \, d\SX + \int_{\p\Omega_D} \!\!\! \big( F_\p(r) e_\q(\q) \big)^\top \delta\dot{r} \, d\SX = 0. \notag
	\end{equation}
	By hypothesis, $\delta\dot{r} = 0$ on $\p\Omega_D$, causing the final boundary integral to vanish. Since $\delta\dot{r}$ is arbitrary both within $\Omega$ and on $\p\Omega_N$, we deduce the dynamic momentum balance $\dot{p} = -\mathcal{F}_\X(r)^*e_\q(\q) + b$ in $\Omega$ and the corresponding Neumann boundary condition $\tau_N = F_\p(r) e_\q(\q)$ on $\p\Omega_N$. 
	
	\textit{Step 4.} The state vector is defined as $x = [p^\top \; \q^\top \; r^\top]^\top$. The co-energy variables $\var_x H = [e_p^\top \; e_\q^\top \; e_r^\top]^\top$ are defined through the first variation of the Hamiltonian functional $H(x)$ in the direction of the arbitrary variations $\delta x = [\delta p^\top \; \delta\q^\top \; \delta r^\top]^\top$:
	\begin{equation}
		\delta H = \int_\Omega \left( \delta p^\top \delta_p H + \delta\q^\top \delta_\q H + \delta r^\top \delta_r H \right) d\X. \notag
	\end{equation}
	Computing the first variation of the strictly algebraic Stokes-Dirac Hamiltonian directly yields:
	\begin{equation}
		\delta H = \int_{\Omega} \left[ \delta p^\top (\mathcal{M}^{-1}p) + \delta\q^\top \left(\frac{\partial \Psi(\q)}{\partial \q}\right) - \delta r^\top b \right] d\X. \notag
	\end{equation}
	By matching the terms in the integral over $\Omega$ with the definition of co-energy variables, they are  identified as:
	\begin{align}
		e_p(p) = \delta_p H = \mathcal{M}^{-1}p, \notag & &
		e_\q(\q) = \delta_\q H = \frac{\partial \Psi(\q)}{\partial \q}, \notag & &
		e_r = \delta_r H = -b. \notag 
	\end{align}
	Substituting these co-energies into the dynamic and kinematic relations derived in Steps 2 and 3 yields the governing equations $\dot{p} = -\mathcal{F}_\X(r)^*e_\q(\q) +b$, $\dot{\q} = \mathcal{F}_\X(r)e_p(p)$, and $\dot{r} = e_p(p)$. Consequently, the system recovers the geometric structure presented in Definition \ref{def:PHS_stokesDirac}. Lastly, through the specification of the power-conjugate boundary ports in accordance with Definition \ref{def:boundaryPorts}, this dynamic system  defines a PHS in Stokes-Dirac structure.
\end{proof}

\subsection{Proof of Proposition \ref{prop:lifted_StokesPHS2}}  \label{app:proof_Prop4}

\begin{proof}
	\textit{Step 1.} The functional $\mathcal{P}_{\HW}$ is evaluated by substituting the velocity decomposition $\dot{r} = \dot{r}_r + v_L$ and the displacement $r = r_r + r_L$. Since the lifting velocity $v_L$ is kinematically prescribed, its independent variation is zero ($\delta v_L = 0$). Expanding the Lagrangian rate density yields:
	\begin{equation}
		\dot{\mathfrak{L}}_{p\q r}(p,\q, \dot{p},\dot{\q}, \dot{r} = \dot{r}_r + v_L) = (\mathcal{M}^{-1}p)^\top \dot{p} - \left(\frac{\partial \Psi(\q)}{\partial \q}\right)^\top \dot{\q} + b^\top (\dot{r}_r + v_L). \notag
	\end{equation}
	\textit{Step 2.} Substituting this rate into the functional and setting the stationary condition $\delta \mathcal{P}_{\HW} = \delta_{\dot{p}} \mathcal{P}_{\HW} + \delta_{\dot{\q}} \mathcal{P}_{\HW} + \delta_{e_\q} \mathcal{P}_{\HW} + \delta_{\dot{r}_r} \mathcal{P}_{\HW} = 0$ provides four independent conditions. Grouping the variations with respect to $\delta\dot{p}$, $\delta\dot{\q}$, and $\delta e_\q$ yields:
	\begin{align}
		\delta_{\dot{p}} \mathcal{P}_{\HW} &= \int_\Omega \delta\dot{p}^\top \left( \dot{r}_r + v_L - \mathcal{M}^{-1}p \right) d\X = 0, \notag \\
		\delta_{\dot{\q}} \mathcal{P}_{\HW} &= \int_\Omega \delta\dot{\q}^\top \left( \frac{\partial \Psi(\q)}{\partial \q} - e_\q \right) d\X = 0, \notag \\
		\delta_{e_\q} \mathcal{P}_{\HW} &= \int_\Omega -\delta e_\q^\top \left( \dot{\q} - \mathcal{F}_\X(r_r + r_L)(\dot{r}_r + v_L) \right) d\X = 0. \notag
	\end{align}
	Since these variations are arbitrary, they globally enforce the kinetic relation $\dot{r}_r = \mathcal{M}^{-1}p - v_L$, the constitutive elastic effort $e_\q = \frac{\partial \Psi(\q)}{\partial \q}$, and the kinematic strain rate $\dot{\q} = \mathcal{F}_\X(r_r + r_L)(\dot{r}_r + v_L)$ in $\Omega$.

	\textit{Step 3.} Evaluating the variation with respect to the relative velocity $\delta\dot{r}_r$ yields:
	\begin{equation}
		\delta_{\dot{r}_r} \mathcal{P}_{\HW} = \int_\Omega \left[ \dot{p}^\top \delta\dot{r}_r - b^\top \delta\dot{r}_r + e_\q(\q)^\top \mathcal{F}_\X(r_r + r_L)\delta\dot{r}_r \right] d\X - \int_{\p\Omega_N} \!\!\! \tau_N^\top \delta\dot{r}_r \, d\SX = 0. \notag
	\end{equation}
	To factor out $\delta\dot{r}_r$, Lemma \ref{lemma:integration} is applied to the term involving the differential operator $\F_\X(r_r + r_L)$:
	\begin{equation}
		\int_\Omega e_\q(\q)^\top \mathcal{F}_\X(r_r + r_L)\delta\dot{r}_r \, d\X = \int_\Omega \big(\mathcal{F}_\X(r_r + r_L)^*e_\q(\q)\big)^\top \delta\dot{r}_r \, d\X + \int_{\p\Omega} \delta\dot{r}_r^\top F_\p(r_r + r_L) e_\q(\q) \, d\SX. \notag
	\end{equation}
	Substituting this back and grouping terms gives:
	\begin{equation}
		\int_\Omega \left[ \dot{p} + \mathcal{F}_\X(r_r + r_L)^*e_\q(\q) - b \right]^\top \delta\dot{r}_r \, d\X - \int_{\p\Omega_N} \!\!\! \left[ \tau_N - F_\p(r_r + r_L) e_\q(\q) \right]^\top \delta\dot{r}_r \, d\SX + \int_{\p\Omega_D} \!\!\! \big(F_\p(r_r + r_L) e_\q(\q)\big)^\top \delta\dot{r}_r \, d\SX = 0. \notag
	\end{equation}
	Since $\delta\dot{r}_r \in \mathcal{V}_0 \implies \delta\dot{r}_r|_{\p\Omega_D} = 0$, causing the final boundary integral to vanish. For arbitrary $\delta\dot{r}_r$ in $\Omega$ and on $\p\Omega_N$, we deduce the dynamic balance $\dot{p} = -\mathcal{F}_\X(r_r + r_L)^*e_\q(\q) + b$ in $\Omega$ and the Neumann boundary condition $\tau_N = F_\p(r_r + r_L) e_\q(\q)$ on $\p\Omega_N$.

	\textit{Step 4.} The augmented state vector is defined as $x = [p^\top \; \q^\top \; r_r^\top \; r_L^\top]^\top$. The co-energy variables $\var_x H = [e_p^\top \; e_\q^\top \; e_{r_r}^\top \; e_{r_L}^\top]^\top$ are formally defined through the first variation of the Hamiltonian functional \eqref{eq:HAMILTONIAN_PHS2_Lifted} in the direction of the arbitrary variations $\delta x = [\delta p^\top \; \delta\q^\top \; \delta r_r^\top \; \delta r_L^\top]^\top$:
	\begin{equation}
		\delta H = \int_\Omega \left( \delta p^\top \delta_p H + \delta\q^\top \delta_\q H + \delta r_r^\top \delta_{r_r} H + \delta r_L^\top \delta_{r_L} H \right) d\X. \notag
	\end{equation}
	Unlike the jet-bundle formulation, the Stokes-Dirac Hamiltonian is purely algebraic. Computing its first variation directly yields:
	\begin{equation}
		\delta H = \int_{\Omega} \left[ \delta p^\top (\mathcal{M}^{-1}p) + \delta\q^\top \left(\frac{\partial \Psi(\q)}{\partial \q}\right) - (\delta r_r + \delta r_L)^\top b \right] d\X. \notag
	\end{equation}
	Matching the terms with the formal definition, the co-energies are  identified as:
	\begin{align}
		e_p(p) = \delta_p H = \mathcal{M}^{-1}p, \notag , & &
		e_\q(\q) = \delta_\q H = \frac{\partial \Psi(\q)}{\partial \q}, \notag & &
		e_{r_r}(r) = \delta_{r_r} H = -b, \notag  & &
		e_{r_L}(r) = \delta_{r_L} H = -b. \notag
	\end{align}
	Substituting these co-energies into the dynamic and kinematic relations derived in Steps 2 and 3 yields the governing equations $\dot{p} = -\mathcal{F}_\X(r_r + r_L)^*e_\q(\q) +b$, $\dot{\q} = \mathcal{F}_\X(r_r+ r_L)e_p(p)$, and $\dot{r}_r = e_p(p) - v_L$. In addition to $\dot{r}_L = v_L$, the system recovers the geometric structure presented in Proposition \ref{prop:lifted_StokesPHS2}.
	
	\textit{Step 5.} The boundary ports are defined by computing the time derivative of $H(x)$:
	\begin{equation}
		\dot{{H}} = \int_\Omega \left( e_p(p)^\top \dot{p} + e_\q(\q)^\top \dot{\q} - b^\top \dot{r}_r - b^\top \dot{r}_L \right) d\X. \notag
	\end{equation}
	Since $\dot{r}_r + \dot{r}_L = \dot{r}_r + v_L = e_p(p)$, the gravity terms simplify to $-b^\top e_p(p)$. Substituting $\dot{\q} = \mathcal{F}_\X(r_r+r_L)e_p(p)$ and $\dot{p} = -\mathcal{F}_\X(r_r+r_L)^*e_\q(\q) + b$ yields:
	\begin{align*}
		\dot{{H}} =&  \int_\Omega \left[ e_p(p)^\top \big(-\mathcal{F}_\X(r_r+r_L)^*e_\q(\q) + b\big) + e_\q(\q)^\top \mathcal{F}_\X(r_r+r_L)e_p(p) - b^\top e_p(p) \right] d\X  \notag \\
		= & \int_\Omega \left[ e_\q(\q)^\top \mathcal{F}_\X(r_r+r_L)e_p(p) - e_p(p)^\top \mathcal{F}_\X(r_r+r_L)^*e_\q(\q) \right]  d\X \notag \\
		= & \int_{\p\Omega} e_p(p)^\top F_\p(r_r+r_L) e_\q(\q) \, d\SX, \notag
	\end{align*}
	where Lemma \ref{lemma:integration} was applied. Recognizing that $e_p(p)|_{\p\Omega_D} = (\dot{r}_r + v_L)|_{\p\Omega_D} = v_D(\SX,t)$ since  $\dot{r}_r|_{\p\Omega_D} = 0$ and $v_L|_{\p\Omega_D} = v_D(\SX,t)$, the power balance is equivalently written as:
	\begin{equation}
		\dot{{H}} = \int_{\p\Omega_N} v_N^\top \tau_N \, d\SX + \int_{\p\Omega_D} v_D^\top \tau_D \, d\SX, \notag
	\end{equation}
	with the given boundary inputs and outputs ports.
\end{proof}

\section{Proofs of Section \ref{sec:FEM}}

\subsection{Proof of Theorem \ref{theo:discrete_lifted_jet}} \label{app:Theo1}

\begin{proof}
	\textit{Step 1.} We derive the total discrete Hamiltonian and its gradients. Substituting the spatial Galerkin approximations into the continuous functional \eqref{eq:HAMILTONIAN_PHS1_Lifted}, and using $N_{r_\Omega}^e = N_r^e$, the local displacement simplifies to $\tilde{r}^e = N_r^e (\hat{r}_r^e + \hat{r}_\Omega^e) + N_{r_D}^e \hat{r}_D^e$. Evaluating the integrals yields the element-level discrete energy:
	\begin{equation}
		\hat{H}^e(\hat{z}^e) = \frac{1}{2} (\hat{p}^e)^\top \!\! \left[ \int_{\Omega^e} (N_p^e)^\top \mathcal{M}^{-1} N_p^e \, d\X \right] \! \hat{p}^e + \int_{\Omega^e} \Psi(\q(\tilde{r}^e)) \, d\X - (\hat{r}_r^e + \hat{r}_\Omega^e)^\top \!\! \int_{\Omega^e} (N_r^e)^\top b \, d\X - (\hat{r}_D^e)^\top \!\! \int_{\Omega^e} (N_{r_D}^e)^\top b \, d\X. \notag
	\end{equation}
	Summing over all elements $\hat{H}(\hat{z}) = \sum_{e=1}^{n_e} \hat{H}^e(\hat{z}^e)$ via the boolean location matrices defines the global matrices $\hat{M}_{pp}$, $\hat{b}_r$, and $\hat{b}_D$, thereby recovering \eqref{eq:discrete_hamiltonian}. 
	The co-energy variables are obtained through the gradient $\nabla_{\hat{z}} \hat{H}(\hat{z})$. For the momentum, this yields $\nabla_{\hat{p}} \hat{H}(\hat{z}) = \hat{M}_{pp} \hat{p} = \hat{e}_p(\hat{p})$. For the relative displacement $\hat{r}_r$, applying the chain rule gives:
	\begin{equation}
		\nabla_{\hat{r}_r} \hat{H}(\hat{z}) = \sum_{e=1}^{n_e} (L_r^e)^\top \left[ \int_{\Omega^e} \left( \frac{\p \q(\tilde{r}^e)}{\p \hat{r}_r^e} \right)^\top \frac{\p \Psi}{\p \q}\bigg|_{\q(\tilde{r}^e)} \, d\X - \int_{\Omega^e} (N_r^e)^\top b \, d\X \right]. \notag
	\end{equation}
	Recognizing the kinematic differential $\frac{\p \q}{\p \hat{r}_r^e} = {\F}_\X(\tilde{r}^e) N_r^e$ and the constitutive law $e_\q(\tilde{r}^e) = \frac{\p \Psi}{\p \q}$, the gradient evaluates to $\nabla_{\hat{r}_r} \hat{H}(\hat{z}) = \hat{F}_{e_\q r}(\hat{r}) - \hat{b}_r = \hat{e}_{r_r}(\hat{r})$. Because $\hat{r}_r$ and $\hat{r}_\Omega$ multiply the same shape functions $N_r^e$ within $\tilde{r}^e$, their partial derivatives are identical, proving $\hat{e}_{r_r}(\hat{r}) = \hat{e}_{r_\Omega}(\hat{r})$. The same procedure for $\hat{r}_D$ yields $\hat{e}_{r_D}(\hat{r})$.
	
	\textit{Step 2.} We substitute the mappings $\delta\dot{\hat{p}}^e = L_p^e \delta\dot{\hat{p}}$, $\dot{\hat{r}}_r^e = L_r^e \dot{\hat{r}}_r$, $\hat{v}_D^e = L_{r_D}^e \hat{v}_D$, $\hat{v}_\Omega^e = L_{r_\Omega}^e \hat{v}_\Omega$, and $\hat{p}^e = L_p^e \hat{p}$ into \eqref{eq:weak_p_jet}. Factoring out the common arbitrary global variation $\delta\dot{\hat{p}}^\top$ gives:
	\begin{equation}
		\delta\dot{\hat{p}}^\top \left[ \sum_{e=1}^{n_e} (L_p^e)^\top \!\! \int_{\Omega^e} \!\! (N_p^e)^\top \left( N_r^e \dot{\hat{r}}_r + N_{r_D}^e \hat{v}_D + N_{r_\Omega}^e \hat{v}_\Omega - \mathcal{M}^{-1} N_p^e \hat{p} \right) d\X \right] = 0. \notag
	\end{equation}
	Since $N_{r_\Omega}^e = N_r^e$ and $L_{r_\Omega}^e = L_r^e$, the integral associated with $\hat{v}_\Omega$ forms the global matrix $\hat{M}_{pr}$. Using the definitions from Theorem \ref{theo:discrete_lifted_jet}, the discrete equation becomes:
	\begin{equation}
		\delta\dot{\hat{p}}^\top \left[\hat{M}_{pr} \dot{\hat{r}}_r + \hat{M}_{pD} \hat{v}_D + \hat{M}_{pr} \hat{v}_\Omega - \hat{e}_p(\hat{p}) \right] = 0 \quad \implies \quad  \dot{\hat{r}}_r = \hat{M}_{pr}^{-1}\hat{e}_p(\hat{p}) - \hat{M}_{pr}^{-1}\hat{M}_{pD} \hat{v}_D - \hat{v}_\Omega. \label{eq:proof_r_dot_jet}
	\end{equation}
	Because $N_p^e(\X) = A_p^e N_r^e(\X)$ with $A_p^e = (A_p^e)^\top > 0$, the matrix $\hat{M}_{pr}$ is symmetric and positive definite, thus invertible.
	
	\textit{Step 3.} To resolve the dynamic momentum balance, Lemma \ref{lemma:integration} is applied to the local stationary condition \eqref{eq:weak_r_jet} to transfer the differential operator to the kinematically admissible test function, yielding the element-level weak form:
	\begin{equation}
		\int_{\Omega^e} \left[ \delta\dot{r}_r^{e\top} \dot{p}^e + \big(\mathcal{F}_\X(r^e)\delta\dot{r}_r^e\big)^\top e_\q(\q(r^e)) - \delta\dot{r}_r^{e\top} b \right] d\X - \int_{\p\Omega^e_N} \delta\dot{r}_r^{e\top} \tau_N^e \, d\SX = 0. \notag
	\end{equation}
	Substituting the mappings $\delta\dot{\hat{r}}_r^e = L_r^e \delta\dot{\hat{r}}_r$, $\dot{\hat{p}}^e = L_p^e \dot{\hat{p}}$, $\hat{\tau}_N^e = L_{\tau_N}^e \hat{\tau}_N$, assembling over $n_e$, and factoring out $\delta\dot{\hat{r}}_r^\top$ yields:
	\begin{equation}
		\delta\dot{\hat{r}}_r^\top \left[ \hat{M}_{pr}^\top \dot{\hat{p}} + \hat{F}_{e_\q r}(\hat{r}) - \hat{b}_r - \hat{B}_N \hat{\tau}_N \right] = 0 \quad\implies\quad  \dot{\hat{p}} = -\hat{M}_{pr}^{-\top}\hat{e}_{r_r}(\hat{r}) + \hat{M}_{pr}^{-\top}\hat{B}_N \hat{\tau}_N. \label{eq:proof_p_dot_jet}
	\end{equation}
	\textit{Step 4.} The lifting kinematics impose strongly $\dot{\hat{r}}_D = \hat{v}_D$ and $\dot{\hat{r}}_\Omega = \hat{v}_\Omega$. Grouping these expressions with equations \eqref{eq:proof_r_dot_jet} and \eqref{eq:proof_p_dot_jet} recovers the PHS structure defined in \eqref{eq:theo_discrete_jet_dyn}. The output ports are computed by evaluating $\hat{y} = \hat{G}^\top \nabla_{\hat{z}} \hat{H} = [\hat{y}_\p^\top \;\; \hat{y}_\Omega^\top]^\top$. The third block evaluates the distributed output:
	\begin{equation}
		\hat{y}_\Omega = -I_{N_F}^\top \hat{e}_{r_r}(\hat{r}) + I_{N_F}^\top \hat{e}_{r_\Omega}(\hat{r}). \notag
	\end{equation}
	Since $\hat{e}_{r_r}(\hat{r}) = \hat{e}_{r_\Omega}(\hat{r})$ as proved in Step 1, this expression  reduces to $\hat{y}_\Omega = 0$, validating the exact conservation of energy neutrality at the discrete level. Finally, computing the time derivative of the Hamiltonian yields $\dot{\hat{H}} = \nabla_{\hat{z}} \hat{H}^\top \dot{\hat{z}} = \nabla_{\hat{z}} \hat{H}^\top [ \hat{J} \nabla_{\hat{z}} \hat{H} + \hat{G} \hat{u} ]$. Since $\hat{J} = -\hat{J}^\top$, the internal power dissipation term $\nabla_{\hat{z}} \hat{H}^\top \hat{J} \nabla_{\hat{z}} \hat{H} = 0$. Recognizing that $\hat{y}^\top = \nabla_{\hat{z}} \hat{H}^\top \hat{G}$, the global power balance evaluates to $\dot{\hat{H}} = \hat{y}^\top \hat{u} = \hat{y}_\p^\top \hat{u}_\p$. This concludes the proof, demonstrating that the spatial discretization preserves the underlying PHS structure and its exact boundary power flow, while achieving the strong imposition of Dirichlet boundary velocities $\dot{\hat{r}}_D = \hat{v}_D$.
	\end{proof}

\subsection{Proof of Theorem \ref{theo:discrete_lifted_stokes}} \label{app:Theo2}

\begin{proof}
	\textit{Step 1.} We derive the discrete Hamiltonian and its gradients. Substituting the Galerkin approximations into \eqref{eq:HAMILTONIAN_PHS2_Lifted} and using $N_{r_\Omega}^e = N_r^e$, the element displacement is $\tilde{r}^e = N_r^e (\hat{r}_r^e + \hat{r}_\Omega^e) + N_{r_D}^e \hat{r}_D^e$. The potential energy depends only on $\tilde{\q}^e$. The element-level discrete energy is:
	\begin{equation}
		\hat{H}^e(\hat{x}^e) = \frac{1}{2} (\hat{p}^e)^\top \!\! \left[ \int_{\Omega^e} (N_p^e)^\top \mathcal{M}^{-1} N_p^e \, d\X \right] \! \hat{p}^e + \int_{\Omega^e} \Psi(\tilde{\q}^e) \, d\X - (\hat{r}_r^e + \hat{r}_\Omega^e)^\top \!\! \int_{\Omega^e} (N_r^e)^\top b \, d\X - (\hat{r}_D^e)^\top \!\! \int_{\Omega^e} (N_{r_D}^e)^\top b \, d\X. \notag
	\end{equation}
	Summing over all elements, $\hat{H}(\hat{x}) = \sum_{e=1}^{n_e} \hat{H}^e(\hat{x}^e)$, and applying the boolean mappings $\hat{p}^e = L_p^e \hat{p}$, $\hat{r}_r^e = L_r^e \hat{r}_r$, $\hat{r}_\Omega^e = L_{r_\Omega}^e \hat{r}_\Omega = L_r^e \hat{r}_\Omega$, and $\hat{r}_D^e = L_{r_D}^e \hat{r}_D$ yields the discrete Hamiltonian in \eqref{eq:discrete_hamiltonian_stokes} with  $\hat{M}_{pp}$, $\hat{b}_r$, and $\hat{b}_D$ as defined. Taking partial derivatives yields the co-energy variables. For momentum, $\nabla_{\hat{p}} \hat{H}(\hat{x}) = \hat{M}_{pp} \hat{p} = \hat{e}_p(\hat{p})$. For the strain field, applying the chain rule:
	\begin{equation}
		\nabla_{\hat{\q}} \hat{H}(\hat{x}) = \sum_{e=1}^{n_e} \left( \frac{\p \hat{\q}^e}{\p \hat{\q}} \right)^\top \int_{\Omega^e} \left( \frac{\p \tilde{\q}^e}{\p \hat{\q}^e} \right)^\top \frac{\p \Psi}{\p \q}\bigg|_{\q(\tilde{\q}^e)} d\X. \notag
	\end{equation}
	Using $\frac{\p \hat{\q}^e}{\p \hat{\q}} = L_\q^e$ and $\frac{\p \tilde{\q}^e}{\p \hat{\q}^e} = N_\q^e$, this gives $\nabla_{\hat{\q}} \hat{H}(\hat{x}) = \check{e}_\q(\hat{\q})$. For the displacements, deriving the linear potential terms yields $\nabla_{\hat{r}_r} \hat{H}(\hat{x}) = -\hat{b}_r$, $\nabla_{\hat{r}_\Omega} \hat{H}(\hat{x}) = -\hat{b}_r$, and $\nabla_{\hat{r}_D} \hat{H}(\hat{x}) = -\hat{b}_D$.
	
	\textit{Step 2.} Substituting the local approximations, together with $\delta\dot{\hat{\q}}^e = L_\q^e \delta\dot{\hat{\q}}$ and $\hat{e}_\q^e = L_{e_\q}^e \hat{e}_\q$ into \eqref{eq:weak_q_stokes} yields:
	\begin{equation}
		\delta\dot{\hat{\q}}^\top \left[ \sum_{e=1}^{n_e} (L_\q^e)^\top \int_{\Omega^e} (N_\q^e)^\top N_{e_\q}^e \, d\X \, L_{e_\q}^e \, \hat{e}_\q - \sum_{e=1}^{n_e} (L_\q^e)^\top \int_{\Omega^e} (N_\q^e)^\top \frac{\p \Psi}{\p \q}\bigg|_{\q(\tilde{\q}^e)} d\X \right] = 0. \notag
	\end{equation}
	Recognizing the global definitions, this evaluates to $\delta\dot{\hat{\q}}^\top [ \hat{M}_{e\q}^\top \hat{e}_\q - \check{e}_\q(\hat{\q}) ] = 0$, implying $\hat{e}_\q = \hat{M}_{e\q}^{-\top} \check{e}_\q(\hat{\q})$. The condition $N_\q^e(\X) = A_\q^e N_{e_\q}^e(\X)$ with $A_\q^e = (A_\q^e)^\top>0$ guarantees $\hat{M}_{e\q}$ symmetric and invertible. Substituting $\delta\dot{\hat{p}}^e = L_p^e \delta\dot{\hat{p}}$ into \eqref{eq:weak_p_stokes} recovers the velocity kinematics from Theorem \ref{theo:discrete_lifted_jet}:
	\begin{equation}
		\dot{\hat{r}}_r = \hat{M}_{pr}^{-1}\hat{e}_p(\hat{p}) - \hat{M}_{pr}^{-1}\hat{M}_{pD} \hat{v}_D - \hat{v}_\Omega. \label{eq:proof_r_dot_stokes}
	\end{equation}
	Substituting $\delta\hat{e}_\q^e = L_{e_\q}^e \delta\hat{e}_\q$ into \eqref{eq:weak_e_stokes} and using $N_r^e = N_{r_\Omega}^e$ gives:
	\begin{equation}
		\delta\hat{e}_\q^\top \left[ \hat{M}_{e\q} \dot{\hat{\q}} - \hat{F}_r(\hat{r}) \dot{\hat{r}}_r - \hat{F}_D(\hat{r}) \hat{v}_D - \hat{F}_r(\hat{r}) \hat{v}_\Omega \right] = 0 \quad \implies \quad \dot{\hat{\q}} = \hat{M}_{e\q}^{-1} \big[ \hat{F}_r(\hat{r}) \dot{\hat{r}}_r + \hat{F}_D(\hat{r}) \hat{v}_D + \hat{F}_r(\hat{r}) \hat{v}_\Omega \big]. \notag
	\end{equation}
	Substituting \eqref{eq:proof_r_dot_stokes} into this expression yields:
	\begin{align}
		\dot{\hat{\q}} = & \; \hat{M}_{e\q}^{-1} \hat{F}_r(\hat{r}) \hat{M}_{pr}^{-1}\hat{e}_p(\hat{p}) + \hat{M}_{e\q}^{-1} \big( \hat{F}_D(\hat{r}) - \hat{F}_r(\hat{r}) \hat{M}_{pr}^{-1} \hat{M}_{pD} \big) \hat{v}_D - \hat{M}_{e\q}^{-1} \hat{F}_r(\hat{r}) \hat{v}_\Omega + \hat{M}_{e\q}^{-1} \hat{F}_r(\hat{r}) \hat{v}_\Omega  \notag \\[2mm]
		= & \; \hat{M}_{e\q}^{-1} \hat{F}_r(\hat{r}) \hat{M}_{pr}^{-1}\hat{e}_p(\hat{p}) + \hat{M}_{e\q}^{-1} \big( \hat{F}_D(\hat{r}) - \hat{F}_r(\hat{r}) \hat{M}_{pr}^{-1} \hat{M}_{pD} \big) \hat{v}_D. \label{eq:proof_q_dot_stokes}
	\end{align}
	\textit{Step 3.} To resolve the momentum balance, Lemma \ref{lemma:integration} is applied to \eqref{eq:weak_r_stokes} to transfer the differential operator, yielding the element-level weak form:
	\begin{equation}
		\int_{\Omega^e} \left[ \delta\dot{r}_r^{e\top} \dot{p}^e + \big(\mathcal{F}_\X(r^e)\delta\dot{r}_r^e\big)^\top e_\q^e - \delta\dot{r}_r^{e\top} b \right] d\X - \int_{\p\Omega^e_N} \delta\dot{r}_r^{e\top} \tau_N^e \, d\SX = 0. \notag
	\end{equation}
	Substituting the local to global mappings, assembling, and factoring out $\delta\dot{\hat{r}}_r^\top$ gives:
	$ \delta\dot{\hat{r}}_r^\top [ \hat{M}_{pr}^\top \dot{\hat{p}} + \hat{F}_{r}(\hat{r})^\top \hat{e}_\q - \hat{b}_r - \hat{B}_N \hat{\tau}_N ] = 0 $. Replacing $\hat{e}_\q = \hat{M}_{e\q}^{-\top} \check{e}_\q(\hat{\q})$ and isolating the momentum rate:
	\begin{equation}
		\dot{\hat{p}} = -\hat{M}_{pr}^{-\top}\hat{F}_{r}(\hat{r})^\top \hat{M}_{e\q}^{-\top} \check{e}_\q(\hat{\q}) + \hat{M}_{pr}^{-\top}\hat{b}_r + \hat{M}_{pr}^{-\top}\hat{B}_N \hat{\tau}_N. \label{eq:proof_p_dot_stokes}
	\end{equation}
	\textit{Step 4.} The lifting kinematics impose strongly $\dot{\hat{r}}_D = \hat{v}_D$ and $\dot{\hat{r}}_\Omega = \hat{v}_\Omega$. Grouping these expressions with equations \eqref{eq:proof_r_dot_stokes}, \eqref{eq:proof_q_dot_stokes}, and \eqref{eq:proof_p_dot_stokes} gives the PHS structure in \eqref{eq:theo_discrete_stokes_dyn}. The output ports are $\hat{y} = \hat{G}(\hat{x})^\top \nabla_{\hat{x}} \hat{H}(\hat{x}) = [\hat{y}_\p^\top \;\; \hat{y}_\Omega^\top]^\top$. The distributed interior output evaluates to: \\[-1mm]
	\begin{equation}
		\hat{y}_\Omega = -I_{N_F}^\top \nabla_{\hat{r}_r} \hat{H}(\hat{x}) + I_{N_F}^\top \nabla_{\hat{r}_\Omega} \hat{H}(\hat{x}) = -(-\hat{b}_r) + (-\hat{b}_r) = 0. \notag
	\end{equation}
	Then, the global power balance evaluates to $\dot{\hat{H}} = \hat{y}^\top \hat{u} = \hat{y}_\p^\top \hat{u}_\p$. This concludes the proof, showing the spatial discretization preserves the Stokes-Dirac PHS structure and its boundary power flow, while strongly imposing Dirichlet boundary velocities $\dot{\hat{r}}_D = \hat{v}_D$.
\end{proof}

\subsection{Proof of Proposition \ref{prop:kinematicEQ}} \label{app:prop7}

\begin{proof}
	Substituting the local approximations, together with $\delta\dot{\hat{\q}}^e = L_\q^e \delta\dot{\hat{\q}}$ and $\hat{e}_\q^e = L_{e_\q}^e \hat{e}_\q$ into \eqref{eq:weak_Work_GHP} yields: 
	\begin{equation}
		\delta \hat{e}_\q^\top \bigg[ \ub{\sum_{e=1}^{n_e}  (L_{e_\q}^e)^\top \!\!\! \int_{\Omega^e} (N_{e_\q}^e)^\top N_{\q}^e \, d\X \, L_{\q}^e }{\hat{M}_{e\q}} \hat{\q}(t) - \sum_{e=1}^{n_e} (L_{e_\q}^e)^\top \!\!\! \int_{\Omega^e} (N_{e_\q}^e)^\top \q(\tilde{r}^e) \, d\X \bigg] = 0. \notag
	\end{equation}
	Since $\delta \hat{e}_\q^\top $ is arbitrary, the above expression implies:
	$
	\hat{\q}(t) = \hat{M}_{e\q}^{-1} \sum_{e=1}^{n_e} (L_{e_\q}^e)^\top \! \int_{\Omega^e} (N_{e_\q}^e)^\top \q(\tilde{r}^e) \,d\X = \check{\q}(\hat{r}).
	$
\end{proof}



\bibliographystyle{ieeetr}

\bibliography{autosam}               

\end{document}